%% file: lr.tex
\def\OPTIONAnonymize{1}%
\def\OPTIONArxiv{0}%
\def\OPTIONAppendix{1}%
\def\OPTIONConf{2}%
\ifnum\OPTIONConf=1%
  \ifnum\OPTIONArxiv=0
     \documentclass[preprint,nonatbib]{sigplanconf}
  \else
         \documentclass[preprint,nonatbib]{sigplanconf}
  \fi%
  \usepackage[breaklinks]{hyperref}
  \usepackage{breakurl}
  \usepackage{jdunfield}
  \usepackage{euler}
\fi

\ifnum\OPTIONConf=2%
   \documentclass[a4paper]{llncs}
   \usepackage[breaklinks]{hyperref}
   \usepackage{breakurl}
   \usepackage{jdunfield}
   \usepackage{amssymb}

   \usepackage{amsmath}
   \usepackage{amsthm}
   \usepackage{thmtools}
   \usepackage{thm-restate}
   \renewcommand{\bfdefault}{b}
   \let\vec\undefined
   \usepackage{euler}
\usepackage{etoolbox}
\fi%

\ifnum\OPTIONConf=0%
  \documentclass{purple}
  \usepackage{jdunfield}
\fi

\usepackage{pst-node,graphicx,eepic}

\ifnum\OPTIONConf=0%
    \usepackage{fullpage}
\fi

\ifnum\OPTIONConf=1%
\fi

\usepackage[overload]{textcase}

\usepackage{listings}
\usepackage[authoryear]{natbib}
\bibpunct{(}{)}{;}{a}{}{,}

\usepackage{pgf,tikz}
\usetikzlibrary{shapes,arrows,matrix,fit,backgrounds,calc,decorations,decorations.pathmorphing}

\ifnum\OPTIONConf=0%
\fi
\ifnum\OPTIONConf=1%
\fi
\ifnum\OPTIONConf=2%
\fi
\renewcommand{\ttdefault}{cmtt}   %

\ifnum\OPTIONConf=2%
  \usepackage{latexsym}
\else
  \usepackage{amsmath,amssymb,amsthm}
  \usepackage{thmtools,thm-restate}
\fi

\usepackage{multirow}
\usepackage{stmaryrd}
\usepackage{pifont}
\usepackage{enumerate}
\usepackage{comment}
\usepackage{mathpartir} \usepackage{proof}
\usepackage{framed}
\usepackage{rotating}
\usepackage{relsize}

\let\Plus\undefined   %
\let\MathRightArrow\Rightarrow %
\usepackage{marvosym} %
\def\Rightarrow{\MathRightArrow}

\usepackage{srcltx}

\usepackage{rulelinks}
\usepackage{llproof}

\input kontroll
\def\OPTIONLoudLabels{0}
\input local

\def\MathparLineskip{\lineskiplimit=0.9em\lineskip=0.6em plus 0.2em}

\title{%
\ifnum\OPTIONConf=0
  \normalfont\bfseries\sffamily\selectfont
\fi
  Extensible Datasort Refinements%
}

\makeatletter
\renewenvironment{proof}[1][\proofname]{\par
  \vspace{-1.3ex}
  \pushQED{\qed}%
  \normalfont
  \trivlist
  \item[\hskip\labelsep
        \itshape
    #1\@addpunct{.}]\ignorespaces
}{%
  \popQED\endtrivlist\@endpefalse
  \addvspace{2pt plus 2pt} %
}
\makeatother

\ifnum\OPTIONArxiv=0%
  \else
   \fi
\ifnum\OPTIONConf=1
  \ifnum\OPTIONAnonymize=1
    \authorinfo{}{}{}%
  \else
    \authorinfo{Jana Dunfield}
       {University of British Columbia \\ Vancouver, Canada}
       {} %
  \fi
\fi
\ifnum\OPTIONConf=2
  \author{Jana Dunfield}
  \institute{University of British Columbia \\ Vancouver, Canada}
\fi

\begin{document}
\maketitle

\input{abstract}

\setcounter{footnote}{0}

\section{Introduction}
\Label{sec:intro}

Type systems provide guarantees about run-time behaviour;
for example, that a record will not be multiplied by a string.
However, the guarantees provided by traditional type systems
like Hindley--Milner do not include the absence of a practically
important class of run-time failures: nonexhaustive match exceptions.
For example, the type system of Standard ML allows
a case expression over lists that omits a branch for the empty list:
\begin{verbatim}
  case elems of head :: tail => ...
\end{verbatim}
If this expression is evaluated with \texttt{elems} bound to the empty list \texttt{[]},
the exception \textvtt{Match} will be raised.

Datasort refinements eliminate this problem: a datasort can express, within the
static type system, that \texttt{elems} is not empty; therefore, the case expression
is safe.  Datasorts can also express less shallow properties.
For example, the definition in \Figureref{fig:cnf} encodes
conjunctive normal form---a formula that consists of (possibly nested) $\CAnd$s
of clauses, where a clause consists of (possibly nested) $\Or$s of literals,
where a literal is either a positive literal (a variable) or a negation of a positive literal.
A case expression comparing two values of type $\Clause$ would
only need branches for $\Or$, $\Neg$ and $\Var$; the $\CAnd$ branch
could be omitted, since $\CAnd$ does not produce a $\Clause$.

\begin{figure}[t]
  \centering

  \begin{rclll}
        \PosLiteral &\subsort& \Literal,
        ~~ \Literal ~\subsort~ \Clause, 
        ~~ \Clause ~\subsort~ \CNF,
        \\
        \Var &:& \tyname{symbol} \arr \PosLiteral,
        \\
        \Neg &:& \PosLiteral \arr \Literal,
        \\
        \Or &:& (\Clause * \Clause) \arr \Clause,
        \\
        \CAnd &:& (\CNF * \CNF) \arr \CNF
  \end{rclll}

  \lesscaptionspace
  \caption{Datasorts for conjunctive normal form}
  \label{fig:cnf}
\end{figure}

Datasorts correspond to regular tree grammars, which can encode various
data structure invariants (such as the colour invariant of red-black trees),
as well as properties such as CNF and A-normal form.
Datasort refinements are less expressive than
the ``refinement type'' systems (such as liquid types) that followed
work on index refinements and indexed types;
like regular expressions, which ``can't count'', datasorts cannot ``count''
the length of a list or the height of a tree.
However, types with datasorts are simpler in some respects;
most importantly, types with datasorts never require quantifiers.
Avoiding quantifiers, especially existential quantifiers,
also avoids many complications in the type checker.
By analogy, regular expressions cannot solve every problem---but when
they \emph{can} solve the problem, they may be the best solution.

The goal of this paper is to make datasort refinements more usable---not by
making datasorts express more invariants, but by liberating them from
the necessity of a fixed specification (a fixed signature).
First, we review the trajectory of research on datasorts.

The first approach to datasort refinements~\citep{Freeman91,FreemanThesis}
extended ML, using
abstract interpretation~\citep{Cousot77} to infer refined types.  The usual argument in favour
of type inference is that it reduces a direct burden on the programmer.  When type
annotations are boring or self-evident, as they often are in plain ML,
this argument is plausible.  But datasorts can express more subtle specifications,
calling that argument into question.
Moreover, inference defeats a form of fine-grained modularity.
Just as we expect a module system to support information hiding, so that clients
of a module cannot depend on its internal details, a type system should prevent
the callers of a function from depending on its internal details.
Inferring refinements exposes those details.
For example, if a function over lists is written with only nonempty input in mind,
the programmer may not have thought about what the function should do for empty input,
so the type system shouldn't let the function be applied to an empty list.
Finally, inferring all properties means that the inferred refined types can
be long, \eg inferring a 16-part intersection type for a simple function
\citep[p.\ 271]{Freeman91}.

Thus, the second generation of work on datasort refinements \citep{Davies00icfpIntersectionEffects,DaviesThesis} used bidirectional typing, rather than inference.
Programmers have to write more annotations,
but refinement checking will never fabricate unintended
invariants.
A third generation of work \citep{Dunfield04:Tridirectional,DunfieldThesis}
stuck with bidirectional type checking, though this was overdetermined:
other features of their type system made inference untenable.

All three generations (and later work by \citet{LovasThesis} on datasorts for LF)
shared the constraint that a given datatype could be refined only once.
The properties tracked by datasorts could not be subsequently extended;
the same set of properties must be used throughout the program.  Modular
refinement checking could be achieved only by duplicating the type definition and
all related code.  Separate type-checking of refinements enables simpler reasoning
about programs, separate compilation, and faster type-checking (simpler refinement
relations lead to simpler case analyses).

The history of pattern typing (typing for case expressions) is also worth noting, as
formulating pattern typing seems to be the most difficult step in the design of datasort
type systems.
Freeman supported a form of pattern matching that was oversimplified. %
Davies implemented the full SML pattern language
and formalized most of it, but omitted \textkw{as}-patterns---which become
nontrivial when datasort refinements are in the picture.

In this paper, we allow multiple, separately declared refinements of a type
by revising a fundamental mechanism of datasort refinements: the signature.
Refinements are traditionally described using a \emph{signature} that specifies---for the
entire program---which values of a datatype belong to which refinements.  For example,
the type system can track the parity of bitstrings using the following signature,
which says that $\Even$ and $\Odd$ are subsorts (subtypes) of the type $\Bits$
of bitstrings, the empty bitstring has even parity, appending a 1 flips the parity,
and appending a 0 preserves parity.
\begin{rclll}
      \Even &\subsort& \Bits,~~~
      \Odd ~\subsort~ \Bits,
      \\
      \datacon{Empty} &:& \Even,
      \\
      \datacon{One} &:& (\Even \arr \Odd) \sectty (\Odd \arr \Even),
      \\
      \datacon{Zero} &:& (\Even \arr \Even) \sectty (\Odd \arr \Odd)
\end{rclll}
The connective $\sectty$, read ``and'' or ``intersection'', denotes conjunction
of properties: adding a $\datacon{One}$
makes an even bitstring odd ($\Even \arr \Odd$),
\emph{and} makes an odd bitstring even ($\Odd \arr \Even$).
Thus, if $b$ is a bitstring known to have odd parity, then appending a 1 yields
a bitstring with even parity:
\begin{mathdispl}
   b : \Odd ~\entails~ \datacon{One}(b) : \Even
\end{mathdispl}
In some datasort refinement systems \citep{DunfieldThesis,LovasThesis}, 
the programmer specifies the refinements by writing a signature like the one above.
In the older systems of Freeman and Davies, the programmer
writes a regular tree grammar\footnote{A \emph{regular tree grammar} is like
a regular grammar (the class of grammars equivalent to regular expressions),
but over trees instead of strings~\citep{TATA-book}; the leftmost terminal symbol
in a production of a regular grammar corresponds to the symbol at the root of a tree.
},
from which the system infers a signature, including
the constructor types and the subsort relation:
\begin{rclll}
  \Even   &=&   \datacon{Empty}
                \matchor
                \datacon{Zero}(\Even)
                \matchor
                \datacon{One}(\Odd)
  \\
  \Odd   &=&
                  \datacon{Zero}(\Odd)
                  \matchor
                  \datacon{One}(\Even)
\end{rclll}

\newcommand{\nat}{\tyname{nat}}
\newcommand{\tainted}{\tyname{tainted}}
\newcommand{\untainted}{\tyname{untainted}}
In either design, the typing phase uses the same form of signature.  We use the
first design, where the programmer gives the signature directly.
Giving
the signature directly is more expressive, because it enables refinements
to carry information not present at run time.  For example, we can
refine natural numbers by $\tyname{Tainted}$ and $\tyname{Untainted}$:
\begin{rclll}
      \datacon{Z} &:& \nat, &~~ \datacon{S} : \nat \arr \nat,
      \\
      \tainted &\subsort& \multicolumn{2}{l}{%
        \!\!\nat,
        ~~
        \untainted ~\subsort~ \nat,}
      \\
     \datacon{Z} &:& \tainted, &~~ \datacon{S} : \tainted \arr \tainted,
     \\
     \datacon{Z} &:& \untainted, &~~ \datacon{S} : \untainted \arr \untainted
\end{rclll}
The sorts $\tainted$ and $\untainted$ have the same closed inhabitants,
but a program cannot directly create
an instance of $\untainted$ from an instance of $\tainted$:
\[
    x : \tainted \not\entails \datacon{S}(x) : \untainted
\]
Thus, the two sorts have different \emph{open} inhabitants.  This is analogous
to dimension typing, where an underlying value is just an integer or float, but
the type system tracks that the number is in (for example) metres~\citep{KennedyThesis}.

Being able to give the signature directly allows programmers to choose
between a variety of subsorting relationships.  For example, to
allow untainted data to be used where tainted data is expected, write
$\untainted \subsort \tainted$.  In effect, subsorting can be either \emph{structural}
(as the signatures generated from grammars) or \emph{nominal} (as in the example above).

In this paper, giving signatures directly is helpful: it
enables \emph{extension} of a signature without translating between
signatures and grammars.

\paragraph{Contributions.}
This paper makes the following contributions:

\begin{itemize}
\item A language and type system with \emph{extensible signatures} for datasort refinements (\Sectionref{sec:language}).
  Refinements are extended by \emph{blocks} that are checked to ensure
  that they do not weaken a sort's inversion principle, which would make typing unsound.

\item A new formulation of typing (\Sectionref{sec:pattern-typing})
  for case expressions.  This formulation is based on a notion of finding the
  intersection of a type with a pattern; it concisely models the interesting aspects of
  realistic ML-style patterns.

\item Type (datasort) preservation and progress for the type assignment system,
  stated in \Sectionref{sec:meta} and proved in
  Appendix \ref{sec:proofs}, %
  with respect to a standard call-by-value operational semantics (\Sectionref{sec:opsem}).

 \item
   A bidirectional type system (\Sectionref{sec:bidir}), 
   which directly yields an algorithm.
   We prove that this system is sound (given a bidirectional typing derivation,
   erasing annotations yields a type assignment derivation) and
   complete (given any type assignment derivation,
   annotations can be added to make bidirectional typing succeed).
\end{itemize}

The appendix, which includes definitions and proofs omitted for space reasons,
can be found at
\href{http://www.cs.queensu.ca/~jana/papers/extensible/}%
{\textvtt{http://www.cs.queensu.ca/$\sim$jana/papers/extensible/}}.

\section{Datasort Refinements}

\paragraph{What are datasort refinements?}
Datasort refinements are a syntactic discipline for enforcing invariants.
This is a play on Reynolds's definition of types as a ``syntactic discipline for enforcing
levels of abstraction''~\citep{Reynolds83}.  Datasorts allow programmers to conveniently
categorize inductive data, and operations on such data, more precisely than
in conventional type systems.

Indexed types and related systems (\eg liquid types and other ``refinement types'')
also serve that purpose,
but datasorts are highly syntactic, whereas indexed types
depend on the semantics of a constraint domain.  For example,
to check the safety of accessing the element at position $2k$ of a 0-based array of length $n$,
an indexed type system must check whether the proposition $2k < n$ is entailed
in the theory of integers (under some set of assumptions, \eg $0 \leq k \leq n/3$).
The truth of $2k < n$ depends on the semantics of arithmetic, whereas
membership in a datasort only depends on a head constructor and the datasorts of its
arguments.
Put roughly, datasorts express regular grammars, and indexed types express
grammars with more powerful side conditions.  (Unrestricted dependent types
can express arbitrarily precise side conditions.) %

\paragraph{Applications of datasort refinements.}
Datasorts are especially suited to applications of symbolic computing,
such as compilers and theorem provers.  Compilers usually work
with multiple internal languages, from abstract syntax 
through to intermediate languages.  These internal languages
may be decomposed into further variants: source ASTs with and without
syntactic sugar, A-normal form, and so on.
Similarly, theorem provers, SMT solvers, and related tools
transform formulas into various normal forms or sublanguages:
quantifier-free Boolean formulas, conjunctive normal form,
formulas with no free variables, etc.
Many such invariants can be expressed by regular tree grammars,
and hence by datasorts.

Our extensible refinements offer the ability to use
new refinements of a datatype when the need arises,
without the need to update a global refinement declaration.
For example, we could extend the types in \Figureref{fig:cnf},
in which $\Clause$ contains disjunctions of literals
and $\CNF$ contains conjunctions of clauses,
with a new sort for \emph{conjunctions} of literals:

  \begin{rclll}
        \multicolumn{3}{l}{\hspace*{-4ex}\text{[everything from \Figureref{fig:cnf}]}}
        ~
        \arrayenvl{
        \Literal ~\subsort~ \ConjLiteral,~
        \ConjLiteral \subsort \CNF,
        \\
        \CAnd ~:~ (\ConjLiteral * \ConjLiteral) \arr \ConjLiteral
        }
  \end{rclll}

\paragraph{What are datasort refinements \emph{not}?}
First, datasorts are not really types, at least not in the sense of Hindley--Milner type systems.
A function on bitstrings (\Sectionref{sec:intro}) has a best, or principal, type: $\Bits \arr \Bits$.
In contrast, such a function may have many refined types (sometimes called \emph{sorts}),
depending not only on the way the programmer chose to refine the $\Bits$ type, but on which
possible properties they wish to check.  The type, or sort, of a function is a tiny module interface.
In a conventional Hindley--Milner type system, there is a best interface (the principal type); with
datasorts, the ``best'' interface is---as with a module interface, which may reveal
different aspects of the module---\emph{the one the programmer thinks best}.
Maybe the programmer only cares that the function preserves odd parity, and annotates
it with $\Odd \arr \Odd$; the compiler will reject calls with $\Even$ bitstrings, even though
such a call would be conventionally well-typed.

To infer sorts, as in the original work of Freeman, is like assuming that all declarations in a module
should be exposed. (Tools that suggest possible invariants could be useful, just as a tool that suggests
possible module interfaces %
could be useful.  But such tools are not the focus of this paper.)

\section{A Type System with Extensible Refinements}
\Label{sec:language}

This section gives our language's syntax, introduces signatures, discusses the introduction
and elimination forms for datasorts, and presents the typing rules.
The details of typing pattern matching are in \Sectionref{sec:pattern-typing}.

\subsection{Syntax}

\input{fig-expressions.tex}

The syntax of expressions (\Figureref{fig:expressions}) includes
functions $\Lam{x} e$, function application $e_1\,e_2$, pairs $\Pair{e_1}{e_2}$,
constructors $c(e)$,
and case expressions.
Signatures are extended by 
$\declare{\Sigma}{e}$.  %

\input{fig-types.tex}

Types (\Figureref{fig:types}), written $A$ and $B$,
include unit ($\unitty$), function, and product types,
along with datasorts $s$ and $t$.
The intersection type $A \sectty B$ represents the conjunction
of the two properties denoted by $A$ and $B$; for example,
a function to repeat a bitstring could be checked against type
$(\Odd \arr \Even) \sectty (\Even \arr \Even)$: given
any bitstring $b$, the repetition $bb$ has even parity.

\subsection{Unrefined types and signatures}

Our unrefined types $\tau$, in \Figureref{fig:signatures}, are very simple: unit $\unitty$,
functions $\tau_1 \arr \tau_2$, products $\tau_1 * \tau_2$, and datatypes $d$.
We assume that each datatype has a known set of constructors: for example,
the bitstring type of \Sectionref{sec:intro} has constructors \datacon{Empty}, \datacon{One}
and \datacon{Zero}.  Refinements don't add constructors;
they only refine the types of the given constructors.  We assume
that each program has some \emph{unrefined signature} $\Ursig$ that gives
datatype names ($d$) and (unrefined) constructor typings ($c : \tau \arr d$).
Since this signature is the same throughout a program, we elide it in most judgment forms.

The judgment $\Sigma \entails A \refines \tau$ says that $A$ is a refinement of $\tau$.
Both the symbol $\refinessym$ and several of the rules are reminiscent of subtyping,
but that is misleading: sorts and types are not in an inclusion relation in the sense
of subtyping, because the rule for $\arr$ is covariant, not contravariant.
Covariance is needed for functions whose domains
are nontrivially refined, \eg $\Odd \arr \cdots$, which is not a subtype
of $\Bits \arr \cdots$ because $\Bits \not\subtype \Odd$.

Rule \RefinesSect implements the usual refinement restriction:
both parts of an intersection $A_1 \sectty A_2$ must refine the
same unrefined type $\tau$.

\subsection{Signatures}
\Label{sec:signatures}

Refinements are defined by \emph{signatures} $\Sigma$ (\Figureref{fig:signatures}).

\input{fig-signatures.tex}

\input{fig-type-wf.tex}

As in past datasort systems, we separate signatures $\Sigma$ from typing contexts
$\Gamma$.  Typing
assumptions over term variables ($x$, $y$, etc.) in $\Gamma$ can mention sorts
declared in $\Sigma$, but the signature $\Sigma$ cannot mention the term variables declared in
$\Gamma$.  Thus, our judgment for term typing will have the form
$\Sigma; \Gamma \entails e : A$, where the term $e$ can include constructors
declared in $\Sigma$ and variables declared in $\Gamma$, and the type $A$ can include
sorts declared in $\Sigma$.
Some judgments, like subsorting $\Sigma \entails s \subsort t$
and subtyping $\Sigma \entails A \subtype B$,
are independent of variable typing and don't include $\Gamma$ at all.

Traditional formulations of refinements assume the signature is given once at the beginning
of the program.  Since the same signature is used throughout a given typing derivation,
the signature can be omitted from the typing judgments.
In this paper, our goal is to support extensible refinements, where the signature can
evolve within a typing derivation; in this respect, the signature is analogous to an
ordinary typing context $\Gamma$, which is extended in subderivations that type
$\lambda$-expressions and other binding forms.
So the signature must be explicit in our judgment forms.

Constructor types $C$ are types of the form $A \arr s$.
In past formulations of datasorts, constructor types in the signature 
use intersection to represent multiple behaviours.  For example, a ``one'' constructor
for bitstrings, which represents appending a 1 bit, takes odd-parity bitstrings
to even-parity and vice versa; its type in the signature is the intersection type
$(\Odd \arr \Even) \sectty (\Even \arr \Odd)$.
Such a formulation ensures that the
signature has a standard property of (typing) contexts: each data constructor
is declared only once; additional behaviours are conjoined (intersected) within a single
declaration $c : C_1 \sectty C_2 \sectty \cdots$.
In our setting, we must be careful about not only \emph{which} types a constructor has,
but \emph{when} those types were declared.
The reasons are explained below;
for now, just note that we will write something like $c : C_1,\, \dots, c : C_2$
rather than $c : C_1 \sectty C_2$.

\mypara{Structure of signatures.}
A signature $\Sigma$ is a sequence of \emph{blocks} $S\sortblock{\blk}$
of declarations,
where refinements declared in outer scopes in the program
appear to the left of those declared in inner scopes.

Writing $(s{\refines}d)\sortblock{\blk}$
declares $s$ to be a sort refining some (unrefined) datatype $d$;
however, we usually elide the datatype and write just $s\sortblock{\blk}$.
The declarations $\blk$, called the \emph{block} of $s$,
define the values (constructors) of $s$, and the subsortings for $s$.
Declarations outside this block may declare new subsorts and supersorts
of $s$ \emph{only} if doing so would not affect $s$---for example, adding
inhabitants to $s$ via a constructor declaration, or declaring a new subsorting
between $s$ and previously declared sorts, would affect $s$ and will be
forbidden (via signature well-formedness).
The grammar generalizes this construct to multiple sorts, \eg
$(s_1{\refines}d_1,s_2{\refines}d_2)\sortblock{\blk}$,
abbreviated as $(s_1, s_2)\sortblock{\blk}$.

Writing $s_1 \subsort s_2$ says that $s_1$ is a subsort
of $s_2$, and $c : C$ says that constructor $c$ has type $C$, where $C$ has
the form $A \arr s$.  A constructor $c$ can be given more than one type:
~$   \Sigma
   ~=~
   (s, s_1, s_2)\sortblock{s_1 \subsort s,\, s_2 \subsort s,\, c : s_1{\arr}s_2,\, c : s_2{\arr}s_1}
$.

Adding inhabitants to a sort is only allowed within its block.
Thus, the following signature is ill-formed, because $c' : \unitty{\arr}s$
adds the value $c'\unit$ to $s$, but $c' : \unitty{\arr}s$ is not within $s$'s block:
$
   s\sortblock{c : s{\arr}s},
   \;
   t\sortblock{c' : \unitty{\arr}s}
$.
New sorts can be declared as subsorts and supersorts of each other, and of
previously declared sorts:
$
  s\sortblock{c_1 : \unitty{\arr}s, c_2 : \unitty{\arr}s},
  \;
  t\sortblock{t \subsort s, c_2 : \unitty{\arr}t}
$.

However, a block cannot modify the subsorting relation between earlier sorts;
``backpatching'' $s_1 \subsort s_2$ into the first block,
through a new intermediate sort $t$, is not permitted:
The signature
$
  \Sigma_* =
  (s_1, s_2)\sortblock{c : \unitty{\arr}s_1, c : \unitty{\arr}s_2},
  \;
  t\sortblock{s_1 \subsort t, t \subsort s_2}
$
is not permitted even though it looks safe: sorts $s_1$ and $s_2$
have the same set of inhabitants---the singleton set $\{c\unitexp\}$---so
the values of $s_1$ are a subset of the values of $s_2$.  But
this fact was not declared in the first block, which is the definition of
$s_1$ and $s_2$.  We assume the declaration of the first block completely
reflects the programmer's intent: if they had wanted structural subsorting,
rather than nominal subsorting, they should have declared $s_1 \subsort s_2$
in the first block.  Allowing backpatching would not violate soundness,
but would reduce the power of the type system: nominal subsorting
would no longer be supported, since it could be made structural
after the fact.

\mypara{Ordering.}
A block $S\sortblock{\blk}$ can refer to the sorts $S$ being defined
and to sorts declared to the left.
In contrast to block ordering, the order of declarations inside a block doesn't matter.

\subsection{Introduction form}

From a type-theoretic perspective, the first questions about a type are:
(1) How are the type's inhabitants created?  That is, what are the type's introduction rules?
(2) How are its inhabitants used?  That is, what are its elimination rules?  (\citet{Gentzen35} would ask the questions in this order;
the reverse order has been considered by Dummett, among others~\citep{ZeilbergerThesis}.)
In our setting, we must also ask: What happens with the introduction
and elimination forms when new refinements are introduced?

In the introduction rule---\TypeDataI in \Figureref{fig:typing}---the
signature $\Sigma$ is separated from the ordinary context $\Gamma$
(which contains
typing assumptions of the form $x : A$).
The typing of $c$ is delegated to its first premise, $\Sigma \entails c : A \arr s$,
so we need a way to derive this judgment.
At the top of \Figureref{fig:typing}, we define a single rule \NoLinkConArr,
which looks up the constructor in the signature
and weakens the result type (codomain), expressing a subsumption principle.
(Since we'll have
subsumption as a typing rule, including it here is an unforced choice;
its presence is meant to make the metatheory of constructor typing go more smoothly.)

In a system of extensible refinements,
adding refinements to a signature should preserve typing.
That is, if $\Sigma; \Gamma \entails e : B$, then $\Sigma, \Sigma'; \Gamma \entails e : B$.
This is a weakening property: we can derive,
from the judgment that $e$ has type $B$ under $\Sigma$, the logically weaker judgment
that $e$ has type $B$ under more assumptions $\Sigma, \Sigma'$.  (The signature
becomes longer, therefore stronger; but a turnstile is a kind of implication with the signature
as antecedent, so the judgment becomes weaker, hence ``weakening''.)
So for the introduction form, we need that if $\Sigma \entails c : A \arr s$,
then $\Sigma, \Sigma' \entails c : A \arr s$.  Under our formulation of the signature,
this holds:  If $c : A \arr s$, then there exists $(c : A \arr s') \in \Sigma$
such that $s' \subsort s$.  Therefore, there exists $(c : A \arr s') \in (\Sigma, \Sigma')$.
Likewise, since $\Sigma \entails s' \subsort s$, we also have $\Sigma, \Sigma' \entails s' \subsort s$.
One cannot use $\Sigma'$ to withdraw a commitment made in $\Sigma$.%
\footnote{Under the traditional formulation where each constructor has just one type
  in a signature, the relationship between the old signature $\Sigma$ and the new
  signature would be slightly more complicated: the old signature might contain $c : C_1$,
  and the new signature $c : C_1 \sectty C_2$, and we would need to explicitly
  eliminate the intersection to expose the old type $C_1$.  In our 
  formulation, the new signature appends additional typings for $c$ while keeping the typing $c : C_1$ intact.}

\subsection{Elimination form: case expressions}
\Label{sec:elimform}

Exhaustiveness checking for case expressions assumes complete knowledge
about the inhabitants of types.  Thus, we must avoid extending a signature
in a way that adds inhabitants to previously declared sorts.
Consider the case expression
$
  \Case{x : \sortname{empty}}{\Match{\datacon{Nil}\unitexp}{\unitexp}}
$
which is exhaustive for the signature
$   \Sigma = (\sortname{list}, \sortname{empty})
     \lblock\arrayenvl{\sortname{empty} \subsort \sortname{list},
\,%
              \datacon{Nil} : \unitty{\arr}\sortname{empty},
\,%
              \datacon{Cons} : \sortname{list}{\arr}\sortname{list}\rblock}
$
but not for %
\begin{rclll}
   (\Sigma, \Sigma')
   &=& (\sortname{list}, \sortname{empty})\lblock\arrayenvl{\sortname{empty} \subsort \sortname{list},
~%
              \datacon{Nil} : \unitty{\arr}\sortname{empty},
~%
              \datacon{Cons} : \sortname{list}{\arr}\sortname{list}\rblock,
              } \\  &&
              \lblock\datacon{Cons} : \sortname{list}{\arr}\sortname{empty}\rblock
\end{rclll}
Suppose we type-check the case expression under $\Sigma$, but then extend
$\Sigma$ to $(\Sigma, \Sigma')$.
Evaluating the above case expression with $x = \datacon{Cons}(\datacon{Nil}\unitexp)$
will ``fall off the end''.  The inversion principle that
``every \sortname{empty} has the form $\datacon{Nil}\unitexp$''
is valid under $\Sigma$, but with the additional type for \datacon{Cons} in $\Sigma'$,
that inversion principle becomes invalid under $(\Sigma, \Sigma')$.
Our system will reject the latter signature as ill-formed.

In the following, ``up'' and ``down'' are used
in the usual sense: a subsort is below its supersort.
In $\Sigma'$, the second constructor type for \datacon{Cons} had a smaller codomain than the first: the second
type had \sortname{empty}, instead of \sortname{list}.
Varying the codomain downward \emph{can} be sound when the lower codomain
is newly defined:
$
   \Sigma, \Sigma''
   \,=\,
   \Sigma,
\,%
   \sortname{subempty}\lblock
        \arrayenvl{
          \sortname{subempty} \subsort \sortname{empty},
\,%
          \datacon{Nil} : \unitty{\arr}\sortname{subempty}
          \rblock
        }
$.
Here, the inversion principle that every \sortname{empty} is \datacon{Nil} is still valid (along with the
new inversion principle that every \sortname{subempty} is \datacon{Nil}).  We only added
information about a new sort \sortname{subempty}, without changing the definition of 
\sortname{list} and \sortname{empty}.

\paragraph{Moving the domain down.}
Giving a new type whose domain is smaller, but that has the same codomain, is sound but pointless.
For example, extending
$\Sigma$ with $\datacon{Cons} : \sortname{empty}{\arr}\sortname{list}$, which
is the same as the type $\Sigma$ has for $\datacon{Cons}$ except that
the domain is \sortname{empty} instead of \sortname{list}, is sound.  The inversion
principle for values $v$ of type \sortname{list} in $\Sigma$ alone is ``either (1) $v$ has the form $\datacon{Nil}\unitexp$,
or (2) $v$ has the form $\datacon{Cons}(y)$ where $y$ has type \sortname{list}''.  Reading off the new inversion
principle for \sortname{list} from $\Sigma, \datacon{Cons} : \sortname{empty}{\arr}\sortname{list}$,
we get ``either (1) $v$ has the form $\datacon{Nil}\unitexp$,
or (2) $v$ has the form $\datacon{Cons}(y)$ where $y$ has type \sortname{list},
or (3) $v$ has the form $\datacon{Cons}(y)$ where $y$ has type \sortname{empty}''.  Since \sortname{empty}
is a subsort of \sortname{list}, part (3) implies part (2), and any case arm that checks under the assumption
that $y : \sortname{list}$ must also check under the assumption that $y : \sortname{empty}$.
Here, the new signature is equivalent to $\Sigma$ alone; the ``new'' type for \datacon{Cons}
is spurious.

\paragraph{Moving the codomain up.}
Symmetrically, giving a new type whose \emph{co}domain gets \emph{larger} is sound but pointless.
For example, adding $\datacon{Nil} : \unitty{\arr}\sortname{list}$ to $\Sigma$ has no effect, because
(in the introduction form) we could use the old type $\datacon{Nil} : \unitty{\,\arr\,}\sortname{empty}$
with subsumption ($\sortname{empty} \subsort \sortname{list}$). %

\paragraph{Moving the domain up.}
Making the domain of a constructor \emph{larger} is unsound in general.  To show this, we need a different
starting signature $\Sigma_2$.
\begin{rclll}
   \Sigma_2
   &=& (\sortname{list}, \sortname{empty}, \sortname{nonempty})
   \lblock
   \sortname{empty} \subsort \sortname{list},
   \sortname{nonempty} \subsort \sortname{list},
   \\ &&
   ~~~~\datacon{Nil} : \unitty{\arr}\sortname{empty},
   \datacon{Cons} : \sortname{empty}{\arr}\sortname{nonempty}
   \rblock
\end{rclll}
This isn't a very useful signature---it doesn't allow construction of any
list with more than one element---but it is illustrative.  We can read off from $\Sigma_2$ the following
inversion principle for values $v$ of sort \sortname{nonempty}: ``$v$ has the form $\datacon{Cons}(y)$
where $y$ has type \sortname{empty}''.
If $x : \sortname{nonempty}$
then
$\Case{x}{\Match{\datacon{Cons}(\datacon{Nil}\unitexp)}{\unitexp}}$
is exhaustive under $\Sigma_2$.
Now, extend $\Sigma_2$:
$   \Sigma_2, \Sigma_2'
   ~=~ \Sigma_2,
   \sortblock{\datacon{Cons} : \sortname{list}{\arr}\sortname{nonempty}}
$.
For the signature $\Sigma_2, \Sigma_2'$, the inversion principle for \sortname{nonempty} should be
``(1) $v$ has the form $\datacon{Cons}(y)$
where $y$ has type \sortname{empty},
or (2) $v$ has the form $\datacon{Cons}(y)$ where $y$ has type \sortname{list}''.
But there are more values of type \sortname{list} than of type \sortname{empty}.
The new inversion principle gives less precise information about the argument $y$,
meaning that the old inversion principle gives \emph{more} precise information than
$(\Sigma_2, \Sigma_2')$ allows.  Concretely, the case expression above
was exhaustive under $\Sigma_2$, but is not exhaustive under $(\Sigma_2, \Sigma_2')$
because $\datacon{Cons}(\datacon{Cons}(\datacon{Nil}\unitexp))$ has type $\sortname{list}$.

The above examples show that signature extension can be sound but useless,
unsound, or sound and useful (when the domain and codomain, or just the codomain,
are moved down).
Ruling out unsoundness will be the main purpose of our type
system, where unsoundness includes raising a ``match'' exception due to a
nonexhaustive case.  The critical requirement is that each block must not affect
previously declared sorts by adding constructors to them, or by adding
subsortings between them.

\subsection{Typing}

\Figureref{fig:typing} gives rules deriving the main typing
judgment $\Sigma; \Gamma \entails e : A$.  The variable rule
\TypeVar, the introduction (\TypeArrI) and elimination (\TypeArrE) rules for $\arr$,
and the introduction rules for
the unit type (\TypeUnitI) and products (\TypeProdI) are standard.
Products can be eliminated via $\Case{e}{\Match{\Pair{x_1}{x_2}}{\cdots}}$,
so they need no elimination rule.

\mypara{Subsumption.}
A subsumption rule \TypeSub incorporates subtyping, based on the
subsort relation $\subsort$; see \Sectionref{sec:subtyping}.
Several of the subtyping rules express the same properties as elimination
rules would; for example, anything of type $A_1 \sectty A_2$
has type $A_1$ and also type $A_2$.  Consequently, we can omit these
elimination rules without losing expressive power.

\input{fig-typing.tex}

\mypara{Intersection.}  The introduction rule \TypeSectI 
corresponds to a binary version of the introduction rule for
parametric polymorphism in System F.  The restriction to
a value $v$ avoids unsoundness in the presence of mutable
references~\citep{Davies00icfpIntersectionEffects}, similar
to SML's value restriction for parametric polymorphism~\citep{Wright95:value-restriction}.
We omit the elimination rules, which
are admissible %
using  \TypeSub
and subtyping (\Sectionref{sec:subtyping}).
\[
   \Infer{}
        {\Sigma; \Gamma \entails e : A_1 \sectty A_2}
        {\Sigma; \Gamma \entails e : A_1}
  ~~~~~
   \Infer{}
        {\Sigma; \Gamma \entails e : A_1 \sectty A_2}
        {\Sigma; \Gamma \entails e : A_2}
\]
\noskipmypara{Datasorts.}  
Rule \TypeDataI introduces a datasort, according to a constructor
type found in $\Sigma$ (via the $\Sigma \entails c : C$ judgment).
Rule \TypeDataE examines an expression $e$
of type $A$ and checks matches $ms$ under the assumption
that the expression matches the wildcard pattern $\wildcard$;
see \Sectionref{sec:pattern-typing}.

\mypara{Re-refinement.}
Rule \TypeDeclare allows sorts to be declared.
Its premises check that (1) the signature $\Sigma'$ is a valid extension
of $\Sigma$ (see \Sectionref{sec:signature-wf}); (2) the type $B$
of the expression is well-formed \emph{without} the extension $\Sigma'$,
which prevents sorts declared in $\Sigma'$ from escaping their scope;
(3)  that the expression $e$ is well-typed under the extended signature $(\Sigma, \Sigma')$.

\subsection{Subtyping}
\Label{sec:subtyping}

\input{fig-subtyping.tex}

Our subtyping judgment $\Sigma \entails A \subtype B$ says that
all values of type $A$ also have type $B$.  The rules follow the style
of \citet{Dunfield03:IntersectionsUnionsCBV};
in particular, the rules are orthogonal (each rule mentions only one
kind of connective) and transitivity is admissible.  
Instead of an explicit transitivity rule, we bake transitivity into
each rule; for example, rule \SubSectL{1} has a premise $A_1 \subtype B$
and conclusion $(A_1 \sectty A_2) \subtype B$, rather than just
$(A_1 \sectty A_2) \subtype A_1$ (with no premises).
This makes the rules easier to implement:
to decide whether $A \subtype C$, we never have to guess
a middle type $B$ such that $A \subtype B$ and $B \subtype C$.

\subsection{Signature well-formedness}
\Label{sec:signature-wf}

\input{fig-wf.tex}

A signature is well-formed if standard conditions (\eg no duplicate
declarations of sorts) \emph{and} conservation conditions hold.
Reading \Figureref{fig:wf} from bottom to top, we start with well-formedness
of signatures $\sigjudg{\Sigma}$.  For each block $S\sortblock{\blk}$,
rule \SigBlock checks that the sorts $S$ are not duplicates
($S \sect \dom{\Sigma} = \emptyset$), and then checks that
(1) subsorting is conserved by $\blk$ and
(2) each element in $\blk$ is \emph{safe}.

\mypara{(1) Subsorting preservation.}
  The subsortings declared in $\blk$ must not affect the subsort relation
  between sorts previously declared in $\Sigma$.  The left-to-right direction of
  this ``iff'' always holds by weakening: adding to a signature cannot delete
  edges in the subsort relation.  The right-to-left direction is contingent on
  the contents of $\blk$; see signature $\Sigma_*$ in \Sectionref{sec:signatures}.
  This premise could also be written as
  $(\Sigma \entails {\subsort}|_{\dom\Sigma})
  = (\Sigma, S\sortblock{\blk} \entails {\subsort}|_{\dom\Sigma})$,
  where ${\subsort}|_{\dom\Sigma}$
  is the $\subsort$ relation restricted to sorts in $\dom{\Sigma}$.

\mypara{(2a) Subsort elements.}
Rule \BlockSubsort checks that the subsorts are in scope. %

\mypara{(2b) Constructor element safety.}
Rule \BlockCon's first premise checks that $s \in S$.  (Certain declarations
   with $s \notin S$ would be safe, but useless.)
  Its second premise checks that the
  constructor type $A \arr s$ is well-formed.
  Finally, for all sorts $t$ that were (1) previously declared (in
  $\dom{\Sigma}$) and %
  (2) supersorts of the constructor's codomain ($s \subsort t$), the rule checks that
  the constructor is ``safe at $t$''.

The judgment $\safeextat{\Sigma}{S\sortblock{\blk}}{A \arr s}{c}{t}$
says that adding the constructor typing
$c : A \arr s$ does not invalidate $\Sigma$'s inversion principle for $t$.
Rule \SafeConAt
checks that signature $\Sigma$ already has a constructor typing
$c : A' \arr s'$, where $s' \subsort t$, such that $A \subtype A'$.
Thus, any value $c(v)$ typed using $c : A \arr s$
can already be typed using $c : A' \subsort s'$, which is a subsort of $t$,
so the new constructor typing $c : A \arr s$ does not add inhabitants to $t$.

This check is \emph{not} analogous to function subtyping, because
we need covariance ($A \subtype A'$), not contravariance.  The
relation $\refines$ (\Figureref{fig:signatures}) is a closer analogy.

More subtly, \SafeConAt also checks that $s \subsort s'$.
Suppose we have the signature
$
  \Sigma
  ~=~
  (t, s_1, s_2)\sortblock{s_1 \subsort t,\, s_2 \subsort t,\, c_1 : s_1,\, c_2 : s_2}
$
and extend it with $s\sortblock{s \subsort t,\, c_1 : s}$.  (To focus on the issue at
hand, we assume $c_1$ and $c_2$ take no arguments.)
For the original signature $\Sigma$, the inversion principle for $t$ is: If a value $v$
has type $t$,
then either $v = c_1$ and $v$ has type $s_1$, or $v = c_2$ and $v$ has type $s_2$.
However, under the extended signature, there is a new possibility: $v$ has type $s$.
Merely being inhabited by $c_1$ is not sufficient to allow $s$ to be a subsort of
$t$.

If, instead, we start with
$
  \Sigma'
  ~=~
  (t, s_1, s_2)\sortblock{\fighi{c_1 : t},\, s_1 \subsort t,\, s_2 \subsort t,\, c_1 : s_1,\, c_2 : s_2}
$
then the inversion principle for $t$ under $\Sigma'$ is that $v$ has type $s_1$, type
$s_2$, \emph{or} type $t$.  Therefore, any case arm whose pattern is $x \As c_1$ must
be checked assuming $x : t$.  If an expression can be typed assuming $x : t$,
then it can be typed assuming $x : t'$ for any $t' \subsort t$, so the inversion principle
(again, under $\Sigma'$ before extension) is equivalent to ``$v$ has type $t$''.
Extending $\Sigma'$ with $s\sortblock{s \subsort t,\, c_1 : s}$ would extend the
inversion principle to say ``if $v : t$ then $v$ has type $t$, or $v$ has type $s$'',
but since $s \subsort t$ the extended inversion principle is equivalent to that for $t$
under $\Sigma'$.

The $s \subsort s'$ premise of \SafeConAt is needed
to prove the constructor lemma (Lemma \ref{lem:constructor}),
which says that a constructor typing in an extended signature must be below a
constructor typing in the original signature.

\section{Typing Pattern Matching}
\Label{sec:pattern-typing}

Pattern matching is how a program gives different answers on different inputs.
A key motivation for datasort refinements is to
exclude impossible patterns, so that programmers can avoid having to choose between
writing impossible case arms (that raise an ``impossible'' exception)
and ignoring nonexhaustiveness warnings.
The pattern typing rules must model the relationship
between datasorts and the operational semantics of pattern matching.
It's no surprise, then, that in datasort refinement systems,
case expressions lead to the most interesting typing rules.

The relationship between types and patterns is more involved than with, say, Damas--Milner plus inductive
datatypes: with (unrefined) inductive datatypes, all the information needed to check for
exhaustiveness (also called coverage) is immediately available as soon as the type of
the scrutinee is known.
Moreover, types for pattern variables can be
``read off'' by traversing the pattern top-down,
tracking the definition of the scrutinee's inductive datatype.
But with datasorts, a set of patterns that looks nonexhaustive
at first glance---looking only at the head constructors---may in fact be exhaustive, thanks to the
inner patterns.

Giving types to pattern variables is also tricky, because sufficiently precise types
may be evident only after examining the whole pattern.  For example, when matching $x : \Bits$
against the pattern $y \As \datacon{One}(\datacon{Empty})$, we shouldn't settle on
$y : \Bits$ because the scrutinee $x$ has type $\Bits$; we should descend into the pattern
and observe that $\datacon{Empty} : \Even$ and $\datacon{One} : (\Even \arr \Odd)$,
so $y$ must have type $\Odd$.

Restricting the form of case expressions to a single layer of clearly disjoint
  patterns $c_1(x_1) \matchor \dots \matchor c_n(x_n)$ would simplify
  the rules, %
  at the cost of a big gap between theory and practice:  Since real implementations
  need to support nested patterns, the theory fails to model the real complexities of exhaustiveness checking
  and pattern variable typing.  Giving code examples becomes fraught; either we flatten case expressions
  (resulting in code explosion), or we handwave a lot.

Another option is to support the full syntax of case expressions, \emph{except for \xAs-patterns}, so that pattern variables
  occur only at the leaves.
  If subsorting were always structural, as in Davies's system, we could
  exploit a handy equivalence between patterns and values: if the pattern is $x \As c(p_0)$,
  let-bind $x$ to $c(p_0)$ inside the case arm, letting rule \TypeDataI figure out the type
  of $x$. 
  But with nominal subsorting, constructing a value is \emph{not} equivalent;
  see \citet[pp.\ 234--5]{DaviesThesis} and \citet[pp.\ 112--3]{DunfieldThesis}.

Our approach is to support the full syntax, including \xAs-patterns.  This approach
was taken by \citet[Chapter 4]{DunfieldThesis}, but our system seems
simpler---partly because (except for signature extension) our type system omits
indexed types and union types, but also because we avoid embedding typing derivations
inside derivations of pattern typing.

Instead, we confine most of the complexity to a single mechanism:
a function called $\intersectx$,
which returns a set of types (and contexts that type $\xAs$-variables) that
represent the intersection between a type and a pattern.  The definition of this
function is not trivial, but does not refer to expression-level typing.

\subsection{Unrefined pattern typing, match typing, and pattern operations}

\input{fig-pattern-type.tex}

\Figureref{fig:pattern-type} defines a judgment $\Ursig \entails p : \tau$ that
says that pattern $p$ matches values of unrefined type $\tau$ under
the unrefined signature $\Ursig$.

\input{fig-pattern-typing.tex}

Rule \TypeDataE for case expressions (\Figureref{fig:typing}) invokes a
match typing judgment, $\matchassn{\Sigma}{\Gamma}{A}{p} \entails ms : D$.
In this judgment, $p$ is a \emph{residual pattern} that represents the
space of possible values.  For the first arm in a case expression, no patterns have
yet failed to match, so the residual pattern in the premise of \TypeDataE is $\wildcard$.

Each arm, of the form $\Match{p_1}{e_1}$,
is checked by rule \TypeMs (\Figureref{fig:pattern-typing}).
The leftmost premises check that the type $A$
corresponds to the pattern type $\tau$.  The middle ``for all'' checks $e_1$ under
various assumptions produced by the $\intersectx$ function (\Sectionref{sec:intersect})
with respect to the pattern $p \patsect p_1$, ensuring that if $p_1$ matches the value at run time,
the arm is well-typed.
The last premise moves on to the remaining matches;
there, we know that the value did not match $p_1$,
so we subtract $p_1$ from the previous residual pattern $p$---expressed as
$p \patsect \patcomplement p_1$.
These operations are defined in the appendix (\Figureref{fig:complement-intersection}).

When typing reaches the end of the matches, $ms = \emptyms$ in rule \TypeMsEmpty,
we check that the case expression is exhaustive by checking that $\intersectx$
returns $\emptyset$.  For case expressions that are syntactically exhaustive, such as
a case expression over lists that has both \datacon{Nil} and \datacon{Cons} arms,
the residual pattern $p$ will be the empty pattern $\emptypattern$; the $\intersectx$
function on an empty pattern returns $\emptyset$.

We define pattern complement $\patcomplement p$
and pattern intersection $p_1 \patsect p_2$ in the appendix
(\Figureref{fig:complement-intersection}).
For example, $\patcomplement \wildcard = \emptypattern$.
No types appear in these definitions,
but the complement of a constructor pattern $c(p_0)$ uses the (implicit) unrefined signature $\Ursig$.
Our definition of pattern complement never generates $\xAs$-patterns, so we need not
define intersection for $\xAs$-patterns.

\input{fig-patterns.tex}

\subsection{The \textsf{intersect} function}
\Label{sec:intersect}

We define a function $\intersectx$ that builds the ``intersection''
of a type and a pattern.  Given a signature $\Sigma$,
type $A$ and pattern $p$, the $\intersectx$ function returns a (possibly empty) set
of \emph{tracks} $\{ \intresult{\Gamma_1'}{B_1}, \dots,
\intresult{\Gamma_n'}{B_n}\}$.
Each track $\intresult{\Gamma'}{B}$
has %
a list of typings $\Gamma'$ (giving the types of $\xAs$-variables)
and a type $B$ that %
represents the subset of values inhabiting $A$ that also match $p$.
The union of $B_1$ through $B_n$ constitutes the intersection of $A$ and $p$.
We call these ``tracks'' because each one represents a possible shape of the
values that match $p$, and the type-checking ``train'' must check a given
case arm under each track's $\Gamma'$.

Many of the clauses in the definition of $\intersectx$ (see \Figureref{fig:pattern-typing}) are
straightforward.  The intersection of $A$ with the wildcard $\wildcard$ is just
$\{\intresult{\cdot}{A}\}$.
Dually, the intersection of $A$
with the empty pattern $\emptypattern$ is the empty set.  In the same vein, the intersection of $A$
with the or-pattern $p_1 \patunion p_2$ is the union of two intersections ($A$ with $p_1$, and
$A$ with $p_2$).  The intersection of a product $A_1 * A_2$ with a pair pattern is the union of products of the pointwise intersections.

The most interesting case is when we intersect a sort $s$ with a pattern of the form $c(p_0)$.
For this case, $\intersectx$ iterates through all the constructor declarations in $\Sigma$ that
could have been used to create the given value: those of the form $(c : A_c \arr s_c)$
where $s_c \subsort s$.  For each such declaration, it calls $\intersectx$ on $A_c$ and $p_0$.
For each resulting track $\intresult{\Gamma'}{B}$,
it returns a track
$\intresult{\Gamma'}{s_c}$.

\mypara{Optimization.}
In practice, it may be necessary to optimize the
result of $\intersectx$.  If
$\Sigma = 
(\sortname{list}, \sortname{empty})
     \lblock\sortname{empty} \subsort \sortname{list},
              \arrayenvl{\datacon{Nil} : \unitty{\arr}\sortname{empty},
              \datacon{Cons} : \sortname{empty}{\arr}\sortname{list},
              \datacon{Cons} : \sortname{list}{\arr}\sortname{list}\rblock}
$
then
$\intersect{\Sigma}{\datacon{Cons}(x \As \wildcard)}{\sortname{list}}$
returns
$
\big\{\intresult{x : \sortname{empty}}{\sortname{list}},~
\intresult{x : \sortname{list}}{\sortname{list}}\big\}
$.
Since any case arm that checks under $x : \sortname{list}$
will check under $x : \sortname{empty}$, there is no point
in trying to check under $x : \sortname{empty}$.  Instead,
we should check only under $x : \sortname{list}$.
A similar optimization 
in the Stardust type checker could reduce the size of the set
of tracks by ``about an order of magnitude'' \citep[p.\ 112]{DunfieldThesis}.

\mypara{Missing clauses?}  As is standard in typed languages, pattern matching
doesn't look inside $\lambda$, so $\intersectx$ needs no clause for $\arr$/$\lambda$.
If we can't match on an arrow type, we don't need to match on intersections of arrows.
The other useful case of intersection is on sorts, $s_1 \sectty s_2$.  However, an intersection
of sorts can be obtained by declaring a new sort below $s_1$ and $s_2$ with the
appropriate constructor typings, so we omit such a clause from the definition.

\mypara{Comparison to an earlier system.}
A declarative system of rules in \citet[Chapter 4]{DunfieldThesis}
appears to be a conservative extension of $\intersectx$:
the earlier system supports a richer type system, but for the features in common,
the information produced is similar to that of $\intersectx$.
The earlier system was based on a judgment
$\Sigma \entails p \against A \KK (e \against D)$.
To clarify the connection to the present system, we adjust notation;
for example, we make $\Sigma$ explicit.

The meta-variables $\Sigma$, $p$, and $A$ directly correspond to the arguments to $\intersectx$,
while $e$ and $D$ correspond to $e_1$ and $D$ in our rule \TypeMs.
No meta-variables correspond directly to the tracks in the \emph{result} of $\intersectx$,
but within
$\Sigma \entails p \against A \KK (e \against B)$,
we find subderivations of
$B + \Gamma \entails \FORGETTYPE \KK e \against D$,
where the set of pairs $\langle \Gamma, B \rangle$
indeed correspond to the result of $\intersectx$.

Cutting through the differences in the formalism,
and omitting rules for unions and other features not present in this paper,
the earlier system behaves like $\intersectx$.
For example, $\Pair{p_1}{p_2}$ was also handled by considering each
component, and assembling all resulting combinations.
Perhaps most importantly, $c(p_0)$ was also handled by considering each constructor type
in the signature, filtering out inappropriate codomains, and recursing on $p_0$.
A rule for $\sectty$ appears in the declarative system in \citet[Chapter 4]{DunfieldThesis},
but the rule was never implemented, and seems not to be needed in practice.

Since the information given by the older system is precise enough to check interesting
invariants of actual programs, our definition of $\intersectx$ should also be precise enough.

\section{Operational Semantics}
\Label{sec:opsem}

We prove our results with respect to a call-by-value, small-step operational semantics.
The main judgment form is $e \step e'$, which uses evaluation contexts $\E$.  Stepping
\textkw{case} expressions
is modelled using a judgment $ms \step_v e'$, which compares each pattern in
$ms$ against the value $v$ being cased upon.
This comparison is handled by the judgment $\vmatch{p}{v}{\theta}$,
which says that $\theta$ is evidence that $p$ matches $v$ (that is, $[\theta]p = v$).
The rules are in \Figureref{fig:step} in the appendix.

\input{meta.tex}

\input{bidir-meta.tex}

\section{Related Work}
\Label{sec:related}

\noskipmypara{Datasort refinements.}
\citet{Freeman91} %
introduced datasort refinements with intersection types,
defined the refinement restriction (where $A \sectty B$ is well-formed only if
$A$ and $B$ are refinements of the same type), and developed an inference algorithm
in the spirit of abstract interpretation.  As discussed earlier, the lack
of annotations not only makes the types difficult to see, but makes inference prone
to finding long, complex types that include accidental invariants.

\citet{DaviesThesis}, building on the type system developed by
\citet{Davies00icfpIntersectionEffects}, used a bidirectional typing algorithm,
guided by annotations on redexes.
This system supports parametric polymorphism through a front end based on
Damas--Milner inference, but---like Freeman's system---does not support extensible
refinements. %
Davies's CIDRE implementation~\citep{DaviesSMLCIDREnew} goes beyond
his formalism by allowing a single type to be refined via multiple declarations,
but this has no formal basis; CIDRE appears to simply gather the multiple declarations
together, and check the entire program using the combined declaration, even when
this violates the expected scoping rules of SML declarations.

Datasort refinements were combined with union types and indexed types by
\citet{Dunfield03:IntersectionsUnionsCBV,Dunfield04:Tridirectional}, who
noticed the expressive power of nominal subsorting, called ``invaluable refinement''~\citep[pp.\ 113, 220--230]{DunfieldThesis}.

Giving multiple refinement declarations for a single datatype was
mentioned early on, as future work:
``embedded refinement type declarations'' \citep[p.\ 275]{Freeman91};
``or even \dots declarations that have their scope limited'' \citep[p.\ 167]{FreemanThesis};
``it does seem desirable to be able to make local datasort declarations'' \citep[p.\ 245]{DaviesThesis}.
But the idea seems not to have been pursued.

\mypara{Logical frameworks.}
In the logical framework LF \citep{Harper93}, data is characterized
by declaring constructors with their types.  In this respect, our system is closer
to LF than to ML: %
LF doesn't require all of a type's constructors to be declared together.
By itself, LF has no need for inversion principles.  However,
systems such as Twelf~\citep{Pfenning99:Twelf}, Delphin~\citep{Poswolsky08:Delphin}
and Beluga~\citep{Pientka10:Beluga} use LF as an object-level language but also
provide meta-level features.  One such feature is coverage (exhaustiveness) checking,
which needs inversion principles for LF types.
Thus, these systems mark a type as \emph{frozen} when its inversion principle is
applied (to process \textvtt{\%covers} in Twelf, or a case expression in Beluga);
they also allow the user to mark types as frozen.
These systems lack subtyping and subsorting; once a type is frozen,
it is an error to declare a new constructor for it. %

\citet{LovasThesis} extended LF with refinements and subsorting, and developed
a constraint-based algorithm for signature checking.
This work did not consider meta-level features such as coverage checking, so it yields no
immediate insights about inversion principles or freezing.
Since Lovas's system takes the subsorting relation directly from declarations, rather than
by inferring it from a grammar, it supports what \citet{DunfieldThesis} called
invaluable refinements; see Lovas's example~\citep[pp.\ 145--147]{LovasThesis}.

\mypara{Indexed types and refinement types.}
As the second generation of datasort refinements (exemplified by the work of Davies and Pfenning)
began, so did a related approach to lightweight type-based verification: \emph{indexed types}
or \emph{limited dependent types}~\citep{Xi99popl,XiThesis},  %
in which datatypes are refined by indices drawn from a (possibly infinite) constraint domain.
Integers with linear inequalities are the standard example of an
index domain; another good example is physical units or dimensions~\citep{Dunfield07:Stardust}.
More recent work in this vein, such as liquid types~\citep{Rondon08}, uses ``refinement types''
for a mechanism close to indexed types.%

Datasort refinements have always smelled like a special case of indexed types.
At the dawn of indexed types (and the second generation
of datasort refinements), the relationship was obscured by datasorts' ``fellow traveller'',
intersection types, which were absent from the first indexed type systems, and remain absent
from the approaches now called ``refinement types''.  That is, while datasorts themselves
strongly resemble a specific form of indices---albeit related by a partial order (subtyping),
rather than by equality---and would thus suggest that indexed type systems subsume 
datasort refinement type systems, the inclusion of intersection types confounds such a comparison.
Intersection types \emph{are} present, along with both datasorts and indices,
in \citet{Dunfield03:IntersectionsUnionsCBV}
and \citet{DunfieldThesis};  the relationship is less obscured.  But no one has given an encoding
of types with datasorts into types with indices, intersections or no.

The focus of this paper is a particular kind of \emph{extensibility} of datasort refinements,
so it is natural to ask whether indexed types and (latter-day) refinement types have anything
similar.  Indexed types are not immediately extensible: both Xi's DML and Dunfield's Stardust require
that a given datatype be refined exactly once.  Thus, a particular list type may carry its length,
or the value of its largest element, or the parity of its boolean elements.  By refining the type with a tuple
of indices, it may also carry combinations of these, such as its length \emph{and} its largest
element.  Subsequent uses of the type can leave out some of the indices, but the combination
must be stated up front.

However, some of the approaches descended from DML, such as liquid types, allow
refinement with a predicate that can mention
various attributes.  These attributes are declared separately
from the datatype; adding a new attribute does not invalidate existing code.
Abstract refinement types~\citep{Vazou13} even allow types to quantify over predicates.

Setting aside extensibility, datasort refinements %
can
express certain invariants more clearly and succinctly than indexed types (and
their descendants).  %
\mypara{Program analysis.}
\citet{Koot15} formulate a type system that analyzes where exceptions can be raised, including
match exceptions raised by nonexhaustive case expressions.  This system appears to be less
precise than datasorts, but has advantages typical to program analysis: no type annotations
are required.

\section{Future Work}
\Label{sec:future}

\paragraph{Modular refinements.}
This paper establishes a critical mechanism for extensible refinements,
safe signature extension, in the setting of a core language without
modules: refinements are lexically scoped.
To scale up to a language with modules, we need to ask:
what notions of scope are appropriate?
For example, a strict $\lambda$-calculus interpreter
could be refined with a sort \tyname{val} of values, while a lazy
interpreter could be refined with a sort \tyname{whnf}
of terms in weak head normal form.  If every \tyname{val}
is a \tyname{whnf}, we might want to have $\tyname{val} \subsort \tyname{whnf}$.
In the present system, these two refinements could be in separate
\keyword{declare} blocks; in that case, \tyname{val} and \tyname{whnf}
could not both be in scope, and the subsorting is not well-formed.
Alternatively, one \keyword{declare} block could be nested inside the other.
In that case, $\tyname{val} \subsort \tyname{whnf}$ could be given in the nested
block, since it would not add new subsortings within the outer refinement.
In a system with modules, we would likely want to
have $\tyname{val} \subsort \tyname{whnf}$,
at least for clients of \emph{both} modules;
such backpatching is currently not allowed, but should be safe
since the new subsorting crosses two independent signature blocks
(the block declaring \tyname{val} and the block declaring \tyname{whnf})
without changing the subsortings within each block.

\mypara{Type polymorphism.}
Standard parametric polymorphism is absent in this paper, but
it should be feasible to follow the approach of
\citet{DaviesThesis}, as long as the unrefined datatype declarations
are not themselves extensible (which would break signature
well-formedness, even without polymorphism).

\mypara{Datasort polymorphism.}
Extensible signatures open the door to sort-bounded polymorphism.  In our current system,
a function that iterates over an abstract syntax tree and $\alpha$-renames free variables---which
would conventionally have the type $\Exp \arr \Exp$---must be duplicated,
even though the resulting tree has the same shape and the same constructors,
and therefore should always produce a tree of the same sort as the input tree
(at least, if the free variables are not specified with datasorts).  We would like
the function to check against a polymorphic type
$\BAll{\alpha}{\Exp} \alpha \arr \alpha$, which works for any
sort $\alpha$ below $\Exp$.

We would like to
reason ``backwards'' from a pattern match over a polymorphic sort variable
$\alpha$.  For example, if a value of type $\alpha$ matches the pattern
$\Plus(x1, x2)$, then we know that $\Plus : (\alpha_1 * \alpha_2) \arr \alpha$ for some
sorts $\alpha_1$ and $\alpha_2$.
The recursive calls on $x1$ and $x2$ must preserve the property of
being in $\alpha_1$ and $\alpha_2$, so $\Plus(f~x1, f~x2)$
has type $\alpha$, as needed.
The mechanisms we have developed may be a good foundation
for adding sort-bounded polymorphism: the $\intersectx$ function
would need to return a signature, as well as a context and type,
so that the constructor typing $\Plus : (\alpha_1 * \alpha_2) \arr \alpha$
can be made available.

\mypara{Implementation.}
Currently, we have a prototype of a few pieces of the system, including
a parser and implementations of the $\sigjudg{\Sigma}$ judgment
and the $\intersectx$ function.  Experimenting with these pieces was helpful
during the design of the system (and reassured us that the most novel
parts of our system can be implemented), but they fall short of a usable
implementation.

\ifnum\OPTIONConf=0
  \addtolength{\bibsep}{-2.0pt}
\fi
\ifnum\OPTIONConf=1
  \addtolength{\bibsep}{-0.5pt}
\fi
\ifnum\OPTIONConf=2
  \addtolength{\bibsep}{2.0pt}
\fi

\bibliographystyle{plainnat}
\bibliography{intcomp}

\clearpage
\appendix
\input{apxdefs.tex}
\input{proofs.tex}

\end{document}

%% file: kontroll.tex
\ifnum\OPTIONConf=1%
\def\OPTIONLoudLabels{0}
\else
\def\OPTIONLoudLabels{1}
\fi

%% file: local.tex
\makeatletter
\newcommand{\XLabel}[1]{%
  \ifcsname IDEMPOTFLAG#1\endcsname%
      \global\edef\CurrentProofName{#1}%
  \else%
       \Label{#1}%
      \expandafter\gdef\csname IDEMPOTFLAG#1\endcsname{d}%
   \fi}
\makeatother

\newcommand{\sectty}{\mathrel{\mathcolor{dBlue}{\land}}}          %
\newcommand{\unty}{\mathrel{\mathcolor{dRed}{\lor}}}          %

\newcommand{\step}{\mapsto}

\newcommand{\xdom}{\mathsf{dom}}

\newcommand{\xCase}{\ensuremath{\keyword{case}}}
\newcommand{\Case}[2]{{\xCase\;#1\;\keyword{of}\;#2}}

\newcommand{\emptyms}{\ensuremath{\emptyset}}
\newcommand{\matchor}{\ensuremath{\normalfont\,\texttt{|\hspace{-1.120ex}|}\,}}
\newcommand{\Match}[2]{{#1}\Rightarrow{#2}}

\newcommand{\Bits}{\tyname{bits}}

\newcommand{\Exp}{\tyname{exp}}
\newcommand{\Plus}{\datacon{Plus}}

\newcommand{\Even}{\tyname{even}}
\newcommand{\Odd}{\tyname{odd}}

\newcommand{\Var}{\datacon{Var}}
\newcommand{\Neg}{\datacon{Not}}
\newcommand{\CAnd}{\datacon{And}}
\newcommand{\Or}{\datacon{Or}}

\newcommand{\CNF}{\tyname{cnf}}
\newcommand{\Clause}{\tyname{clause}}
\newcommand{\PosLiteral}{\tyname{pos\text{-}literal}}
\newcommand{\Literal}{\tyname{literal}}
\newcommand{\ConjLiteral}{\tyname{conj\text{-}literal}}

\def\gets{\ensuremath{\mathrel{{:}{=}}}}

\newcommand{\Unit}{\text{\normalfont\sf 1}}
\newcommand{\unitty}{\Unit}

\newcommand{\unit}{\textt{()}}
\newcommand{\unitexp}{\unit}

\newcommand{\chkcolor}{dBlue}
\newcommand{\syncolor}{dRed}

\newcommand{\chk}{\mathrel{\mathcolor{\chkcolor}{\Leftarrow}}}
\newcommand{\uncoloredsyn}{{\Rightarrow}}
\newcommand{\syn}{\mathrel{\mathcolor{\syncolor}{\uncoloredsyn}}}

\ifnum\OPTIONConf=2
 \makeatletter

 \let\c@lemma\undefined

 \makeatother
\declaretheoremstyle[
  bodyfont=\sl
]{mytheoremstyle}

\declaretheorem[name=Theorem, style=mytheoremstyle]{theorem}

\declaretheorem[name=Lemma, style=mytheoremstyle]{lemma}

\declaretheorem[name=Corollary, style=mytheoremstyle, sibling=theorem]{corollary}
\declaretheorem[name=Conjecture, style=mytheoremstyle, sibling=lemma]{conjecture}
\declaretheorem[name=Definition, style=mytheoremstyle, sibling=lemma]{definition}
\else
\declaretheoremstyle[
  bodyfont=\sl
]{mytheoremstyle}

\declaretheorem[style=mytheoremstyle]{theorem}

\declaretheorem[style=mytheoremstyle]{lemma}

\declaretheorem[style=mytheoremstyle, sibling=theorem]{corollary}
\declaretheorem[style=mytheoremstyle, sibling=lemma]{conjecture}
\declaretheorem[style=mytheoremstyle, sibling=lemma]{definition}
\fi

\renewcommand{\phi}{\varphi}

\newcommand{\sig}{~\textit{sig}}
\newcommand{\val}{~\textit{value}}
\newcommand{\type}{~\textit{type}}

\newcommand{\subsortsym}{\preceq}
\newcommand{\subsort}{\mathbin{\subsortsym}}

\newcommand{\propersubsortsym}{\prec}

\newcommand{\notpropersubsortsym}{\not\propersubsortsym}

\newcommand{\anno}[2]{\text{\textvtt{(}}{#1} \text{\texttt{:}} {#2}\text{\textvtt{)}}}

\newcommand{\emptysubst}{\cdot}

\newcommand{\dec}[1]{#1^*}

\newcommand{\Bdec}{\dec{B}}

\newcommand{\Avec}{\vec{A}}

\ifnum\OPTIONLoudLabels=1%
\newcommand{\Label}[1]{\LoudLabel{#1}}%
\newcommand{\FLabel}[1]{\label{#1}%
{\tt\scriptsize{#1}}}%
\else%
\newcommand{\Label}[1]{\label{#1}}%
\newcommand{\FLabel}[1]{\label{#1}}%
\fi

\newcommand{\mydefaultarraystretch}{1.2}

\newcommand{\wildcard}{\text{\_\!\_}}
\newcommand{\xAs}{\keyword{as}}
\newcommand{\As}{~\xAs~}

\newcommand{\KK}{\,\mathrel{\rhd}\,}

\newcommand{\emptypattern}{\emptyset}
\newcommand{\patunion}{\sqcup}
\newcommand{\patsect}{\sect}
\newcommand{\patcomplementsym}{\lnot}
\newcommand{\patcomplement}{\patcomplementsym}

\newcommand{\emptyctx}{\cdot}

\def\quofile{15}
\newcommand{\startMquotes}{\immediate\openout \quofile=\jobname.quo}
\newcommand{\MquoteM}[4]{{\small \label{#3} \begin{quotation} #1 \vspace{-0.5em} \flushright{---#2} \end{quotation}} { \immediate\write\quofile{\expandafter\csname Mquoteentry\endcsname{}{#2}{#3}{\expandafter\csname #4\endcsname}}}}
\newcommand{\Mquote}[4]{{\small \label{#3} \begin{quotation} #1 \vspace{-0.5em} \flushright{---#2} \end{quotation}} { \immediate\write\quofile{\expandafter\csname Mquoteentry\endcsname{}{#2}{#3}{#4}}}}

\newcommand{\AllSym}{\forall}
\newcommand{\xAll}[1]{\AllSym#1}
\newcommand{\All}[1]{\xAll{#1}.\:}

\newcommand{\BAll}[2]{\All{{#1}{\subsort}{#2}}}

\newdimen\zzlistingsize
\newdimen\zzlistingsizedefault
\zzlistingsize=\zzlistingsizedefault
\global\def\CommentCopter{0}
\newcommand{\Lstbasicstyle}{\fontsize{\zzlistingsize}{1.05\zzlistingsize}\ttfamily}
\newcommand{\keywordcopter}{\fontsize{1.0\zzlistingsize}{1.0\zzlistingsize}\bf}
\newcommand{\sillycopter}{\if0\CommentCopter\keywordcopter\fi}
\newcommand{\commentcopter}{\def\CommentCopter{1}\fontsize{0.95\zzlistingsize}{1.0\zzlistingsize}\rmfamily\slshape}

\newlength{\zzlstwidth}
\newcommand{\setlistingsize}[1]{\zzlistingsize=#1%
\settowidth{\zzlstwidth}{{\Lstbasicstyle~}}}
\setlistingsize{\zzlistingsizedefault}

\newcommand{\dom}[1]{\xdom(#1)}
\newcommand{\xconstructors}{\mathsf{constructors}}
\newcommand{\constructors}[1]{\xconstructors(#1)}

\newcommand{\lblock}{\langle}
\newcommand{\rblock}{\rangle}
\newcommand{\sortblock}[1]{\lblock {#1} \rblock}

\newcommand{\sigjudg}[1]{{#1}\sig}
\newcommand{\sigjudgPf}[2]{\Pf{}{#1}{\!\sig} {#2}}

\newcommand{\disjoint}[2]{{#1} \sect {#2} = \emptyset}

\newcommand{\blkjudg}[4]{#1; #2\sortblock{#3} \entails {#4}\textit{~safe}}
\newcommand{\blkjudgPf}[5]{\ePf{#1; #2\sortblock{#3}}{{#4}\textit{~safe}} {#5}}

\newcommand{\contypejudg}[2]{{#1} \entails {#2}\textit{~contype}}
\newcommand{\contypejudgPf}[3]{\ePf{#1}{#2\textit{~contype}}{#3}}
\newcommand{\typejudg}[2]{{#1} \entails {#2}\textit{~type}}

\newcommand{\refinessym}{\sqsubset}
\newcommand{\refines}{\mathrel{\refinessym}}
\newcommand{\safeextat}[5]{{#1}; {#2} \entails {#4} : {#3} \textit{~safe at~}{#5}}
\newcommand{\safeextatPf}[6]{\ePf{{#1}; {#2}}{{#4} : {#3} \textit{~safe at~}{#5}}{#6}}
\newcommand{\matchassn}[4]{{#1}; {#2}; {#4} : {#3}}

\newcommand{\Thnkx}[1]{\textsc{{\MakeLowercase{#1}}}}

\newcommand{\FORGETTYPE}{\Thnkx{forgettype}\xspace}

\newcommand{\against}{\mathrel{\Leftarrow}}

\newcommand{\intersectx}{\ensuremath{\mathsf{intersect}}}
\newcommand{\intersect}[3]{\intersectx({#1} \entails {#2};\, {#3})}

\newcommand{\intersectresult}[3]{%
     (\intersectresultNOPARENS{#1}{#2}{#3})%
}
\newcommand{\intersectresultNOPARENS}[3]{%
     {#1}%
     \,;\; {#2}
     \,{\vdash}\, %
     {#3}%
}

\newcommand{\intresult}[2]{%
     (\intresultNOPARENS{#1}{#2})%
}
\newcommand{\intresultNOPARENS}[2]{%
     {#1}%
     \,{\vdash}\, %
     {#2}%
}

\newcommand{\intrsctresult}[3]{%
     (\intrsctresultNOPARENS{#1}{#2}{#3})%
}
\newcommand{\intrsctresultNOPARENS}[3]{%
     {#2}%
     \,{\vdash}\, %
     {#3}%
}

\newcommand{\xdeclare}{%
  \textkw{declare}}
\newcommand{\declare}[1]{%
  \xdeclare~{#1}%
       \textkw{~in~}%
}

\newcommand{\rulename}[1]{\text{\normalfont\textsf{#1}}}

\newcommand{\TypeMsrulename}[1]{\rulename{TypeMs#1}}
\newrulecommand{TypeMsEmpty}{\TypeMsrulename{Empty}}
\newrulecommand{TypeMs}{\TypeMsrulename{}}

\newcommand{\Steprulename}[1]{\rulename{Step#1}}

\newrulecommand{StepAppL}{\Steprulename{AppL}}
\newrulecommand{StepAppR}{\Steprulename{AppR}}
\newrulecommand{StepBeta}{\Steprulename{Beta}}
\newrulecommand{StepPairL}{\Steprulename{PairL}}
\newrulecommand{StepPairR}{\Steprulename{PairR}}
\newrulecommand{StepConArg}{\Steprulename{ConArg}}
\newrulecommand{StepCase}{\Steprulename{Case}}

\newrulecommand{StepDeclare}{\Steprulename{Declare}}
\newrulecommand{StepContext}{\Steprulename{Context}}

\newrulecommand{StepMatch}{\Steprulename{Match}}
\newrulecommand{StepElse}{\Steprulename{Else}}

\newcommand{\Refinesrulename}[1]{\rulename{$\refines$#1}}

\newrulecommand{RefinesUnit}{\Refinesrulename{$\unitty$}}
\newrulecommand{RefinesArr}{\Refinesrulename{$\arr$}}
\newrulecommand{RefinesProd}{\Refinesrulename{$*$}}
\newrulecommand{RefinesData}{\Refinesrulename{Data}}
\newrulecommand{RefinesSect}{\Refinesrulename{$\sectty$}}

\newcommand{\Substrulename}[1]{\rulename{Subst#1}}

\newrulecommand{SubstEmpty}{\Substrulename{Empty}}
\newrulecommand{SubstVar}{\Substrulename{Var}}

\newcommand{\Subsortrulename}[1]{\ensuremath{{\subsortsym}\rulename{#1}}}

\newrulecommand{SubsortAssum}{\Subsortrulename{Assum}}
\newrulecommand{SubsortRefl}{\Subsortrulename{Refl}}
\newrulecommand{SubsortTrans}{\Subsortrulename{Trans}}

\newcommand{\Notpropersubsortrulename}[1]{\ensuremath{{\notpropersubsortsym}\rulename{#1}}}

\newrulecommand{NotpropersubsortNotsubsort}{\Notpropersubsortrulename{Notsubsort}}
\newrulecommand{NotpropersubsortSupersort}{\Notpropersubsortrulename{Supersort}}

\newcommand{\Saferulename}[1]{\ensuremath{\rulename{Safe#1}}}

\newrulecommand{SafeArr}{\Saferulename{Arr}}   %
\newrulecommand{SafeConAt}{\Saferulename{ConAt}}   %
\newrulecommand{SafeCon}{\Saferulename{Con}}   %

\newcommand{\Contyperulename}[1]{\ensuremath{\rulename{Contype#1}}}

\newrulecommand{ContypeArr}{\Contyperulename{Arr}}
\newrulecommand{ContypeSect}{\Contyperulename{Sect}}

\newcommand{\Sigrulename}[1]{\ensuremath{\rulename{Sig#1}}}

\newrulecommand{SigEmpty}{\Sigrulename{Empty}}
\newrulecommand{SigBlock}{\Sigrulename{Block}}
\newrulecommand{SigBound}{\Sigrulename{Bound}}

\newcommand{\Blkrulename}[1]{\ensuremath{\rulename{Block#1}}}

\newrulecommand{BlkEmpty}{\Blkrulename{Empty}}
\newrulecommand{BlockCon}{\Blkrulename{Con}}
\newrulecommand{BlockSubsort}{\Blkrulename{Subsort}}

\newcommand{\Ctxrulename}[1]{\ensuremath{\rulename{Ctx#1}}}

\newrulecommand{CtxEmpty}{\Ctxrulename{Empty}}
\newrulecommand{CtxVar}{\Ctxrulename{Var}}

\newcommand{\Conrulename}[1]{\ensuremath{\rulename{Con#1}}}

\newrulecommand{ConArr}{\Conrulename{Arr}}
\newrulecommand{ConSub}{\Conrulename{Sub}}   %
\newrulecommand{ConSect}{\Conrulename{Sect}}   %

\newcommand{\Matchrulename}[1]{\ensuremath{\rulename{Match#1}}}

\newrulecommand{MatchWild}{\Matchrulename{Wild}}
\newrulecommand{MatchAs}{\Matchrulename{As}}
\newrulecommand{MatchOr}{\Matchrulename{Or}}
\newrulecommand{MatchCon}{\Matchrulename{Con}}
\newrulecommand{MatchUnit}{\Matchrulename{Unit}}
\newrulecommand{MatchPair}{\Matchrulename{Pair}}

\newcommand{\Nomatchrulename}[1]{\ensuremath{\rulename{NoM#1}}}

\newrulecommand{NomatchEmpty}{\Nomatchrulename{Wild}}
\newrulecommand{NomatchAs}{\Nomatchrulename{As}}
\newrulecommand{NomatchOr}{\Nomatchrulename{Or}}
\newrulecommand{NomatchConInner}{\Nomatchrulename{ConInner}}
\newrulecommand{NomatchConHead}{\Nomatchrulename{ConHead}}
\newrulecommand{NomatchUnit}{\Nomatchrulename{Unit}}
\newrulecommand{NomatchPairInner}{\Nomatchrulename{PairInner}}
\newrulecommand{NomatchPairHead}{\Nomatchrulename{PairHead}}

\newcommand{\Pattyperulename}[1]{\ensuremath{\rulename{p-#1}}}

\newrulecommand{PattypeWild}{\Pattyperulename{Wild}}
\newrulecommand{PattypeAs}{\Pattyperulename{As}}
\newrulecommand{PattypeEmpty}{\Pattyperulename{Empty}}
\newrulecommand{PattypeOr}{\Pattyperulename{Or}}
\newrulecommand{PattypeUnit}{\Pattyperulename{Unit}}
\newrulecommand{PattypePair}{\Pattyperulename{Pair}}
\newrulecommand{PattypeCon}{\Pattyperulename{Con}}

\newcommand{\Typewfrulename}[1]{\ensuremath{\rulename{WfType#1}}}

\newrulecommand{TypewfUnit}{\Typewfrulename{$\unitty$}}
\newrulecommand{TypewfProd}{\Typewfrulename{$*$}}
\newrulecommand{TypewfArr}{\Typewfrulename{$\arr$}}
\newrulecommand{TypewfSect}{\Typewfrulename{$\sectty$}}
\newrulecommand{TypewfBin}{\Typewfrulename{Bin}}
\newrulecommand{TypewfSort}{\Typewfrulename{Sort}}
\newrulecommand{TypewfAll}{\Typewfrulename{$\AllSym$}}

\newcommand{\Typerulename}[1]{\ensuremath{\rulename{#1}}}

\newrulecommand{TypeVar}{\Typerulename{Var}}
\newrulecommand{TypeSub}{\Typerulename{Sub}}
\newrulecommand{TypeAnno}{\Typerulename{Anno}}
\newrulecommand{TypeUnitI}{\Typerulename{$\unitty$I}}
\newrulecommand{TypeArrI}{\Typerulename{$\arr$I}}
\newrulecommand{TypeArrE}{\Typerulename{$\arr$E}}
\newrulecommand{TypeProdI}{\Typerulename{$*$I}}
\newrulecommandONE{TypeProdE}{\Typerulename{$*$E$_{#1}$}}

\newrulecommand{TypeSectI}{\Typerulename{$\sectty$I}}
\newrulecommandONE{TypeSectE}{\Typerulename{$\sectty$E$_{#1}$}}

\newrulecommand{TypeAllI}{\Typerulename{$\AllSym$I}}
\newrulecommand{TypeAllE}{\Typerulename{$\AllSym$E}}
\newrulecommand{TypeDataI}{\Typerulename{DataI}}
\newrulecommand{TypeDataE}{\Typerulename{DataE}}
\newrulecommand{TypeDeclare}{\Typerulename{Declare}}
\newrulecommand{TypeDatatype}{\Typerulename{Datatype}}

\newcommand{\Chkrulename}[1]{\ensuremath{\rulename{Chk#1}}}
\newcommand{\Synrulename}[1]{\ensuremath{\rulename{Syn#1}}}

\newrulecommand{SynVar}{\Synrulename{Var}}
\newrulecommand{ChkSub}{\Chkrulename{Sub}}
\newrulecommand{SynAnno}{\Synrulename{Anno}}
\newrulecommand{ChkUnitI}{\Chkrulename{$\unitty$I}}
\newrulecommand{ChkArrI}{\Chkrulename{$\arr$I}}
\newrulecommand{SynArrE}{\Synrulename{$\arr$E}}
\newrulecommand{ChkProdI}{\Chkrulename{$*$I}}

\newrulecommand{ChkSectI}{\Chkrulename{$\sectty$I}}
\newrulecommandONE{SynSectE}{\Synrulename{$\sectty$E$_{#1}$}}

\newrulecommand{ChkDataI}{\Chkrulename{DataI}}
\newrulecommand{ChkDataE}{\Chkrulename{DataE}}
\newrulecommand{ChkDeclare}{\Chkrulename{Declare}}
\newrulecommand{ChkDatatype}{\Chkrulename{Datatype}}
\newcommand{\Subrulename}[1]{\ensuremath{\subtypesym\rulename{#1}}}

\newrulecommand{SubUnit}{\Subrulename{$\unitty$}}
\newrulecommand{SubArr}{\Subrulename{$\arr$}}
\newrulecommand{SubProd}{\Subrulename{$*$}}
\newrulecommand{SubSort}{\Subrulename{Data}}

\newrulecommand{SubAllL}{\Subrulename{$\forall$L}}
\newrulecommand{SubAllR}{\Subrulename{$\forall$R}}

\newrulecommandONE{SubSectL}{\Subrulename{$\sectty$L$_{#1}$}}
\newrulecommand{SubSectR}{\Subrulename{$\sectty$R}}

\newcommand{\matchsym}{\;\text{\normalfont\sf match}\;}

\newcommand{\vmatch}[3]{{#1}\matchsym{#2} \longrightarrow {#3}}
\newcommand{\vmatchPf}[4]{\Pf{}{#1}{\!\!\matchsym{#2} \longrightarrow {#3}}{#4}}

\newcommand{\novmatch}[2]{{#1}\matchsym{#2} \not\rightarrow\relax}
\newcommand{\novmatchPf}[3]{\Pf{}{}{{#1}\matchsym{#2} \not\rightarrow{}}{#3}}

\newcommand{\Ursig}{\hspace{0.15pt}\mathcal{U}}  %
\newcommand{\blk}{K}
\newcommand{\blkelem}{K_{\textsf{elem}}}

\newcommand{\sstrsym}{{\leq}_{\mathsf{sig}}}
\newcommand{\sstr}{\mathrel{\sstrsym}}
\newcommand{\sstrPf}[3]{\Pf{#1}{\sstrsym}{#2}{#3}}

\newcommand{\gstrsym}{{\leq}_{\mathsf{ctx}}}
\newcommand{\gstr}{\mathrel{\gstrsym}}

\newcommand{\trkstrsym}{{\leq}_{\mathsf{trk}}}
\newcommand{\trkstr}{\mathrel{\trkstrsym}}
\newcommand{\trkstrPf}[3]{\Pf{#1}{\trkstrsym}{#2}{#3}}

\ifnum\OPTIONAppendix=1

    \includecomment{includeapx}
    \excludecomment{excludeapx}
\else
    \excludecomment{apxproof}
    \excludecomment{includeapx}
    \includecomment{excludeapx}
\fi

\newcommand{\argle}{}
\newcommand{\bargle}{}
\newcommand{\arglebargle}{}

\newcommand{\lesscaptionspace}{\vspace*{-2ex}}

\newcommand{\noskipmypara}[1]{{\noindent \textit{#1}}\;}
\newcommand{\mypara}[1]{\smallskip {\noindent \textit{#1}}\;}

%% file: abstract.tex
\begin{abstract}
  Refinement types turn typechecking into lightweight verification.
  The classic form of refinement type is the datasort refinement,
  in which datasorts identify subclasses of inductive datatypes.
  
  Existing type systems for datasort refinements require that all the
  refinements of a type be specified when the type is declared;
  multiple refinements of the same type can be obtained only by duplicating
  type definitions, and consequently, duplicating code.

  We enrich the traditional notion of a signature, which describes the inhabitants
  of datasorts, to allow \emph{re-refinement} via signature extension,
  without duplicating definitions.  Since arbitrary updates to a signature can invalidate
  the inversion principles used to check case expressions, we develop
  a definition of signature well-formedness that ensures that extensions maintain
  existing inversion principles.  This definition allows different parts
  of a program to extend the same signature in different ways, without conflicting
  with each other.  Each part can be type-checked independently, allowing
  separate compilation.
\end{abstract}

%% file: fig-expressions.tex
\begin{figure}[t]
   \centering

  ~\begin{tabular}{@{\,}lr@{~~~}c@{~~~}ll@{}}
      Term vars. & \multicolumn{3}{l}{$x, y, \dots$}
      \\[0.2ex]
      Expressions\!   & $e$ & $\bnfas$&
      \multicolumn{2}{l}{\!\!%
              $x
              \bnfalt \Lam{x} e
              \bnfalt e_1\,e_2
              \bnfalt \Pair{e_1}{e_2}
              \bnfalt
              c(e)
              \bnfalt \Case{e}{ms}
              $}
       \\ &&&\!\!\!\!\!
              $
              \bnfalt \declare{\Sigma}{e}
              $
              ~~\text{\small---signature extension (Fig.\ \ref{fig:signatures})}
  \\[0.3em]
  Matches & \!\!\!\!$ms$ & $\bnfas$&
              $
              \cdot
              \bnfalt
              \big((\Match{p}{e}) \,\matchor\, ms\big)
              $
  \\[0.3em]
  Values   & $v$ & $\bnfas$&  
              $x
              \bnfalt \Lam{x} e
              \bnfalt \Pair{v_1}{v_2}
              \bnfalt c(v)
              \bnfalt \anno{v}{A}
              $
  \\[0.3em]
  Patterns    & $p$ & $\bnfas$&
             $\wildcard
             \bnfalt \emptypattern
             \bnfalt c(p)
             \bnfalt \Pair{p_1}{p_2}
             \bnfalt x \As p
             \bnfalt p_1 \patunion p_2
             $
  \end{tabular}

  \caption{Expressions}
  \FLabel{fig:expressions}
\end{figure}

%% file: fig-types.tex
\begin{figure}[t]
  \centering
  
  \begin{tabular}{@{}lr@{~~~}c@{~~~}ll@{}}
           Datasorts\!\!\! & $s, t, \dots$\hspace{-4ex}
      \\[0.2ex]
   Types   & \hspace{-5ex}$A, B, D$ & $\bnfas$&
              $\unitty
              \bnfalt A \arr B
              \bnfalt A * B %
              \bnfalt s
              \bnfalt A \sectty B
              $
  \\[0.5ex]
     Typing contexts  & $\Gamma$ & $\bnfas$&
             $\emptyctx
             \bnfalt \Gamma, x : A
             $
  \end{tabular}

  \caption{Types and contexts}
  \FLabel{fig:types}
\end{figure}

%% file: fig-signatures.tex
\begin{figure}[t]
  \centering
  
  \begin{grammar}
      Unrefined datatype names
        &
        $d$
        &
    \\[0.2ex]
      Unrefined types
        &
        $\tau$
        &
        $\bnfas$
        &
        $
        \unitty
        \bnfalt \tau_1 \arr \tau_2
        \bnfalt \tau_1 * \tau_2
        \bnfalt d
        $
    \\[0.2ex]
      Unrefined signatures
        &
        $\Ursig$
        &
        $\bnfas$
        &
        $
        \cdot
        \bnfalt \Ursig, d
        \bnfalt \Ursig,\, c : \tau \arr d$
  \end{grammar}
  
   \vspace{-4ex}

    \begin{grammar}
          Constructor types\!\!\!   & $C$ & $\bnfas$&  
                  $A \arr s
                  $
      \\[0.2em]
           Blocks  & $\blk$ & $\bnfas$&
                 $\emptyctx$
                 & empty block
       \\ &&&\!\!\!\!\!
                 $\bnfalt \blk, s_1 \subsort s_2$
                 & subsorting declaration
       \\ &&&\!\!\!\!\!
                 $\bnfalt \blk, c : C$
                 & constructor type decl.
    \\[0.2em]
        Sort sets &
           $S$ & $\bnfas$& $(s_1{\refines}d_1, \dots, s_n{\refines}d_n)$
           \hspace{-8ex}
           &
           \hspace{8ex}
    \\[0.2em]
        Abbrev.\ sort sets\!\!\!\!\!\!\!\! &
           $S$ & $\bnfas$& $(s_1, \dots, s_n)$ & %
    \\[0.2em]
           Signatures\!\!\!\!  & $\Sigma, \Omega$ & $\bnfas$&
                 $\emptyctx$
                 & empty signature
       \\ &&&
                   \!\!\!\!\!$\bnfalt \Sigma, S\sortblock{\blk}$
            & \text{datasort specification}
    \end{grammar}

  \judgbox{\Sigma \entails A \refines \tau}
       {Under signature $\Sigma$ (and unrefined signature $\Ursig$),
         \\
         type $A$ is a refinement of unrefined type $\tau$}

  \smallskip
  \small
  \begin{mathpar}
    \Infer{\RefinesUnit}
         {}
         {\Sigma \entails \unitty \refines \unitty}
    ~~~
    \Infer{\RefinesArr}
         {%
             \Sigma \entails A_1 \refines \tau_1
             \\
             \Sigma \entails A_2 \refines \tau_2
         }%
         {\Sigma \entails (A_1 \arr A_2) \refines (\tau_1 \arr \tau_2)}
    ~~~
    \Infer{\RefinesProd}
         {%
             \Sigma \entails A_1 \refines \tau_1
             \\
             \Sigma \entails A_2 \refines \tau_2
         }%
         {\Sigma \entails (A_1 * A_2) \refines (\tau_1 * \tau_2)}
    \and
    \Infer{\RefinesData}
         {(s{\refines}d) \in \Sigma}
         {\Sigma \entails s \refines d}
    \and
    \Infer{\RefinesSect}
         {
           \Sigma \entails A_1 \refines \tau
           \\
           \Sigma \entails A_2 \refines \tau
         }
         {\Sigma \entails (A_1 \sectty A_2) \refines \tau}
  \end{mathpar}

  \lesscaptionspace
  \caption{Unrefined types and signatures, refined signatures, $\refines$}
  \FLabel{fig:signatures}
\end{figure}

%% file: fig-type-wf.tex
\begin{figure}[t]

  \judgbox{\typejudg \Sigma A}%
        {Type $A$ is well-formed}
  \vspace*{-2.5ex}
  {\small
  \begin{mathpar}
      \Infer{\TypewfUnit}
          {}
          {\typejudg \Sigma \unitty}
      \and
      \Infer{\TypewfProd}
          {
              \typejudg \Sigma {A_1}
              \\
              \typejudg \Sigma {A_2}
          }
          {\typejudg \Sigma {A_1 * A_2}}
     \\
      \Infer{\TypewfArr}
          {\arrayenvl{
              \typejudg \Sigma {A_1}
              \\
              \typejudg \Sigma {A_2}
          }}
          {\typejudg \Sigma {A_1 \arr A_2}}
      ~~~~
      \Infer{\TypewfSect}
          {\arrayenvl{
              \typejudg \Sigma {A_1}
              \\
              \typejudg \Sigma {A_2}
            }
            ~~~~~
            \arrayenvl{
            \Sigma \entails A_1 \refines \tau
            \\
            \Sigma \entails A_2 \refines \tau
            }
          }
          {\typejudg \Sigma {A_1 \sectty A_2}}
     \vspace{-0.8ex}
     \\
      \Infer{\TypewfSort}
          {s \in S }
          {\typejudg {\Sigma, S\sortblock{\blk}, \Sigma'} {s}}
  \vspace*{-4.0ex}
  \end{mathpar}
  }

  \judgbox{\contypejudg{\Sigma}{c : C}}%
          {Under (prefix) context $\Sigma$,
            the typing
            $c : C$
            refines a typing in $\Ursig$
}
  \vspace*{-1.2ex}
  \begin{mathpar}
    \Infer{\ContypeArr}
        {
          \typejudg{\Sigma}{A}
          \\
          \typejudg{\Sigma}{s}
          \\
          (c : \tau \arr d) \in \Ursig          %
          \\
          \Sigma \entails (A \arr s) \refines (\tau \arr d)           %
        }
        {\contypejudg{\Sigma}{c : A \arr s}}
  \end{mathpar}

  \lesscaptionspace
  \caption{Type well-formedness}
  \FLabel{fig:type-wf}
\end{figure}

%% file: fig-typing.tex
\begin{figure}[t]
  \centering

  \judgbox{\Sigma \entails c : C}{Under signature $\Sigma$, \\ constructor $c$ has type $C$}
  \vspace{-6.5ex}
  \begin{mathpar}
    \hspace*{34ex}
     \Infer{\ConArr}
          {(c : A \arr s') \in \Sigma
           \\
           \Sigma \entails s' \subsort s}
          {\Sigma \entails c : A \arr s}
  \vspace{-2.0ex}
  \end{mathpar}

  \judgbox{\Sigma; \Gamma \entails e : A}{Under signature $\Sigma$ and context $\Gamma$,
        expression $e$ has type $A$}
  \vspace{-1.8ex}
  \begin{mathpar}
    \Infer{\TypeVar}
         { }
         {\Sigma; \Gamma, x : A, \Gamma' \entails x : A}
    \and
    \Infer{\TypeSub}
         {\Sigma; \Gamma \entails e : A
           \\
           \Sigma \entails A \subtype B
         }
         {\Sigma; \Gamma \entails e : B}
    \vspace{-0.5ex}
    \\
    \Infer{\TypeArrI}
         {\Sigma; \Gamma, x : A \entails e : B
         }
         {\Sigma; \Gamma \entails (\Lam{x} e) : (A \arr B)}
    ~~~~~~~
    \Infer{\TypeArrE}
         {
           \Sigma; \Gamma \entails e_1 : (A \arr B)
           \\
           \Sigma; \Gamma \entails e_2 : A
         }
         {\Sigma; \Gamma \entails e_1\,e_2 : B}
   \vspace{-0.5ex}
    \\
\centerruleplaceholder{
    \Infer{\TypeSectI}
         {\Sigma; \Gamma \entails v : A_1
          \\
          \Sigma; \Gamma \entails v : A_2
         }
         {\Sigma; \Gamma \entails v : (A_1 \sectty A_2)}
}
    {\TypeSectE{k} admissible
      \\ via \TypeSub + \SubSectL{k}}
   \vspace{-0.5ex}
    \\
    \centerruleplaceholder{
        \Infer{\TypeProdI}
             {\Sigma; \Gamma \entails e_1 : A_1
              \\
              \Sigma; \Gamma \entails e_2 : A_2
             }
             {\Sigma; \Gamma \entails \Pair{e_1}{e_2} : A_1 * A_2}
   }{
     elimination via \NoLinkTypeDataE
     \\
     with $\Match{\Pair{x_1}{x_2}}{\cdots}$
   }
   \vspace{-0.5ex}
     \\
     \Infer{\TypeUnitI}
             {}
             {\Sigma; \Gamma \entails \unit : \unitty}
     \and
     \Infer{\TypeDataI}
          {\Sigma \entails c : A \arr s
           \\
           \Sigma; \Gamma \entails e : A}
          {\Sigma; \Gamma \entails c(e) : s}
     \and
     \Infer{\TypeDataE}
         {
           \arrayenvbl{
           \Sigma; \Gamma \entails e : A
           \\
           \matchassn{\Sigma}{\Gamma}{A}{\wildcard}
           \entails
           ms : B
           }
         }
         {\Sigma; \Gamma \entails (\Case{e}{ms}) : B}
   \vspace{-0.5ex}
     \and
     \Infer{\TypeDeclare}
          {
            \arrayenvbl{
                \sigjudg{(\Sigma, \Sigma')}
            \\
                \Sigma \entails B \type
            }
                \\
                \Sigma, \Sigma'; \Gamma \entails e : B
          }
          {\Sigma; \Gamma \entails (\declare{\Sigma'}{e}) : B}
  \vspace{-3.0ex}
  \end{mathpar}

  \caption{Typing rules for constructors and expressions}
  \FLabel{fig:typing}
\end{figure}

%% file: fig-subtyping.tex
\begin{figure}[t]
  \centering
  
  \judgbox{\Sigma \entails A \subtype B}{$A$ is a subtype of $B$} %
  \vspace{-5.0ex}
  \begin{mathpar}
    \Infer{\SubUnit}
        {}
        {\Sigma \entails \unitty \subtype \unitty}
    \and
    \Infer{\SubProd}
        {
            \Sigma \entails A_1 \subtype B_1
            \\
            \Sigma \entails A_2 \subtype B_2
        }
        {\Sigma \entails (A_1 * A_2) \subtype (B_1 * B_2)}
    \and
    \arrayenvbl{
    \Infer{\SubSort}
        {\Sigma \entails s \subsort t}
        {\Sigma \entails s \subtype t}
        \\[2pt]
        ~
    }
    \vspace*{-2.5ex}
    \\
    \Infer{\!\SubArr}
        {
            \Sigma \entails B_1 \subtype A_1
            ~~~~
            \Sigma \entails A_2 \subtype B_2
        }
        {\Sigma \entails (A_1 \arr A_2) \subtype (B_1 \arr B_2)}
    ~
    \Infer{\!\SubSectL{k}}
        {\Sigma \entails A_k \subtype B}
        {\Sigma \entails (A_1 \sectty A_2) \subtype B}
    ~~~
    \Infer{\!\SubSectR}
        {
          \arrayenvbl{
            \Sigma \entails A \subtype B_1
            \\
            \Sigma \entails A \subtype B_2
          }
        }
        {\Sigma \entails A \subtype (B_1 \sectty B_2)}
  \vspace{-3.2ex}
  \end{mathpar}

  \caption{Subtyping}
  \FLabel{fig:subtyping}
\end{figure}

%% file: fig-wf.tex
\begin{figure}[htb]
  \centering

    \judgboxa{\dom{\Sigma}}
             {Domain (declared sorts) of %
               $\Sigma$:
             \hspace{1ex}
             $\dom{S_1\sortblock{\blk_1}, \dots, S_n\sortblock{\blk_n}}
               \,=\,
               S_1 \union \cdots \union S_n$
             }
    ~\\[0.5ex]

  \judgbox{\Sigma \entails s_1 \subsort s_2}%
        {Sort $s_1$ is a subsort of $s_2$}
  \vspace{-6.0ex}
      \begin{mathpar}
        \hspace*{0.47\textwidth}
        \Infer{\SubsortAssum}
            {}
            {\Sigma_L, S\sortblock{\dots, s_1 \subsort s_2, \dots}, \Sigma_R \entails s_1 \subsort s_2}
        \vspace{-2.5ex}
        \\
        \Infer{\SubsortRefl}
            {s \in \dom{\Sigma}}
            {
              \arrayenvbl{
              \Sigma \entails s \subsort s
              \\
              ~}
            }
        \and
        \Infer{\SubsortTrans}
            {\Sigma \entails s_1 \subsort s_2
             \\
             \Sigma \entails s_2 \subsort s_3}
            {\Sigma \entails s_1 \subsort s_3}
      \end{mathpar}%
    \vspace{-2.0ex}
  \judgbox{\safeextat{\Sigma}{S\sortblock{\blk}}{C}{c}{t}}
         {$C$ safely extends a type
           given by $\Sigma$ for $c$ %
         }
         {
           \vspace{-1ex}
           \hspace{3ex}
          \begin{mathpar}
            \Infer{\SafeConAt}
                {
                    (c : A' \arr s') \in \Sigma
                    \\ 
                    \arrayenvbl{
                        \Sigma, S\sortblock{\blk} \entails s' \subsort t
                        \\
                        {\Sigma, S\sortblock{\blk} \entails s \subsort s'}
                    }
                    \\
                    \Sigma, S\sortblock{\blk} \entails A \subtype A'
                }
                {
                    \safeextat{\Sigma}{S\sortblock{\blk}}{A \arr s}{c}{t}
                }
           \hspace{-4ex}
          \end{mathpar}
         }

  \vspace{-2.0ex}

  \vspace{-1.5ex}

  \judgbox{\blkjudg{\Sigma}{S}{\blk}{\blkelem}}
       {
         $\blkelem \in \blk$ is safe for $\Sigma; S\sortblock{\blk}$
       }
  \vspace{-2.3ex}
  \begin{mathpar}
  \hspace{0ex}
     \Infer{\BlockSubsort}
         {
             s_1, s_2 \in (\dom{\Sigma} \union S)
         }
         {\arrayenvbl{\blkjudg \Sigma S {\blk} {s_1 \subsort s_2}
           \\ ~}}
    ~~~~~~
    \Infer{\BlockCon}
          {
             s \in S
           ~~~
           \arrayenvbl{
             \contypejudg{\Sigma, S\sortblock{\blk}}{c : A \arr s}
             \\
             \text{for all~}t \in \dom{\Sigma}
                  \\
                  \text{such that~}\Sigma, S\sortblock{\blk}
                  \entails s \subsort t,
                      \\
                      ~~\safeextat{\Sigma}{S\sortblock{\blk}}{A \arr s}{c}{t}
             }
        }
        {\blkjudg{\Sigma}{S}{\blk}{c : A\,{\arr}\,s}}
  \vspace{-2.0ex}
  \end{mathpar}

  \judgbox{\sigjudg{\Sigma}}
       {
         Signature $\Sigma$ is well-formed
       }
  \vspace{-7.0ex}
  \begin{mathpar}
    \hspace*{0.5\textwidth}
    \Infer{\SigEmpty}
        {}
        {\sigjudg{\emptyctx}}
   \vspace{-1.2ex}
    \\
    \Infer{\SigBlock}
        {
           \arrayenvbl{
             \sigjudg{\Sigma}
             \\
             (S \sect  \dom{\Sigma}) = \emptyset
           }
           \\
           \arrayenvbl{
             \text{for all~}t_1, t_2 \in \dom{\Sigma},
             \\
             ~~
             \arrayenvbl{
             \big(\Sigma \entails t_1 \subsort t_2\big)
             \\ \text{iff~}
             \big(\Sigma, S\sortblock{\blk} \entails t_1 \subsort t_2\big)
             }
            }
           \\
           \arrayenvbl{\text{for all~}\blkelem \in \blk,
              \\ ~~ \blkjudg{\Sigma}{S}{\blk}{\blkelem}
          }
        }
        {\sigjudg{\Sigma, S\sortblock{\blk}}}
  \vspace{-4ex}
  \end{mathpar}

  \caption{Signature well-formedness and subsorting}
  \FLabel{fig:wf}
\end{figure}

%% file: fig-pattern-type.tex
\begin{figure}[t]
  \centering

  \judgbox{\Ursig \entails p : \tau}
          {Pattern $p$ is suitable for values of unrefined type $\tau$}
  \vspace{-1.7ex}
  \small  
  \begin{mathpar}
    \Infer{\PattypeWild}
         {}
         {\Ursig \entails \wildcard : \tau}
    \and
    \Infer{\PattypeAs}
         {\Ursig \entails p : \tau}
         {\Ursig \entails (x \As p) : \tau}
    \and
    \Infer{\PattypeEmpty}
         {}
         {\Ursig \entails \emptypattern : \tau}
    \and
    \Infer{\PattypeUnit}
         {
         }
         {\Ursig \entails \unit : \unitty}
    \\
    \Infer{\!\PattypeOr}
         {
           \arrayenvbl{
           \Ursig \entails p_1 : \tau
           \\
           \Ursig \entails p_2 : \tau
           }
         }
         {\Ursig \entails (p_1{\hspace{0.7pt}\patunion\,}p_2) : \tau}
    ~~
    \Infer{\!\PattypePair}
         {
           \arrayenvbl{
           \Ursig \entails p_1 : \tau_1
           \\
           \Ursig \entails p_2 : \tau_2
           }
         }
         {\Ursig \entails \Pair{p_1}{p_2} : \tau_1 * \tau_2}
     ~~
    \Infer{\!\PattypeCon}
         {
               \Ursig \entails c : (\tau \arr d)
           \\
               \Ursig \entails p : \tau
         }
         {\Ursig \entails c(p) : d}
  \vspace{-3.2ex}
  \end{mathpar}

  \caption{Pattern type rules}
  \FLabel{fig:pattern-type}
\end{figure}

%% file: fig-pattern-typing.tex
\begin{figure*}[t]

  \judgbox{\matchassn{\Sigma}{\Gamma}{A}{p} \entails ms : D}
          {
            For a scrutinee of type $A$
            that matches
            residual pattern $p$,
            \\
            check each match in $ms$ against $D$
          }
  \begin{mathpar}
      \Infer{\!\TypeMs}
           {
             \arrayenvbl{
               \Sigma \entails A \refines \tau
               \\
               \Ursig \entails p_1 : \tau
             }
             \\
             \arrayenvbl{
                 \text{for all $\intresult{\Gamma'}{B}$}
                 \\
                 ~~\text{$\in \intersect{\Sigma}{A}{p \patsect p_1}$:}
                 \\
                 ~~~~\Sigma; \Gamma, \Gamma'
                              \entails
                              e_1 : D
             }
             \\
             \matchassn \Sigma \Gamma A {(p \patsect \patcomplement p_1)}
             \entails
             ms : D
           }
           {
             \matchassn \Sigma \Gamma A p
             \entails
             \big( (\Match{p_1}{e_1}) \matchor ms\big) : D
           }
      \and
      \Infer{\!\TypeMsEmpty}
           {
             \intersect{\Sigma}{A}{p} = \emptyset
           }
           {
             \matchassn \Sigma \Gamma A p
             \entails
             \emptyms : D
           }
  \vspace*{-3ex}
  \end{mathpar}

  \caption{Match typing}
  \FLabel{fig:pattern-typing}
\end{figure*}

%% file: fig-patterns.tex
\begin{figure}[t]

  \judgboxa{\intersect{\Sigma}{A}{p} = \vec{\Bdec}}%
          {Intersection of type $A$ with
               pattern $p$ \\
               ~where each $\Bdec$ %
               has the form
               $\intrsctresult{\Sigma'}{\Gamma'}{B'}$}
  \vspace*{-1.1ex}
  \[%
  \begin{array}[t]{@{}r@{~\,}c@{~\,}l@{}ll}
    \intersect{\Sigma}{A}{\wildcard}
        &=&
        \big\{
            \intrsctresult{\cdot}{\cdot}{A}
        \big\}
    \\[0.3ex]
    \intersect{\Sigma}{A}{\emptypattern}
        &=&
        \emptyset
    \\[0.3ex]
    \intersect{\Sigma}{A}{x \As p}
        &=&
        \big\{
            \intrsctresult%
                   {\Sigma'}
                   {\Gamma', x : B}
                   {B}
            \,\big|\;
              \intrsctresult{\Sigma'}{\Gamma'}{B}
              \in
              \intersect{\Sigma}{A}{p}
        \big\}
    \\[0.7ex]
    \intersect{\Sigma}{A}{p_1 \patunion p_2}
        &=&
        \intersect{\Sigma}{A}{p_1}
        \union
        \intersect{\Sigma}{A}{p_2}
   \\[0.4ex]
    \intersect{\Sigma}{A_1\,{*}\,A_2}{\Pair{p_1}{p_2}}
    &=& 
        \big\{
             \intrsctresult%
                   {\Sigma_1, \Sigma_2}
                   {\Gamma_1, \Gamma_2}
                   {B_1 * B_2}
                 ~\big|\!
                 \arrayenvl{
                     ~%
                     \intrsctresult{\Sigma_1}{\Gamma_1}{B_1}
                     \in
                     \intersect{\Sigma}{A_1}{p_1}
                     \\[0.2ex]
                     \hspace*{-5.5ex}
                 \AND
                     \intrsctresult{\Sigma_2}{\Gamma_2}{B_2}
                     \in
                     \intersect{\Sigma}{A_2}{p_2}
        \big\}
             }
   \\[1ex]
    \intersect{\Sigma}{s}{c(p_0)}
        &=&
        \big\{
        \intrsctresult{%
            \Sigma'
            }
                        {\Gamma'}
                        {s_c}
       ~\big|~
           \begin{array}[t]{@{}l@{}}
               (c : A_c \arr s_c) \in \Sigma
               ~\AND~
               \Sigma \entails s_c \subsort s
               \\[0.4ex]
               \!\AND~
                 \intrsctresult{\Sigma'}{\Gamma'}{B}
                 \in
                 \intersect{\Sigma}{A_c}{p_0}
         \big\}
         \end{array}
  \end{array}
  \vspace*{-3ex}
  \]
  \caption{Intersection of a type with a pattern}
  \FLabel{fig:type-pattern-intersection}
\end{figure}

%% file: meta.tex
\section{Metatheory}
\Label{sec:meta}

This section gives definitions, states some lemmas and theorems, and discusses
their significance in proving our main results.
For space reasons, we summarize a number of lemmas; their full statements
appear in the appendix.
All proofs are also relegated to the appendix.

\renewcommand{\Hand}{~}

\mypara{Subtyping and subsorting.}
Subtyping is reflexive and transitive
(Lemmas \ref{lem:subtyping-reflexivity}--\ref{lem:subtyping-transitivity}).

We define what it means for signature extension to preserve subsorting:

\begin{definition}[Preserving subsorting]
\Label{def:preserving}
  Given $\Sigma_1$ and $\Sigma_2$, we say that
  $\Sigma_2$ \emph{preserves subsorting} of $\Sigma_1$
  iff for all sorts $s, t \in \dom{\Sigma_1}$,
  if $\Sigma_1, \Sigma_2 \entails s \subsort t$
  then $\Sigma_1 \entails s \subsort t$.
\end{definition}

This definition allows new sorts in $\dom{\Sigma_2}$ to be subsorts or supersorts
of the old sorts in $\dom{\Sigma_1}$, provided that the subsort relation between
the old sorts doesn't change.

If two signatures do not have subsortings that cross into each other's domain,
they are \emph{non-adjacent}; non-adjacent signatures preserve subsorting.

\begin{definition}[Non-adjacency]
\Label{def:non-adjacency}
  Two signatures $\Sigma_1$ and $\Sigma_2$ are \emph{non-adjacent}
  iff each signature contains no subsortings of the form
  $s_1 \subsort s_2$ or $s_2 \subsort s_1$, where $s_1 \in \dom{\Sigma_1}$
  and $s_2 \in \dom{\Sigma_2}$.
\end{definition}

\vspace*{-1.0ex}

\begin{restatable}[Non-adjacent preservation]{theorem}{thmnonadjacentpreservation}
\XLabel{thm:non-adjacent-preservation}
~\\
  If $\Sigma_2$ preserves subsorting of $\Sigma_1$
  and $\Sigma_3$ preserves subsorting of $\Sigma_1$
  \\
  and $\Sigma_2$ and $\Sigma_3$ are non-adjacent
  then $\Sigma_3$ preserves subsorting of $(\Sigma_1, \Sigma_2)$.
\end{restatable}

\vspace*{-1.0ex}

\mypara{Strengthening, weakening, and substitution.}
\Theoremref{thm:weakening} (Weakening)
will allow the assumptions in a judgment to be changed in two ways:
(1) the signature may be strengthened by replacing a signature $(\Sigma, \Sigma')$
with a signature $(\Sigma, \Omega, \Sigma')$;
and (2) the context may be strengthened by replacing $\Gamma$
with a context $\Gamma^+$ in which any typing assumption $(x : A) \in \Gamma$
can be replaced with $(x : A^+) \in \Gamma$, if $A \subtype A^+$.

Repeatedly applying (1) with different $\Omega$ leads to a more general notion
of strengthening a signature:

\begin{definition}
\Label{def:signature-str}
  A signature $\Sigma'$ is \emph{stronger than} $\Sigma$,
  written $\Sigma' \sstr \Sigma$, if
  $\Sigma'$ can be obtained from $\Sigma$ by inserting entire signatures at any position
  in $\Sigma$.
\end{definition}

We often use the less general notion (inserting a single $\Omega$), which
simplifies proofs.  For any result stated less generally, however, the
more general strengthening of \Definitionref{def:signature-str}
can be shown by induction on the number of blocks inserted.

\begin{definition}
\Label{def:context-str}
  Under $\Sigma$, a context $\Gamma'$ is \emph{stronger than} $\Gamma$,
  written $\Sigma \entails \Gamma' \gstr \Gamma$,
  if for each $(x : A') \in \Gamma'$, there exists $(x : A) \in \Gamma$ such that
  $\Sigma \entails A' \subtype A$.
\end{definition}

Several lemmas show weakening.
Lemma \ref{lem:weakening-low} says that $\Sigma$ in $\Sigma \entails \mathcal{J}$
can be replaced by a stronger $\Sigma'$, where $\mathcal{J}$ has the form
$A \type$ or $s_1 \subsort s_2$ or $A \subtype B$ or $c : A \arr s$ or $A \refines \tau$
or $c : C$.
Lemma \ref{lem:weakening-safeconat}
  says that $(\Sigma, \Omega, \Sigma')$ can replace $(\Sigma, \Sigma')$
  in
  $
  \safeextat{\Sigma, \Sigma'}{S\sortblock{\blk}}{A \arr s}{c}{t}
  $.
 Lemma \ref{lem:supersort-safeconat}
  allows the sort $t'$ in
  the judgment
  $\safeextat{\Sigma}{S\sortblock{\blk}}{A \arr s}{c}{t'}$
  to be replaced by a supersort $t$.

Using the above lemmas and \Theoremref{thm:non-adjacent-preservation},
we can show that the key judgment
``$\cdots c : A{\arr}s \textit{~safe}$''
can be weakened by inserting $\Omega$ inside the signature:

\begin{restatable}[Weakening `safe']{theorem}{thmweakeningsafe} %
\XLabel{thm:weakening-safe}
~\\
    If
    $\sigjudg{(\Sigma, \Sigma')}$ and $\sigjudg{(\Sigma, \Omega)}$
    and
    $\disjoint{\dom{\Sigma'}}{\dom{\Omega}}$
    \\
    and
    $\disjoint{\dom{\Sigma, \Omega, \Sigma'}}{S}$
    and $\blk$ does not mention anything in $\dom{\Omega}$
    \\
    and
    $S\sortblock{\blk}$ preserves subsorting for $(\Sigma, \Sigma')$
    \\
    and
    $(c : A \arr s) \in \blk$
    and
    $\blkjudg{\Sigma, \Sigma'}{S}{\blk}{c : A \arr s}$
    \\
    then
    $\blkjudg{\Sigma,  \Omega, \Sigma'}{S}{\blk}{c : A \arr s}$.
\end{restatable}

With this additional lemma, we have weakening for the judgments involved
in checking that a signature is well-formed, so we can show that
if $\Sigma$ is safely extended by $\Sigma'$ and separately by $\Omega$,
then $\Omega$ and $\Sigma'$, together, safely extend $\Sigma$.

\begin{restatable}[Signature Interleaving]{theorem}{thmsignatureinterleaving}
\XLabel{thm:signature-interleaving}
~\\
  If $\sigjudg{(\Sigma, \Sigma')}$
  and $\sigjudg{(\Sigma, \Omega)}$
  and $\disjoint{\dom{\Sigma'}}{\dom{\Omega}}$
  then
  $\sigjudg{(\Sigma, \Omega, \Sigma')}$.
\end{restatable}

Ultimately, we will show type preservation; in the preservation case for
the \TypeDeclare rule, we extend the signature in a premise.
We therefore need to show that the typing judgment can be weakened.
Since the typing rules for matches involve the $\intersectx$ function,
we need to show that a stronger input to $\intersectx$
yields a stronger output; that is, a longer (stronger) signature 
yields a stronger type $B_+$ (a subtype of $B$)
and a stronger context $\Gamma_+$ typing $\xAs$-variables.

\begin{definition}
\Label{def:track-str}
  Under a signature $\Sigma$,
  a track $\intresult{\Gamma_+}{B_+}$
  is stronger than  $\intresult{\Gamma}{B}$,
  written
$%
  \Sigma
  \entails
  \intresult{\Gamma_+}{B_+}
  \trkstr
  \intresult{\Gamma}{B}
$, %
  if and only if
  $\Sigma \entails \Gamma_+ \gstr \Gamma$
  and
  $\Sigma \entails B_+ \subtype B$.

  \smallskip

  A set of tracks $\vec{\Bdec_+}$ is stronger than $\vec{\Bdec}$,
  written $\vec{\Bdec_+} \trkstr \vec{\Bdec}$,
  if and only if, for each track
  $\intresult{\Gamma_+}{B_+} \in \vec{\Bdec_+}$,
  there exists a track
  $\intresult{\Gamma}{B} \in \vec{\Bdec}$
  such that
  $\intresult{\Gamma_+}{B_+}
  \trkstr
  \intresult{\Gamma}{B} \in \vec{\Bdec}$.
\end{definition}

Lemma \ref{lem:intersect-strength}
says that the result of $\intersectx$ on a stronger signature is stronger.

We can then show that weakening holds for the typing judgment itself,
along with substitution typing (defined in the appendix) and match typing.

\begin{restatable}[Weakening]{theorem}{thmweakening}
\XLabel{thm:weakening}
~\\
    If $\sigjudg{(\Sigma, \Sigma')}$,
    $\sigjudg{(\Sigma, \Omega)}$,
    $\dom{\Sigma'} \sect \dom{\Omega} = \emptyset$
    and
    $\Sigma, \Omega, \Sigma' \entails \Gamma^+ \gstr \Gamma$
    then
   \vspace*{-1.0ex}

    \begin{enumerate}[(1)]
      \item[(1)]
          If $\Sigma, \Sigma'; \Gamma \entails e : A$
          then $\Sigma, \Omega, \Sigma'; \Gamma^+ \entails e : A$.
      \item[(2)]
          If $\Sigma, \Sigma'; \Gamma \entails \theta : \Gamma'$
          then $\Sigma, \Omega, \Sigma'; \Gamma^+ \entails \theta : \Gamma'$.
      \item[(3)]
          If $\matchassn{\Sigma, \Sigma'}{\Gamma}{A}{p} \entails ms : D$
          then $\matchassn{\Sigma, \Omega, \Sigma'}{\Gamma^+}{A}{p} \entails ms : D$.
    \end{enumerate}
\end{restatable}

\vspace{-1.5ex}

\mypara{Properties of values.}
Substitution properties (Lemmas \ref{lem:value-substitution} and \ref{lem:multiple-value-substitution})
and inversion (or \emph{canonical forms}) properties (Lemma \ref{lem:inversion}) hold.

\mypara{Type preservation and progress.}
\Label{sec:type-preservation-and-progress}
The last important piece needed for type preservation is that
$\intersectx$ does what it says: if a value $v$
matches $p$, then $v$ has type $B$ where $B$ is one of the outputs
of $\intersectx$.

\begin{restatable}[Intersect]{theorem}{thmintersect}
\XLabel{thm:intersect}
\!\!If  $\sigjudg{\Sigma}$
  and $\Sigma; \cdot \entails v : A$
  and $\Sigma \entails A \type$
  and $\vmatch{p}{v}{\theta}$
  and $\intersect{\Sigma}{A}{p} = \vec{\Bdec}$
  then
  there exists $\intrsctresult{\Sigma'}{\Gamma'}{B} \in \vec{\Bdec}$
  s.t.\ 
  $\Sigma%
  ; \cdot \entails v : B$
\\
  and
  $\Sigma%
  ; \cdot \entails \theta : \Gamma'$
  where
  $\Sigma%
  \entails B \type$
  and $\Sigma \entails B \subtype A$.
\end{restatable}

The preservation result allows for a longer signature, to model entering the scope of
a $\xdeclare$ expression or the arms of a \textkw{match}.
We implicitly assume that, in the given typing derivation, all types
are well-formed under the local signature: for any
subderivation of $\Sigma; \Gamma \entails e' : B$, it is the case that
$\Sigma \entails B \type$.

\begin{restatable}[Preservation]{theorem}{thmpreservation}
\XLabel{thm:preservation}
  If $\sigjudg{\Sigma}$
  and $\Sigma; \cdot \entails e : A$
  and $e \step e'$
  \\
  then there exists $\Sigma'$ such that
  $\Sigma, \Sigma' \entails e' : A$
  where $\sigjudg{(\Sigma, \Sigma')}$.
\end{restatable}

\vspace{-2ex}

\begin{restatable}[Progress]{theorem}{thmprogress}
\XLabel{thm:progress}
 ~\\
  If $\sigjudg{\Sigma}$ and $\Sigma; \cdot \entails e : A$
  then $e$ is a value or there exists $e'$ such that $e \step e'$.
\end{restatable}

\vspace{-2.5ex}

%% file: bidir-meta.tex
\renewcommand{\Hand}{~}

\section{Bidirectional Typing}
\Label{sec:bidir}

The type assignment system in \Figureref{fig:typing} is not syntax-directed,
because
the rules \TypeSub and \TypeSectI apply to any shape of expression.
Nor is the system directed by the syntax of types: rule \TypeSub can conclude
$e : B$ for any type $B$ that is a supertype of some other type $A$.
Finally, while the choice to apply rule \TypeDataI is guided by the shape
of the expression---it must be a constructor application $c(e)$---the resulting
sort is not uniquely determined, since the signature can have multiple constructor
typings for $c$.

Fortunately, obtaining an algorithmic system is straightforward,
following previous work with datasort refinements and intersection types.
We follow the bidirectional typing recipe
of
\citet{Davies00icfpIntersectionEffects,DaviesThesis,Dunfield04:Tridirectional}:
\begin{enumerate}
\item Split the typing judgment into checking $\Sigma; \Gamma \entails e \chk A$
and synthesis $\Sigma; \Gamma \entails e \syn A$ judgments.  In the checking judgment,
the type $A$ is input (it might be given via type annotation);
in the synthesis judgment, the type $A$ is output.

\item Allow change of direction: Change the subsumption rule to
  synthesize a type, then check if it is a subtype of a type being checked against;
  add an annotation rule that checks $e$ against $A$ in the annotated expression $(e : A)$.

\item In each introduction rule, \eg \TypeArrI, make the conclusion a checking judgment;
     in each elimination rule, \eg \TypeDataE, make the premise
     that contains the eliminated connective a synthesis judgment.  

\item Make the other judgments in the rules either checking or synthesizing,
  according to what information is available.  For example, the premise of
  \TypeArrI becomes a checking judgment, because we know $B$ from
  the conclusion.

 \item Since the subsumption rule cannot synthesize, add rules such as \SynSectE{1},
   which were admissible in the type assignment system.
\end{enumerate}
This yields the rules in \Figureref{fig:bidir}.  (Rules for the match typing
judgment $\matchassn{\Sigma}{\Gamma}{A}{p} \entails ms \chk B$
can be obtained from \Figureref{fig:pattern-typing}
by replacing ``:'' in ``$e_1 : D$'' and ``$ms : D$'' with ``$\chk$''.)
While this system is much more algorithmic than \Figureref{fig:typing},
the presence of intersection types requires backtracking: if we apply a function of
type $(\Even \arr \Odd) \sectty (\Odd \arr \Even)$, we need to synthesize
$\Even \arr \Odd$ first; if we subsequently fail (\eg if the argument has type $\Odd$),
we backtrack and try
$\Odd \arr \Even$.  Similarly, if the signature contains several typings for a constructor
$c$, we may need to try rule \ChkDataI with each typing.

\input{fig-bidir.tex}

Type-checking for this system is almost certainly PSPACE-complete~\citep{Reynolds96:Forsythe};
however, the experience of \citet{DaviesThesis} shows that a similar system, differing
primarily in whether the signature can be extended, is practical if certain techniques,
chiefly memoization, are used.

Using these rules, annotations are required exactly on
(1) the entire program $e$ (if $e$ is a checked form, such as a $\lambda$)
and (2) expressions not in normal
form, such as a $\lambda$ immediately applied to an argument, a recursive
function declaration, or a let-binding (assuming the rule for \textkw{let}
synthesizes a type for the bound expression).
Rules with ``more synthesis''---such as a synthesizing version of \TypeProdI---could be added along the lines of previous bidirectional type systems~\citep{XiThesis,Dunfield13}.

Following \citet{DaviesThesis}, an annotation can list several types $\vec{A}$.
Rule \SynAnno chooses one of these, backtracking if necessary.
Multiple types may be needed if a $\lambda$-term is checked against intersection type:
when checking $(\Lam{x} x)$ against 
$(\Even \arr \Even) \sectty (\Odd \arr \Odd)$,
the type of $x$ will be $\Even$ inside the left subderivation of \ChkSectI, but $\Odd$ inside the
right subderivation.  
Thus, if we annotate $x$ with $\Even$, the check against $\Odd \arr \Odd$ fails;
if we annotate $x$ with $\Odd$, the check against $\Even \arr \Even$ fails.
For a less contrived example, and for a variant annotation form %
that reduces backtracking, see \citet{Dunfield04:Tridirectional}.

In the appendix, we prove that our bidirectional system is sound and complete with
respect to our type assignment system:

\begin{restatable}[Bidirectional soundness]{theorem}{thmbidirsoundness}
\XLabel{thm:bidir-soundness}
  If
  $\Gamma \entails e \chk A$
  or 
  $\Gamma \entails e \syn A$
  then
  $\Gamma \entails |e| : A$
  where $|e|$ is $e$ with all annotations erased.
\end{restatable}

\vspace*{-1.5ex}

\begin{restatable}[Annotatability]{theorem}{thmannotatability}
\XLabel{thm:annotatability}
  If
  $\Gamma \entails e : A$
  then:
  \vspace{-1.2ex}
  \begin{enumerate}[(1)]
  \item 
    There exists $e_{\chk}$ such that $|e_{\chk}| = e$
    and $\Gamma \entails e_{\chk} \chk A$.
  \item 
    There exists $e_{\syn}$ such that $|e_{\syn}| = e$
    and $\Gamma \entails e_{\syn} \syn A$.
  \end{enumerate}
\end{restatable}

We also prove that the $\syn$ and $\chk$ judgments are decidable
(\Theoremref{thm:decidability}).

\vspace*{-2.5ex}

%% file: fig-bidir.tex
\begin{figure}[t]
  \centering

  \judgbox{
       \normalsize\arrayenvbl{
         \Sigma; \Gamma \entails e \chk A
         ~~~~
         \Sigma; \Gamma \entails e \syn A
       }
     }
      {%
        Expr.\ $e$ checks against ($\chk$) / synthesizes ($\syn$) type $A$
      ~}
  \vspace{-1.3ex}
  \begin{mathpar}
    \Infer{\!\SynVar}
         { }
         {\Sigma; \Gamma, x : A, \Gamma' \entails x \syn A}
    ~~
    \Infer{\!\ChkSub}
         {\arrayenvbl{\Sigma; \Gamma \entails e \syn A
             \\
           \Sigma \entails A \subtype B
           }
         }
         {\Sigma; \Gamma \entails e \chk B}
    ~~
     \Infer{\!\SynAnno}
          {%
              A \in \Avec
              \\
              \Sigma; \Gamma \entails e \chk A
            }
          {\Sigma; \Gamma \entails \anno{e}{\Avec} \syn A}
    \\
    \Infer{\!\ChkArrI}
         {\Sigma; \Gamma, x : A \entails e \chk B
         }
         {\Sigma; \Gamma \entails (\Lam{x} e) \chk (A \arr B)}
    ~~~~
    \Infer{\!\SynArrE}
         {%
           \Sigma; \Gamma \entails e_1 \syn (A \arr B)
           \\
           \Sigma; \Gamma \entails e_2 \chk A
         }
         {\Sigma; \Gamma \entails e_1\,e_2 \syn B}
    \\
    \Infer{\!\ChkSectI}
         {%
             \Sigma; \Gamma \entails v \chk A_1
             \\
             \Sigma; \Gamma \entails v \chk A_2
           }
         {\Sigma; \Gamma \entails v \chk (A_1 \sectty A_2)}
    ~~~~
     \Infer{\!\SynSectE{k}}
          {\Sigma; \Gamma \entails e \syn (A_1 \sectty A_2)
            \\
            k \in \{1, 2\}
          }
          {\Sigma; \Gamma \entails e \syn A_k}
    \\
        \Infer{\!\ChkProdI}
             {\Sigma; \Gamma \entails e_1 \chk A_1
              \!\!\\\!\!
              \Sigma; \Gamma \entails e_2 \chk A_2
             }
             {\Sigma; \Gamma \entails \Pair{e_1}{e_2} \chk A_1 * A_2}
     ~~~~
     \Infer{\!\ChkUnitI}
             {}
             {\Sigma; \Gamma \entails \unit \chk \unitty}
     \\
     \Infer{\!\ChkDataI}
          {
              \Sigma \entails c : A \arr s
              ~~~
              \Sigma; \Gamma \entails e \chk A
          }
          {\Sigma; \Gamma \entails c(e) \chk s}
     ~~
     \Infer{\!\ChkDataE}
         {
             \Sigma; \Gamma \entails e \syn A
             ~~~~
             \matchassn{\Sigma}{\Gamma}{A}{\wildcard}
             \entails
             ms \chk B
         }
         {\Sigma; \Gamma \entails (\Case{e}{ms}) \chk B}
     \\
     \Infer{\ChkDeclare}
          {
                \sigjudg{(\Sigma, \Sigma')}
            \\
                \Sigma \entails B \type
                \\
                \Sigma, \Sigma'; \Gamma \entails e \chk B
          }
          {\Sigma; \Gamma \entails (\declare{\Sigma'}{e}) \chk B}
  \vspace{-2.3ex}
  \end{mathpar}
  
\caption{Bidirectional typing rules}
\FLabel{fig:bidir}
\end{figure}

%% file: apxdefs.tex
\section{Omitted Definitions}
\Label{apx:defs}

\subsection{Pattern complement and intersection}

\input{fig-complement-intersection.tex}

\clearpage
\subsection{Evaluation contexts, matching, and stepping}

\input{fig-step.tex}

\clearpage

\subsection{Substitution typing}

\input{fig-subst-typing.tex}

Substitution typing (\Figureref{fig:subst-typing})
is used to state \Lemmaref{lem:multiple-value-substitution}.

%% file: fig-complement-intersection.tex
\begin{figure}[h]
  \judgboxa{\patcomplement p}{Complement of a pattern under unrefined signature $\Ursig$}
  \vspace{-3.0ex}
  \[
  \begin{array}[t]{r@{~~}c@{~~}lll}
      \patcomplement \wildcard &=& \emptypattern
  \\
      \patcomplement \emptypattern &=& \wildcard
  \\
      \patcomplement (x \As p)
         &=& 
         \patcomplement p
   \\
      \patcomplement \Pair{p_1}{p_2} 
          &=&
          \Pair{\patcomplement p_1}{\wildcard}
          \patunion
          \Pair{\wildcard}{\patcomplement p_2}
  \end{array}
  ~~
  \begin{array}[t]{r@{~~}c@{~~}lll}
       \patcomplement (p_1 \patunion p_2)
         &=&
         (\patcomplement p_1) \patsect (\patcomplement p_2)
  \\
      \patcomplement c(p)
         &=& 
         c(\patcomplement p)
         \patunion
         c_1(\wildcard) \patunion \dots \patunion c_n(\wildcard)
      \\
      \multicolumn{3}{r}{\text{where $\constructors{\Ursig} \setminus \{c\}  =  \{c_1, \dots, c_n\}$}}
  \end{array}
  \vspace*{-0.5ex}
  \]

  \judgboxa{p_1 \patsect p_2}{Intersection of patterns}
  \vspace*{-1.5ex}
  \[
  \begin{array}[t]{@{}r@{~}c@{~}lll@{}}
      \emptypattern \patsect p
        ~=~
        p \patsect \emptypattern
        &=&
        \emptypattern
  \\
      \wildcard \patsect p
        ~=~
        p \patsect \wildcard
        &=& p
  \\
      c(p_1) \patsect c(p_2)
          &=&
          c(p_1 \patsect p_2)
  \\
      c_1(p_1) \patsect c_2(p_2)
          &=&
          \emptypattern
          ~~~~
          \text{where $c_1 \neq c_2$}
  \\
      \Pair{p_1}{p_2} \patsect c(p)
      &=&  %
      c(p) \patsect \Pair{p_1}{p_2}
      ~=~ \emptypattern 
  \end{array}
  ~~~~
  \begin{array}[t]{@{}r@{~}c@{~}lll@{}}
      \multicolumn{3}{@{}l@{}}{\hspace*{-4ex}
          \Pair{p_{11}}{p_{12}}
            \patsect
            \Pair{p_{21}}{p_{22}}
            =
            \Pair{p_{11} \patsect p_{21}}
                 {p_{12} \patsect p_{22}}
      }
      \\
      (p_1 \patunion p_2) \patsect p
      &=&
          (p_1 \patsect p) \patunion (p_2 \patsect p)
      \\ &&
      ~\text{where $p \notin \{\emptypattern, \wildcard\}$}
      \\
      p \patsect (p_1 \patunion p_2)
      &=&
          (p \patsect p_1)
          \patunion
          (p \patsect p_2)
     \\ && ~
      \text{where $p \notin \{\emptypattern,\; \wildcard,\; \cdots \patunion \cdots\}$}
  \end{array}
  \]
  
  \lesscaptionspace
  \caption{Pattern complement and pattern intersection}
  \FLabel{fig:complement-intersection}
\end{figure}

%% file: fig-step.tex
\begin{figure}[h!]
  \begin{grammar}
    Evaluation contexts
    & $\E$ & $\bnfas$ & $
        \hole
        \bnfalt
        \E\,e
        \bnfalt
        v\,\E
        \bnfalt
        c\,\E
        \bnfalt
        \Pair{\E}{e}
        \bnfalt
        \Pair{v}{\E}
        \bnfalt 
        \Case{\E}{ms}
   $        
  \end{grammar}
  \vspace*{-1.0ex}

  \judgbox{\vmatch{p}{v}{\theta}}{Value $v$ matches pattern $p$ by substitution $\theta$, \ie $[\theta]p = v$}
  \vspace*{-1.0ex}
  \begin{mathpar}
    \Infer{\MatchWild}
         {}
         {
           \vmatch{\wildcard}{v}{\emptysubst}
         }
    \and
    \Infer{\MatchAs}
         {
           \vmatch{p}{v}{\theta}
         }
         {
           \vmatch{x \As p}{v}{\theta, v/x}
         }
    \and
    \Infer{\MatchOr}
         {
           \vmatch{p_k}{v}{\theta}\text{~~~for some $k \in \{1, 2\}$}
         }
         {\vmatch{p_1 \patunion p_2}{v}{\theta}}
    \and
    \Infer{\MatchCon}
         {
           \vmatch{p}{v}{\theta}
         }
         {\vmatch{c(p)}{c(v)}{\theta}}
    \and
    \Infer{\MatchUnit}
         {
         }
         {\vmatch{\unit}{\unit}{\emptysubst}}
    ~~
    \Infer{\MatchPair}
         {
           \vmatch{p_1}{v_1}{\theta_1}
           \\
           \vmatch{p_2}{v_2}{\theta_2}
         }
         {\vmatch{\Pair{p_1}{p_2}}{\Pair{v_1}{v_2}}{(\theta_1, \theta_2)}}
  \end{mathpar}

  \smallskip

\judgbox{\novmatch{p}{v}}{Value $v$ does not match pattern $p$}
\vspace*{-1.5ex}
  \begin{mathpar}
    \Infer{\!\NomatchEmpty}
         {}
         {\novmatch{\emptypattern}{v}}
    ~~
    \Infer{\!\NomatchAs}
         {\novmatch{p}{v}}
         {\novmatch{x \As p}{v}}
    ~~
    \Infer{\!\NomatchOr}
         {
           \arrayenvbl{
           \novmatch{p_1}{v}
           \\
           \novmatch{p_2}{v}
           }
         }
         {\novmatch{p_1 \patunion p_2}{v}}
    \and
    \Infer{\NomatchConHead}
         {\text{there does not exist $v'$ such that $v = c(v')$}}
         {\novmatch{c(p)}{v}}
    \and
    \Infer{\NomatchConInner}
         {
           \novmatch{p}{v}
         }
         {\novmatch{c(p)}{c(v)}}
    \and
    \Infer{\NomatchUnit}
         {\text{$v \neq \unit$}}
         {\novmatch{\unit}{v}}
    \\
    \Infer{\NomatchPairHead}
         {\text{$v$ not a pair}}
         {\novmatch{\Pair{p_1}{p_2}}{v}}
    ~
    \Infer{\NomatchPairInner}
         {
           \novmatch{p_k}{v_k}\text{~~~for some $k \in \{1, 2\}$}
         }
         {\novmatch{\Pair{p_1}{p_2}}{\Pair{v_1}{v_2}}}
  \end{mathpar}

  \smallskip

  \judgbox{ms \step_v e'}{Matching $v$, \\ $ms$ step to $e'$}
  \vspace{-6.0ex}
  \begin{mathpar}
    \hspace*{20ex}
     \Infer{\StepMatch}
          {\vmatch{p_1}{v}{\theta}}
          {\big((\Match{p_1}{e_1}) \matchor ms\big)
           \step_v
           [\theta]e_1}
    \and
    \Infer{\StepElse}
         {\novmatch{p_1}{v}
          \\
          ms \step_v e'}
         {\big((\Match{p_1}{e_1}) \matchor ms\big) \step_v e'}
  \vspace{-1ex}
  \end{mathpar}

  \judgbox{e \step e'}{\!$e$ steps to $e'$}
  \vspace{-5.5ex}
  \begin{mathpar}
    \hspace*{95pt}
    \Infer{\!\StepBeta}
         {}
         {(\Lam{x} e_1)\,v_2 \step [v_2/x]e_1}
    ~~~
     \Infer{\!\StepCase}
          {ms \step_v e'}
          {\Case{v}{ms} \step e'}
    \vspace*{-0.5ex}
    \\
    \Infer{\!\StepDeclare}
         {}
         {(\declare{\Sigma'}{e}) \step e}
    ~~~%
    \Infer{\!\StepContext}
         {e \step e'}
         {\E[e] \step \E[e']}
  \end{mathpar}

  \lesscaptionspace
  \caption{Operational semantics}
  \FLabel{fig:step}
\end{figure}

%% file: fig-subst-typing.tex
\begin{figure}[h!]
  \centering
  
  \judgbox{\Sigma; \Gamma \entails \theta : \Gamma'}
       {Substitution $\theta$, applied to something well-formed
         under $\Sigma; (\Gamma, \Gamma')$,
         \\
         replaces variables in $\Gamma'$ to yield
         something well-formed under $\Gamma$
       }
  \begin{mathpar}
    \Infer{\SubstEmpty}
         {}
         {\Sigma; \Gamma \entails \emptysubst : \cdot}
    ~~~~
    \Infer{\SubstVar}
         {
           \Sigma; \Gamma \entails \theta : \Gamma'
           \\
           \Sigma; \Gamma \entails v : A
         }
         {\Sigma; \Gamma \entails (\theta, v/x) : (\Gamma', x : A)}
  \vspace{-3ex}
  \end{mathpar}
  
  \caption{Substitution typing}
  \FLabel{fig:subst-typing}
\end{figure}

%% file: proofs.tex
\section{Proofs}
\Label{sec:proofs}

\renewcommand{\Hand}{~}

\subsection{Properties of subtyping}

\Label{apx:PF-FIRST}

\begin{restatable}[Reflexivity]{lemma}{lemsubtypingreflexivity}
\label{lem:subtyping-reflexivity}
  If $\typejudg{\Sigma}{A}$
  then
  $\Sigma \entails A \subtype A$.
\end{restatable}
\begin{proof}
  By structural induction on $A$.

  The case for $A = \unitty$, the case for $A = A_1 * A_2$, and the case for $A = A_1 \arr A_2$ are straightforward.

  \begin{itemize}
      \ProofCaseThing{A = s}  Use \SubsortRefl and \SubSort.
  \qedhere
  \end{itemize}
\end{proof}

\begin{restatable}[Transitivity]{lemma}{lemsubtypingtransitivity}
\XLabel{lem:subtyping-transitivity}
~\\
  If $\Sigma \entails A \subtype B$
  and $\Sigma \entails B \subtype C$
  where $\typejudg{\Sigma}{A}$ and $\typejudg{\Sigma}{B}$ and $\typejudg{\Sigma}{C}$
  then
  $\Sigma \entails A \subtype C$.
\end{restatable}
\begin{proof}
  By simultaneous induction on the given derivations.

  \begin{itemize}
      \item   If either derivation is by \SubUnit, then $A = B$ or $B = C$, and
          the other given derivation is the desired result.

      \item If both derivations are by \SubProd, apply the i.h.\ as needed, then apply \SubProd.

      \item If both derivations are by \SubArr, apply the i.h.\ as needed, then apply \SubArr.

      \item If both derivations are by \SubSort, apply \SubsortTrans, then apply \SubSort.

      \item If the first derivation is by \SubSectL{k}, we have $A = A_1 \sectty A_2$.

        \begin{llproof}
          \ePf{\Sigma}{A_k \subtype B}  {Subderivation}
          \ePf{\Sigma}{A_k \subtype C}  {By i.h.}
          \ePf{\Sigma}{(A_1 \sectty A_2) \subtype C}  {By \SubSectL{k}}
        \end{llproof}

      \item If the second derivation is by \SubSectR, we have $C = C_1 \sectty C_2$.

        \begin{llproof}
          \ePf{\Sigma}{A \subtype B}   {Given}
          \ePf{\Sigma}{B \subtype C_1}   {Subderivation}
          \ePf{\Sigma}{A \subtype C_1}   {By i.h.}
          \ePf{\Sigma}{A \subtype C_2}   {Similar ($B \subtype C_2$ subderivation)}
        \Hand  \ePf{\Sigma}{A \subtype (C_1 \sectty C_2)}   {By \SubSectR}
        \end{llproof}

      \item If the first derivation $\Dee_1$ is by \SubSectR and the second derivation $\Dee_2$ is by \SubSectL{k},
        we have $B = B_1 \sectty B_2$.

          \begin{llproof}
            \ePf{\Sigma}{A \subtype B_k}   {Subderivation of $\Dee_1$}
            \ePf{\Sigma}{B_k \subtype C}   {Subderivation of $\Dee_2$}
          \Hand  \ePf{\Sigma}{A \subtype C}   {By i.h.}
          \end{llproof}
  \end{itemize}
  Other combinations of concluding rules are impossible.
\end{proof}

\subsection{Subsort properties}

\thmnonadjacentpreservation*
\begin{proof}
  Suppose $\Sigma_1, \Sigma_2, \Sigma_3 \entails s \subsort t$
  where $s, t \in \dom{\Sigma_1, \Sigma_2}$.
  Following \Definitionref{def:preserving}, we need to show that
  $\Sigma_1, \Sigma_2 \entails s \subsort t$.

  Every derivation of a subsorting judgment is essentially a finite path in a directed graph
  from the subsort to the supersort.  The path from $s$ to $t$ must pass through
  edges (subsortings) in $\Sigma_3$ a finite number of times, say $n$ times.
  Proceed by induction on $n$:

  \begin{itemize}
  \item If $n = 0$, the path from $s$ to $t$ does not pass through $\Sigma_3$
    at all, so we can simply replace $(\Sigma_1, \Sigma_2, \Sigma_3)$ in the given
    derivation with $(\Sigma_1, \Sigma_2)$.

  \item If $n > 0$, choose the last $\Sigma_3$ segment in the path:
    \[
      s \subsort \dots \subsort s_1 \subsort \underbrace{t_3 \subsort \dots \subsort  t_3'}_{\text{subsort edges in $\Sigma_3$}} \subsort s_1' \subsort \dots \subsort t
    \]
    Here, $t_3, t_3' \in \dom{\Sigma_3}$.
    Now consider the vertices (sorts) $s_1$ and $s_1'$.  These sorts must be
    in $\dom{\Sigma_1, \Sigma_2}$.  Since $\Sigma_2$ and $\Sigma_3$ are
    non-adjacent, neither $s_1$ nor $s_1'$ can be in $\dom{\Sigma_2}$.
    Therefore, $s_1, s_1' \in \dom{\Sigma_1}$.

    All the edges from $t_3$ to $t_3'$ are in $\Sigma_3$, so from
    $\Sigma_1, \Sigma_2, \Sigma_3 \entails s_1 \subsort s_1'$
    we get $\Sigma_1, \Sigma_3 \entails s_1 \subsort s_1'$.

    It is given that $\Sigma_3$ preserves subsorting of $\Sigma_1$.
    Therefore, $\Sigma_1 \entails s_1 \subsort s_1'$, yielding a path
    \[
      s \subsort \dots \subsort \underbrace{s_1 \subsort \dots \subsort s_1'}_{\text{subsort edges in $\fighi{\Sigma_1}$}}
      \subsort \dots \subsort t
    \]
    This path has one less $\Sigma_3$ segment than the one we started with, so
    the result follows by induction.
  \qedhere
  \end{itemize}
\end{proof}

\subsection{Strengthening, weakening, and substitution}

\subsubsection{Weakening of the supporting judgments}

\begin{restatable}[Weakening (lowest level)]{lemma}{lemweakeninglow}
\XLabel{lem:weakening-low}
~%
Given $\Sigma$ and $\Sigma'$ such that $\Sigma' \sstr \Sigma$:%
    \begin{enumerate}[(i)]
      \item %
          If $\Sigma \entails A \type$ then $\Sigma' \entails A \type$.
      \item %
          If $\Sigma \entails s_1 \subsort s_2$ then $\Sigma' \entails s_1 \subsort s_2$.
      \item %
          If $\Sigma \entails A \subtype B$ then $\Sigma' \entails A \subtype B$.
      \item %
          If $\Sigma \entails c : A \arr s$ then $\Sigma' \entails c : A \arr s$.
      \item %
        If $\Sigma \entails A \refines \tau$
        then $\Sigma' \entails A \refines \tau$.
      \item %
          If $\contypejudg{\Sigma}{c : C}$
          then $\contypejudg{\Sigma'}{c : C}$.
    \end{enumerate}

    In part (iii), the resulting derivation has the same size (number of horizontal lines)
    as the given derivation.
\end{restatable}
\begin{proof}  For each part, by induction on the given derivation,
  assuming lower-numbered parts.

  Part (i): All cases are straightforward.

  Part (ii), subsorting: All 3 cases
  (\SubsortAssum, \SubsortRefl and \SubsortTrans)
  are straightforward.

  Part (iii), subtyping: All cases are straightforward; the case for \SubSort
  uses part (ii) of the i.h.

  Part (iv), constructor typing: There is one case, \ConArr, which uses
  the i.h.\ (ii).

  Part (v), $A \refines \tau$: Straightforward.

  Part (vi), constructor type well-formedness: use part (i).
\end{proof}

\begin{restatable}[Weakening `SafeConAt']{lemma}{lemweakeningsafeconat}
\XLabel{lem:weakening-safeconat}
~\\
  If
  $\safeextat{\Sigma, \Sigma'}{S\sortblock{\blk}}{A \arr s}{c}{t}$
  then
  $\safeextat{\Sigma, \Omega, \Sigma'}{S\sortblock{\blk}}{A \arr s}{c}{t}$.
\end{restatable}
\begin{proof}
  By inversion on \SafeConAt and applying \Lemmaref{lem:weakening-low} parts
  (ii), (ii), and (iii), and then applying \SafeConAt.
\end{proof}

\begin{restatable}[Supersorting `SafeConAt']{lemma}{lemsupersortsafeconat}
\XLabel{lem:supersort-safeconat}
~\\
  If
  $\safeextat{\Sigma}{S\sortblock{\blk}}{A \arr s}{c}{t'}$
  and
  $\Sigma, S\sortblock{\blk} \entails t' \subsort t$
  then
  $\safeextat{\Sigma}{S\sortblock{\blk}}{A \arr s}{c}{{t}}$.
\end{restatable}
\begin{proof}
  Some of the premises of \SafeConAt do not involve $t'$ at all, so we can reuse
  them directly.  The exception is the premise
  $\Sigma, S\sortblock{\blk} \entails s' \subsort t'$.  Applying \SubsortTrans
  gives $\Sigma, S\sortblock{\blk} \entails s' \subsort t$.  
 Now we can apply \SafeConAt.
\end{proof}

\thmweakeningsafe*
\begin{proof} ~
 
        \begin{llproof}
          \blkjudgPf{\Sigma, \Sigma'}{S}{\blk}{c : A \arr s}   {Given}
        1~~~  \inPf{s}{S}   {By inversion on \BlockCon}
         ~~~  \contypejudgPf{\Sigma, \Sigma', S\sortblock{\blk}}{c : A \arr s}   {\ditto}
         ~~~  \safeextatPf{\Sigma, \Sigma'}{S\sortblock{\blk}}{A \arr s}{c}{t}   {\ditto~~~for all $t \in \dom{\Sigma, \Sigma'}$ such that \dots}
        2~~~  \contypejudgPf{\Sigma, \Sigma', S\sortblock{\blk}}{c : A \arr s}   {By \Lemmaref{lem:weakening-low} (vi)}
        \end{llproof}
        
        {\noindent Suppose} that $t \in \dom{\Sigma, \Omega, \Sigma'}$
        and $\Sigma, \Omega, \Sigma', S\sortblock{\blk} \entails s \subsort t$.

        \begin{itemize}
        \ProofCaseThing{t \in \dom{\Sigma, \Sigma'}}

          \begin{llproof}
            \ePf{\Sigma, \Omega, \Sigma', S\sortblock{\blk}}{s \subsort t}   {Assumption}
          \end{llproof}

          It is given that $S\sortblock{\blk}$ preserves subsorting for
          $(\Sigma, \Sigma')$.  By inversion on $\sigjudg{(\Sigma, \Sigma')}$,
          signature $\Sigma'$ preserves subsorting for $\Sigma$.
          Combining these, we have that $(\Sigma', S\sortblock{\blk})$ preserves
          subsorting for $\Sigma$.

          It is also given that $\blk$ does not mention anything in
          $\dom{\Omega}$, and by inversion on $\sigjudg{(\Sigma, \Omega)}$
          (and using $\dom{\Sigma, \Omega, \Sigma'} \sect S = \emptyset$),
          we know that $\Omega$ does not mention anything in $S$; therefore,
          $S\sortblock{\blk}$ and $\Omega$ are non-adjacent
          (\Definitionref{def:non-adjacency}).
          
          By \Theoremref{thm:non-adjacent-preservation}
          with
          $\Sigma_1 = \Sigma$
          and
          $\Sigma_2 = (\Sigma', S\sortblock{\blk})$
          and
          $\Sigma_3 = \Omega$,
          we have that
          $\Omega$ preserves subsorting of $(\Sigma, \Sigma', S\sortblock{\blk})$.

          Preservation of subsorting is invariant under signature permutation, so
          $\Omega$ can be permuted leftward:
          $\Sigma, \Omega, \Sigma', S\sortblock{\blk} \entails s \subsort t$
          if and only if $\Sigma, \Sigma', S\sortblock{\blk} \entails s \subsort t$,
          and we assumed the former judgment.  Therefore:

          \begin{llproof}
            \ePf{\Sigma, \Sigma', S\sortblock{\blk}}{s \subsort t}   {}
          \end{llproof}

          This is the guard of one of the premises under the above ``for all''.  Therefore:

          \begin{llproof}
          \safeextatPf{\Sigma, \Sigma'}{S\sortblock{\blk}}{A \arr s}{c}{t}   {}
          \safeextatPf{\Sigma, \Omega, \Sigma'}{S\sortblock{\blk}}{A \arr s}{c}{t}   {By \Lemmaref{lem:weakening-safeconat}}
          \end{llproof}

        \ProofCaseThing{t \in \dom{\Omega}}
           
          Here we have $\Sigma, \Omega, \Sigma', S\sortblock{\blk} \entails s \subsort t$.
        
          It is given that $\blk$ does not mention any sorts in $\dom{\Omega}$.
          Therefore, $s \subsort t$ must have been derived transitively: there must exist
          another sort $s' \in \dom{\Sigma, \Sigma'} \union S$ such that
          $s \subsort s'$ and $s' \subsort t$.  Using the reasoning in the subcase for
          when $t \in \dom{\Sigma, \Sigma'}$, we get

          \begin{llproof}
            \safeextatPf{\Sigma, \Omega, \Sigma'}{S\sortblock{\blk}}{A \arr s}{c}{\fighi{s'}}   {}
            \safeextatPf{\Sigma, \Omega, \Sigma'}{S\sortblock{\blk}}{A \arr s}{c}{\fighi{t}}   {By \Lemmaref{lem:supersort-safeconat}}
          \end{llproof}          
        \end{itemize}

        {\noindent This} shows the ``for all'' part of \BlockCon.
        Together with ``1'' and ``2'' above, we can apply \BlockCon:

        \begin{llproof}
           \blkjudgPf{\Sigma,  \Omega, \Sigma'}{S}{\blk}{c : A \arr s}   {By \BlockCon}
        \end{llproof}
        \qedhere
\end{proof}

\thmsignatureinterleaving*
\begin{proof}   By induction on $\Sigma'$.

  If $\Sigma' = \cdot$, then we already have our result.

  Otherwise, $\Sigma' = (\Sigma_0', S\sortblock{\blk})$.

        \begin{llproof}
          \sigjudgPf{(\Sigma, \Sigma_0', S\sortblock{\blk})}   {Above}
          \sigjudgPf{(\Sigma, \Sigma_0')} {By inversion on \SigBlock}
          \disjointPf{S}{\dom{\Sigma, \Sigma_0'}}   {\ditto}
               \eqPf{(\Sigma, \Sigma_0' \entails {\subsort}|_{\dom{\Sigma,\Sigma_0'}})}
               {(\Sigma, \Sigma_0', S\sortblock{\blk} \entails {\subsort}|_{\dom{\Sigma,\Sigma_0'}})}
               {\ditto ~~[Preservation]}
          \blkjudgPf{\Sigma, \Sigma_0'}{S}{\blk}{c : A \arr s}   {\ditto~~for all $(c : A \arr s) \in \blk$}
          \proofsep
      1~~  \sigjudgPf{(\Sigma, \Omega, \Sigma_0')} {By i.h.}
          \decolumnizePf
          \disjointPf{\dom{\Sigma'}}{\dom{\Omega}}    {Given}
          \disjointPf{\dom{\Sigma_0', S\sortblock{\blk}}}{\dom{\Omega}}    {$\Sigma' = (\Sigma_0', S\sortblock{\blk})$}
          \disjointPf{\big(\dom{\Sigma_0'} \union S\big)}{\dom{\Omega}}    {By def.\ of $\dom{-}$}
          \disjointPf{S}{\dom{\Omega}}    {By a property of $\sect$ and $\union$}
          \disjointPf{S}{\dom{\Sigma, \Sigma_0'}}   {Above}
      2~~  \disjointPf{S}{\dom{\Sigma, \Omega, \Sigma_0'}}   {By def.\ of $\dom{-}$ and a property of $\sect$}
     \end{llproof}

     We still need to prove the following:

     \begin{llproof}
      3~~  
          \eqPf{(\Sigma, \Omega, \Sigma_0' \entails {\subsort}|_{\dom{\Sigma,\Omega,\Sigma_0'}})}
               {(\Sigma, \Omega, \Sigma_0', S\sortblock{\blk} \entails {\subsort}|_{\dom{\Sigma,\Omega,\Sigma_0'}})}
               {To be proved}
          \decolumnizePf
     4~~          \blkjudgPf{\Sigma, \Omega, \Sigma_0'}{S}{\blk}{c : A \arr s}   {To be proved~~for all $(c : A \arr s) \in \blk$}
          \proofsep
        \Hand
          \sigjudgPf{(\Sigma, \Omega, \Sigma_0', S\sortblock{\blk})}   {By \SigBlock (1, 2, 3, 4)}
      \end{llproof}

      \begin{itemize}
      \item Proof of 3:

           We have $\sigjudg{(\Sigma, \Omega)}$.

           Elaborating equation 3, we need to show, for all
           $t_1, t_2 \in \dom{\Sigma, \Omega, \Sigma_0'}$, that each direction holds:

           \begin{itemize}
           \item (a) If $\Sigma, \Omega, \Sigma_0' \entails t_1 \subsort t_2$
               then $\Sigma, \Omega, \Sigma_0', S\sortblock{\blk} \entails t_1 \subsort t_2$.
           \item (b) If $\Sigma, \Omega, \Sigma_0', S\sortblock{\blk} \entails t_1 \subsort t_2$
             then $\Sigma, \Omega, \Sigma_0' \entails t_1 \subsort t_2$.
           \end{itemize}
           
           For direction (a), \Lemmaref{lem:weakening-low} (ii) suffices.
        
           For direction (b):

           We have $\sigjudg{(\Sigma, \Sigma_0', S\sortblock{\blk})}$.

           By inversion on \SigBlock,
           the block $\blk$ preserves subsorting of the signature
           $(\Sigma, \Sigma_0')$. 

           By inversion on \BlockSubsort, none of the subsortings added by $\blk$
           are in $\dom{\Omega}$, and none of the subsortings added by $\Omega$
           are in $S$.  That is, $S\sortblock{\blk}$ and $\Omega$ are non-adjacent.

           By \Theoremref{thm:non-adjacent-preservation} with
           $\Sigma_1 = (\Sigma, \Sigma_0')$
           and
           $\Sigma_2 = \Omega$
           and
           $\Sigma_3 = S\sortblock{\blk}$,
           we have that
            $\Sigma_3$ preserves subsorting of $(\Sigma_1, \Sigma_2)$, that is,
            $S\sortblock{\blk}$ preserves subsorting of $(\Sigma, \Sigma_0', \Omega)$.

           Preservation of subsorting is invariant under signature permutation,
           so $S\sortblock{\blk}$ preserves subsorting of $(\Sigma, \Omega, \Sigma_0')$.
           By \Definitionref{def:preserving}, (b) holds.

      \item Proof of 4:

           Most of the work will be done by \Theoremref{thm:weakening-safe}.  We need to
           show all of that theorem's conditions.  All of our meta-variables match up with
           the statement of the lemma, except that our $\Sigma_0'$ will play the role
           of $\Sigma'$.

           \begin{itemize}
           \item $\sigjudg{(\Sigma, \Sigma_0')}$: Above.

           \item $\sigjudg{(\Sigma, \Omega)}$: Given.

           \item $\dom{\Sigma_0'} \sect \dom{\Omega} = \emptyset$: Follows from
             $\dom{\Sigma'} \sect \dom{\Omega} = \emptyset$, which was given.
           
           \item $\dom{\Sigma, \Omega, \Sigma_0'} \sect S = \emptyset$: This is ``2'', shown above.

           \item $\blk$ does not mention anything in $\dom{\Omega}$:

             If $\blk$ mentioned anything in $\dom{\Omega}$, it would contradict
             the premise of \BlockSubsort and/or the \textit{contype} premise of
             \BlockCon.

           \item $\blk$ preserves subsorting for $(\Sigma, \Sigma_0')$:
               This is ``Preservation'', shown above.

           \item $(c : A \arr s) \in \blk$: Assumption.

           \item $\blkjudg{\Sigma, \Sigma_0'}{S}{\blk}{c : A \arr s}$: Above.
           \end{itemize}

           The judgment marked 4 follows by \Theoremref{thm:weakening-safe}.
      \end{itemize}

     {\noindent Now} we can apply \SigBlock.

     \begin{llproof}
        \Hand
          \sigjudgPf{(\Sigma, \Omega, \Sigma_0', S\sortblock{\blk})}   {By \SigBlock (1, 2, 3, 4)}
        \end{llproof}
  \qedhere
\end{proof}

\subsubsection{Intersect strengthening}

\begin{restatable}[Properties of stronger contexts]{lemma}{lemcontextstr}
\XLabel{lem:context-str}
~
\begin{enumerate}[(i)] \raggedright
\item \emph{Reflexivity:} For all contexts $\Gamma$, we have $\Sigma \entails \Gamma \gstr \Gamma$.
\item \emph{Transitivity:} If $\Sigma \entails \Gamma'' \gstr \Gamma'$ and $\Sigma \entails \Gamma' \gstr \Gamma$ then $\Sigma \entails \Gamma'' \gstr \Gamma$.
\item \emph{Concatenation:} If $\Sigma \entails \Gamma_1' \gstr \Gamma_1$
  and $\Sigma \entails \Gamma_2' \gstr \Gamma_2$
  then $\Sigma \entails (\Gamma_1', \Gamma_2') \gstr (\Gamma_1, \Gamma_2)$.
\end{enumerate}
\end{restatable}
\begin{proof}
  Part (i): By induction on $\Gamma$, using \Lemmaref{lem:subtyping-reflexivity}.

  Part (ii): Use \Lemmaref{lem:subtyping-transitivity}.

  Part (iii): By induction on $\Gamma_2$.
\end{proof}

\begin{restatable}[Constructor]{lemma}{lemconstructor}
\XLabel{lem:constructor}
  ~\\
     If
     $\sigjudg{\Sigma}$
     and $\sigjudg{\Sigma_+}$
     and
     $\Sigma_+ \sstr \Sigma$
     \\
     and
     $(c : A_c^+ \arr s_c^+) \in \Sigma_+$
     and $\Sigma_+ \entails s_c^+ \subsort s$
     and $s \in \dom{\Sigma}$
     \\
     then
     $(c : A_c \arr s_c) \in \Sigma$
     and $\Sigma_+ \entails s_c^+ \subsort s_c$
     and $\Sigma_+ \entails A_c^+ \subtype A_c$.
\end{restatable}
\begin{proof}
  If $(c : A_c^+ \arr s_c^+) \in \Sigma$, then:  Let $A_c = A_c^+$ and $s_c = s_c^+$.
  By \Lemmaref{lem:subtyping-reflexivity}, $\Sigma_+ \entails A_c^+ \subtype A_c$.

  Otherwise, by inversion on $\sigjudg{\Sigma_+}$,
  there is a derivation using \BlockCon that says that, for some $\Sigma_1$
  such that $\Sigma = (\Sigma_1, \Sigma_2)$, we have
  $\safeextat{\Sigma_1}{S\sortblock{\blk}}{A_c^+ \arr s_c^+}{c}{t}$
  for all $t \in \dom{\Sigma_1}$ such that $\Sigma_1, S\sortblock{\blk} \entails s_c^+ \subsort t$.

  We have $\Sigma_+ \entails s_c^+ \subsort s$.
  By well-formedness of the given signatures, all blocks preserve subsorting of previously-defined
  sorts.  Therefore
  \[
     \Sigma_1, S\sortblock{\blk} \entails s_c^+ \subsort s
  \]
  Instantiating the above $t$ with $s$ gives
  \[
    \safeextat{\Sigma_1}{S\sortblock{\blk}}{A_c^+ \arr s_c^+}{c}{s}
  \]
  \begin{llproof}
    \inPf{(c : A_c \arr s_c)} {\Sigma_1}   {By inversion on \SafeConAt}
    \ePf{\Sigma_1, S\sortblock{\blk}} {s_c \subsort s}  {\ditto}
    \ePf{\Sigma_1, S\sortblock{\blk}} {s_c^+ \subsort s_c} {\ditto}
    \ePf{\Sigma_1, S\sortblock{\blk}} {A_c^+ \subtype A_c} {\ditto}
    \decolumnizePf
  \Hand  \ePf{\Sigma_+} {s_c^+ \subsort s_c}  {By \Lemmaref{lem:weakening-low} (ii)}
  \Hand  \ePf{\Sigma_+} {A_c^+ \subtype A_c} {By \Lemmaref{lem:weakening-low} (iii)}
  \end{llproof}
  \qedhere
\end{proof}

\begin{restatable}[Intersect strengthening]{lemma}{lemintersectstrength}
\XLabel{lem:intersect-strength}
~\\
    If $\sigjudg{\Sigma_+}$ and $\sigjudg{\Sigma}$
    and $\Sigma_+ \sstr \Sigma$
    and $\Sigma \entails A \type$
    \\
    and $\Sigma_+ \entails A_+ \subtype A$
    \\
    then
    $\Sigma_+
    \entails
    \intersect{\Sigma_+}{A_+}{p}
    \trkstr
    \intersect{\Sigma}{A}{p}$.
\end{restatable}
\begin{proof}
  By structural induction on $p$.

  Suppose we have
  \begin{align*}
    \intersect{\Sigma_+}{A_+}{p}
    &~=~
    \big\{
      \intersectresult{\Sigma_+^1}{\Gamma_+^1}{B_+^1},
      \,\dots,\,
      \intersectresult{\Sigma_+^n}{\Gamma_+^n}{B_+^n}
    \big\}
  \\    
    \intersect{\Sigma}{A}{p}
   &~=~
    \big\{
      \intersectresult{\Sigma^1}{\Gamma^1}{B^1},
      \,\dots,\,
      \intersectresult{\Sigma^m}{\Gamma^m}{B^m}
    \big\}
  \end{align*}
 For each track
 $B^*_+ = \intersectresult{\Sigma^+_k}{\Gamma^+_k}{B^+_k}
   \in
   \intersect{\Sigma_+}{A_+}{p}$,
 we need to find some track $B^*$ in $\intersect{\Sigma}{A}{p}$
 such that $B^*_+ \trkstr B^*$.

 \begin{itemize}
   \ProofCaseThing{p = \wildcard}
      In this case, $\intersectx$ ignores the signature entirely.
      The result follows by %
      \Lemmaref{lem:context-str} (i), 
      and the given subtyping $\Sigma_+ \entails A_+ \subtype A$.

   \ProofCaseThing{p = (x \As p_0)}
   
      \begin{llproof}
        \trkstrPf{\intersect{\Sigma_+}{A_+}{p_0}}{\intersect{\Sigma}{A}{p_0}}  {By i.h.}
      \end{llproof}

      Suppose $\intrsctresult{\argle}{\Gamma'_+}{B_+} \in \intersect{\Sigma_+}{A_+}{p_0}$.

      \begin{llproof}
        \inPf{\intrsctresult{\argle}{\Gamma'_+, x : B_+}{B_+}}
             {\intersect{\Sigma_+}{A_+}{x \As p_0}}   {By definition of \intersectx}
        \proofsep
        \inPf{\intrsctresult{\Sigma'}{\Gamma'}{B}} {\intersect{\Sigma}{A}{p_0}}   {By above $\trkstrsym$}
        \trkstrPf{\intrsctresult{\argle}{\Gamma'_+}{B_+}}{\intrsctresult{\Sigma'}{\Gamma'}{B}}   {\ditto}
        \ePf{\Sigma_+\argle}{\Gamma'_+ \gstr \Gamma'}   {From \Definitionref{def:track-str}}
        \ePf{\Sigma_+\argle}{B_+ \subtype B}   {\ditto}
        \proofsep
        \inPf{\intrsctresult{\Sigma'}{\Gamma', x : B}{B}}
             {\intersect{\Sigma}{A}{x \As p_0}}   {By definition of \intersectx}        
        \proofsep
        \trkstrPf{\intrsctresult{\argle}{\Gamma'_+, x : B_+}{B_+}}{\intrsctresult{\Sigma'}{\Gamma', x : B}{B}}  {By \Definitionref{def:track-str}}
      \end{llproof}

   \ProofCaseThing{p = \emptypattern}
      In this case, $\intersectx$ ignores the signature---and the result is the empty set,
      which is, trivially, a strengthening of the empty set.

   \ProofCaseThing{p = (p_1 \patunion p_2)}

      By i.h., $\intersect{\Sigma_+}{A_+}{p_1} \trkstr \intersect{\Sigma}{A}{p_1}$. \\
      Similarly,  $\intersect{\Sigma_+}{A_+}{p_2} \trkstr \intersect{\Sigma}{A}{p_2}$. \\
      Then:
      \[
          \arrayenvl{~~~~~
            \big(\intersect{\Sigma_+}{A_+}{p_1}
          \union
          \intersect{\Sigma_+}{A_+}{p_2}\big)
          \\ 
          \trkstr
          \big(\intersect{\Sigma}{A}{p_1}
          \union
          \intersect{\Sigma}{A}{p_2}\big)
          }
      \]

  \medskip
        
  \ProofCaseThing{p = (p_1 * p_2)}

      \begin{llproof}
        \ePf{\Sigma_+}{(A^1_+ * A^2_+) \subtype (A^1 * A^2)}    {Given}
        \ePf{\Sigma_+}{A^1_+ \subtype A^1}    {By inversion on \SubProd}
        \ePf{\Sigma_+}{A^2_+ \subtype A^2}    {\ditto}
        \trkstrPf{\intersect{\Sigma_+}{A^1_+}{p_1}}{\intersect{\Sigma}{A^1}{p_1}}  {By i.h.}    
        \trkstrPf{\intersect{\Sigma_+}{A^2_+}{p_2}}{\intersect{\Sigma}{A^2}{p_2}}  {By i.h.}
      \end{llproof}

      Suppose:
      \begin{align*}
        \intrsctresult{\argle}{\Gamma^1_+}{B^1_+}
          &\in
          \intersect{\Sigma_+}{A^1_+}{p_1}
        \\
        \intrsctresult{\bargle}{\Gamma^2_+}{B^2_+}
          &\in
          \intersect{\Sigma_+}{A^2_+}{p_2}
      \end{align*}

      \begin{llproof}
        \inPf{\intrsctresult{\arglebargle}{\Gamma^1_+, \Gamma^2_+}{B^1_+ * B^2_+}}
             {\intersect{\Sigma_+}{A_1 * A_2}{\Pair{p_1}{p_2}}}   {By definition of \intersectx}
        \decolumnizePf
        \inPf{\intrsctresult{\Sigma^1}{\Gamma^1}{B^1}} {\intersect{\Sigma}{A_1}{p_1}}   {By above $\trkstrsym$}
        \trkstrPf{\intrsctresult{\argle}{\Gamma^1_+}{B^1_+}}{\intrsctresult{\Sigma^1}{\Gamma^1}{B^1}}   {\ditto}
        \ePf{\Sigma_+\argle}{\Gamma^1_+ \gstr \Gamma^1}   {From \Definitionref{def:track-str}}
        \ePf{\Sigma_+\argle}{B^1_+ \subtype B^1}   {\ditto}
        \proofsep
        \ePf{\Sigma_+\bargle}{\Gamma^2_+ \gstr \Gamma^2}   {Similar, with $2$ substituted for $1$}
        \ePf{\Sigma_+\bargle}{B^2_+ \subtype B^2}   {\ditto}
        \proofsep
        \ePf{\Sigma_+\arglebargle}{(\Gamma^1_+, \Gamma^2_+) \gstr (\Gamma^1, \Gamma^2)}   {By \Lemref{lem:context-str} (iii)}
        \ePf{\Sigma_+\arglebargle}{(B^1_+ * B^2_+) \subtype (B^1 * B^2)}   {By \SubProd}
        \decolumnizePf
        \inPf{\intrsctresult{\Sigma^1, \Sigma^2}{\Gamma^1, \Gamma^2}{(B^1 * B^2)}}
             {\intersect{\Sigma}{A_1 * A_2}{\Pair{p_1}{p_2}}}   {By definition of \intersectx}        
        \proofsep
        \trkstrPf{\intrsctresult{\arglebargle}{\Gamma^1_+, \Gamma^2_+}{B^1_+ * B^2_+}}{\intrsctresult{\Sigma^1, \Sigma^2}{\Gamma^1, \Gamma^2}{B^1 * B^2}}  {By \Definitionref{def:track-str}}
      \end{llproof}

  \ProofCaseThing{p = c(p_0)}
      Suppose $c : (A_c^+ \arr s_c^+) \in \Sigma_+$
      where $\Sigma_+ \entails s_c^+ \subsort s$.
      
      \begin{llproof}
        \inPf{(c : A_c^+ \arr s_c^+)}{\Sigma_+}   {Above}
        \ePf{\Sigma_+} {s_c^+ \subsort s}   {Above}
        \ePf{\Sigma} {s \type}  {Given}
        \inPf{s}{\dom{\Sigma}}  {By inversion on \TypewfSort}
        \proofsep
        \sigjudgPf{\Sigma_+}   {Given}
        \sigjudgPf{\Sigma}   {Given}
        \inPf{(c : A_c \arr s_c)}{\Sigma}   {By \Lemmaref{lem:constructor}}
        \ePf{\Sigma_+} {s_c^+ \subsort s_c}   {\ditto}
        \ePf{\Sigma_+} {A_c^+ \subtype A_c}  {\ditto}
        \decolumnizePf
        \trkstrPf{\intersect{\Sigma_+}{A_c^+}{p_0}}
              {\intersect{\Sigma}{A_c}{p_0}}   {By i.h.}
     \end{llproof}

     Suppose $\intrsctresult{\argle}{\Gamma'_+}{B'_+} \in \intersect{\Sigma_+}{A_c^+}{p_0}$.

     \begin{llproof}
        \inPf{\intrsctresult{\Sigma'}{\Gamma'}{B'}}{\intersect{\Sigma}{A_c}{p_0}}   {By above $\trkstrsym$}
        \trkstrPf{\intrsctresult{\argle}{\Gamma'_+}{B'_+}}{\intrsctresult{\Sigma'}{\Gamma'}{B'}}   {\ditto}
        \ePf{\Sigma_+\argle} {\Gamma'_+ \gstr \Gamma'}  {By \Definitionref{def:track-str}}
        \ePf{\Sigma\argle}{B'_+ \subtype B'}  {\ditto}
        \proofsep
        \ePf{\Sigma_+} {\Gamma'_+ \gstr \Gamma'}  {Above}
        \ePf{\Sigma_+} {s_c^+ \subsort s_c}   {Above}
        \ePf{\Sigma_+} {s_c^+ \subtype s_c}   {By \SubSort}
        \proofsep
\Hand        \ePf{\Sigma_+}  {\intresult{\Gamma'_+}{s_c^+} \trkstr \intresult{\Gamma'}{s_c}}
                        {By \Definitionref{def:track-str}}
      \end{llproof}
 \qedhere
 \end{itemize}
\end{proof}

\subsubsection{Weakening of the main judgments}

\thmweakening*
\begin{proof}  For each part, by induction on the height of the given derivation.

  Part (1), expression typing:

  \begin{itemize}
      \item \TypeVar: apply \TypeVar.
      \item \TypeSub: apply the i.h.\ (1) and \Lemmaref{lem:weakening-low} (iii).
      \item %
        \TypeUnitI, \TypeArrI, \TypeArrE, \TypeProdI,
        \TypeDataI:
        straightforward.  %
        The \TypeDataI case uses \Lemmaref{lem:weakening-low} (iv).

     \DerivationProofCase{\TypeDataE}
         {\Sigma, \Sigma'; \Gamma \entails e : B
          \\
          \matchassn{\Sigma, \Sigma'}{\Gamma}{B}{\wildcard}
          \entails
          ms : A
         }
         {\Sigma, \Sigma'; \Gamma \entails (\Case{e}{ms}) : A}

         \begin{llproof}
             \ePf{\Sigma, \Sigma'; \Gamma} {e : B}  {Subderivation}
             \ePf{\Sigma, \Omega, \Sigma'; \Gamma} {e : B}  {By i.h.\ (1)}
             \proofsep
             \ePf{\matchassn{\Sigma, \Sigma'}{\Gamma}{B}{\wildcard}}{ms : A}   {Subderivation}
             \ePf{\matchassn{\Sigma, \Omega, \Sigma'}{\Gamma^+}{B}{\wildcard}}{ms : A}   {By i.h.\ (3)}
             \proofsep
             \decolumnizePf
             \ePf{\Sigma, \Sigma'; \Gamma}{(\Case{e}{ms}) : A}
                    {By \TypeDataE}
         \end{llproof}

      \DerivationProofCase{\TypeDeclare}
              {
                \arrayenvbl{
                    \sigjudg{(\Sigma, \Sigma', \Sigma'')}  %
                }
                \\
                \Sigma, \Sigma' \entails A \type
                \\
                \Sigma, \Sigma', \Sigma''; \Gamma \entails e_0 : A
              }
              {\Sigma, \Sigma'; \Gamma \entails (\declare{\Sigma''}{e_0}) : A}
           
              \begin{llproof}
                \sigjudgPf{(\Sigma, \Sigma', \Sigma'')}   {Subderivation}
                \sigjudgPf{(\Sigma, \Omega)}   {Given}
                \disjointPf{\dom{\Sigma'}}{\dom{\Omega}}   {Given}
                \disjointPf{\dom{\Sigma''}}{\dom{\Omega}}   {By renaming $\dom{\Sigma''}$}
                \disjointPf{\big(\dom{\Sigma'} \union \dom{\Sigma''}\big)}{\dom{\Omega}}{By set theory}
                \disjointPf{\dom{\Sigma', \Sigma''}}{\dom{\Omega}}   {By def.\ of $\dom{-}$}
               \decolumnizePf
              1~  \sigjudgPf{(\Sigma, \Omega, \Sigma')}  {By \Theoremref{thm:signature-interleaving}}
                \proofsep
                \ePf{\Sigma, \Sigma'} {A \type}   {Subderivation}
              2~  \ePf{\Sigma, \Omega, \Sigma'} {A \type}   {By \Lemmaref{lem:weakening-low} (i)}
                \proofsep
                \ePf{\Sigma, \Sigma', \Sigma''; \Gamma} {e_0 : A}   {Subderivation}
              3~  \ePf{\Sigma, \Omega, \Sigma', \Sigma''; \Gamma} {e_0 : A}   {By i.h.\ (1)}
                \decolumnizePf
              \Hand \ePf{\Sigma, \Omega, \Sigma'; \Gamma} {(\declare{\Sigma''}{e_0}) : A}   {By \TypeDeclare on 1, 2, 3}
              \end{llproof}
  \end{itemize}

  \smallskip

  {\noindent Part} (2), substitution typing: In the \NoLinkSubstVar\xspace case, use part (1).

  \smallskip

  {\noindent Part} (3), where a derivation of $\matchassn{\Sigma, \Sigma'}{\Gamma}{B}{p} \entails ms : D$
  is given:

  \begin{itemize}
      \DerivationProofCase{\TypeMsEmpty}
           {\intersect{\Sigma, \Sigma'}{A}{p} = \emptyset}
           {\matchassn{\Sigma, \Sigma'}{\Gamma}{A}{p} \entails \emptyms : D}

           Suppose $\intersect{\Sigma, \Omega, \Sigma'}{A^+}{p} \neq \emptyset$.
           
           By \Lemmaref{lem:intersect-strength}, for each track in $\intersect{\Sigma, \Omega, \Sigma'}{A^+}{p}$ there exists a track in $\intersect{\Sigma, \Sigma'}{A}{p}$.  But we have as a premise that $\intersect{\Sigma, \Sigma'}{A}{p} = \emptyset$, a contradiction.

           Therefore
           $\intersect{\Sigma, \Omega, \Sigma'}{A^+}{p} = \emptyset$.

           The result follows by \TypeMsEmpty.

{\def\zzpreDerivationProofCase{\hspace{-9ex}}\small      \DerivationProofCase{\!\TypeMs}
           {
             \arrayenvbl{
               \Sigma, \Sigma' \entails A \refines \tau
               \\
               \Ursig \entails p_1 : \tau
             }
             \hspace{-1.5ex}
             \\
             \arrayenvbl{
                 \text{for all $\intresult{\Gamma'}{B}$}
                 \\
                 \text{~$\in \intersect{\Sigma}{A}{p \patsect p_1}$:}
                 \\
                 ~~~~\Sigma; \Gamma, \Gamma'
                              \entails
                              e_1 : D
             }
             \hspace{-2.0ex}
             \\
             \matchassn {\Sigma, \Sigma'} \Gamma A {(p \patsect \patcomplement p_1)}
             \entails
             ms : D
           }
           {
             \matchassn {\Sigma, \Sigma'} \Gamma A p
             \entails
             \big( (\Match{p_1}{e_1}) \matchor ms\big) : D
           }
}
           
           \begin{llproof}
             \ePf{\Sigma, \Sigma'} {A \refines \tau}  {Subderivation}
             \ePf{\Sigma, \Omega, \Sigma'} {A \refines \tau}  {By \Lemmaref{lem:weakening-low} (vi)}
             \ePf{\Ursig} {p_1 : \tau}  {Subderivation}
             \decolumnizePf
             \sstrPf{(\Sigma, \Omega, \Sigma')}{(\Sigma, \Sigma')}       {By \Definitionref{def:signature-str}}
             \ePf{\Sigma, \Omega, \Sigma'}{\Gamma^+ \gstr \Gamma}   {Given}
             \ePf{\Sigma, \Sigma'}{A \type}   {Given}
           \end{llproof}

           By Lemma \ref{lem:intersect-strength},
           \[
                   \intersect{\Sigma, \Omega, \Sigma'}{A}{p \patsect p_1}
                   \trkstr
                   \intersect{\Sigma, \Sigma'}{A}{p \patsect p_1}
           \]
           By \Definitionref{def:track-str}, each track
           in $\intersect{\Sigma, \Omega, \Sigma'}{A}{p \patsect p_1}$ is stronger than some track in
           $\intersect{\Sigma, \Sigma'}{A}{p \patsect p_1}$.

           That is,
           if
           $\intrsctresult{\Sigma_0^+}{\Gamma_0^+}{B^+}
            \in
            \intersect{\Sigma, \Omega, \Sigma'}{A}{p \patsect p_1}$,
           then
           \[
               \Sigma, \Omega, \Sigma'
               \entails
               \intrsctresult{\Sigma_0^+}{\Gamma_0^+}{B^+}
               \trkstr
               \intrsctresult{\Sigma_0}{\Gamma_0}{B}
           \]
           where $\intrsctresult{\Sigma_0}{\Gamma_0}{B} \in \intersect{\Sigma, \Sigma'}{A}{p \patsect p_1}$.

           \begin{llproof}
             \ePf{
               \Sigma, \Sigma' %
               ;
               \Gamma, \Gamma_0 %
             }
             {e_1 : D}
             {Subderivation}
             \proofsep
             \ePf{\Sigma, \Omega, \Sigma'} {\Gamma \gstr \Gamma}  {By \Lemref{lem:context-str} (i)}
             \ePf{\Sigma, \Omega, \Sigma'} {\Gamma_0^+ \gstr \Gamma_0}  {By above $\intrsctresult{\Sigma_0^+}{\Gamma_0^+}{B^+}
               \trkstr
               \intrsctresult{\Sigma_0}{\Gamma_0}{B}$}
             \ePf{\Sigma, \Omega, \Sigma'} {(\Gamma, \Gamma_0^+) \gstr (\Gamma, \Gamma_0)}   {By \Lemref{lem:context-str} (iii)}
             \proofsep
             \ePf{ \Sigma, \Omega, \Sigma'%
               ;\Gamma, \Gamma_0^+%
             }{e_1 : D}  {By i.h.}
           \end{llproof}

          Since the above holds for all tracks, this gives the necessary ``for all'' premises.
          Then:
          
          \begin{llproof}
            \ePf{\matchassn {\Sigma, \Sigma'} \Gamma A {(p \patsect \patcomplement p_1)}}
                   {ms : D}
                   {Subderivation}
            \ePf{\matchassn {\Sigma, \Omega, \Sigma'} {\Gamma^+} {A} {(p \patsect \patcomplement p_1)}}
                   {ms : D}
                   {By i.h.}
            \proofsep
            \ePf{\matchassn{\Sigma, \Omega, \Sigma'}{\Gamma^+}{A}{p}}
                    {\big( (\Match{p_1}{e_1}) \matchor ms\big) : D}
                    {By \TypeMs}
          \end{llproof}
          \qedhere
  \end{itemize}
\end{proof}

\begin{restatable}[Value substitution]{lemma}{lemvaluesubstitution}
\XLabel{lem:value-substitution}
   Suppose $\Sigma; \Gamma_L, \Gamma_R \entails v : A$.
   \vspace{-0.6ex}

    \begin{enumerate}[(i)]
      \item
          If $\Sigma; \Gamma_L, x : A, \Gamma_R \entails e : B$
          then $\Sigma; \Gamma_L, \Gamma_R \entails [v/x]e : B$.
       \item
           If $\matchassn{\Sigma}{\Gamma_L, x : A, \Gamma_R}{B}{p} \entails ms : D$
           then
           $\matchassn{\Sigma}{\Gamma_L, \Gamma_R}{B}{p} \entails [v/x]ms : D$.
    \end{enumerate}
\end{restatable}
\begin{proof}
  In each part, by induction on the derivation specific to that part.

  For part (i), in the \TypeVar case, we have $\Sigma; \Gamma_L, \Gamma_R \entails v : A$,
  and $v = [v/x]x$.  In all other cases, use the i.h.\ on each premise, then apply the same rule.
  (For \TypeDataE, use part (ii).)
  
  For part (ii), we have two cases:

  \begin{itemize}
  \ProofCaseRule{\TypeMs}
      The $\intersectx$ function does not depend on the typing context, so we get
      the same set of tracks.  Apply the i.h.\ (i) to each typing subderivation and
      the i.h.\ (ii) to the last subderivation.

  \ProofCaseRule{\TypeMsEmpty}  This rule does not depend on the typing context at all, so
          we just apply it.
  \qedhere
  \end{itemize}
\end{proof}

\begin{restatable}[Multiple substitution]{lemma}{lemmultiplevaluesubstitution}
\XLabel{lem:multiple-value-substitution}
~\\
   If $\Sigma; \Gamma, \Gamma' \entails e : B$
   and $\Sigma; \Gamma \entails \theta : \Gamma'$
   then $\Sigma; \Gamma \entails [\theta]e : B$.
\end{restatable}
\begin{proof}
  By induction on the derivation of $\Sigma; \Gamma \entails \theta : \Gamma'$.
  In the \SubstEmpty case, apply the equality $[\theta]e = [\cdot]e = e$.
  In the \SubstVar case, use the i.h.\ and \Lemmaref{lem:value-substitution}.
\end{proof}

\subsection{Value inversion}

A value inversion (or \emph{canonical forms}) lemma holds:

\begin{restatable}[Inversion]{lemma}{leminversion}
\XLabel{lem:inversion}
Suppose $\Sigma; \cdot \entails v :B$.

\vspace{-0.6ex}
\begin{enumerate}[(1)]
\item %
  If $\Sigma \entails B \subtype s$
  then 
  there exist $c$ and $v'$ such that
  $v = c(v')$
  \\
  and
  $(c : A \arr t) \in \Sigma$
  and $\Sigma \entails t \subsort s$
  and $\Sigma; \cdot \entails v' : A$.

\item %
  If $\Sigma \entails B \subtype (A_1 \arr A_2)$
  \\
  then
  $v = \Lam{x} e$
  and
  $\Sigma; \cdot, x : B_1 \entails e : A_2$
  where
  $\Sigma \entails A_1 \subtype B_1$.

\item %
  If $\Sigma \entails B \subtype \unitty$
  then
  $v = \unit$.  

\item %
  If $\Sigma \entails B \subtype (A_1 * A_2)$
  then
  $v = \Pair{v_1}{v_2}$
  where $\Sigma; \cdot \entails v_1 : A_1$
  and  $\Sigma; \cdot \entails v_2 : A_2$.
\end{enumerate}  
\end{restatable}
\begin{proof}
  By induction on the given derivation.

  \begin{itemize}
  \item \textbf{Part (1)}:

    We have $\Sigma \entails B \subtype s$.  By inversion on subtyping,
    $B$ has the form
    $B^* \bnfas
      t
      \bnfalt B^*_1 \sectty B^*_2
    $.
    Thus, 
    the only possible cases are \TypeSub,
  \TypeSectI,
  and \TypeDataI.

    (\TypeDeclare is impossible because a $\xdeclare$ is not a value.)

      \begin{itemize}
            \DerivationProofCase{\TypeSub}
                 {\Sigma; \cdot \entails v : B'
                   \\
                   \Sigma \entails B' \subtype B
                 }
                 {\Sigma; \cdot \entails v : B}

                 \begin{llproof}
                   \ePf{\Sigma} {B' \subtype B}   {Subderivation}
                   \ePf{\Sigma; \cdot} {v : B'}   {Subderivation}
                   \ePf{\Sigma} {B \subtype s}   {Given}
                   \ePf{\Sigma} {B' \subtype s}   {By \Lemmaref{lem:subtyping-transitivity}}
                \end{llproof}

                The result follows by i.h.

            \DerivationProofCase{\TypeSectI}
                 {\Sigma; \cdot \entails v : B_1
                  \\
                  \Sigma; \cdot \entails v : B_2
                 }
                 {\Sigma; \cdot \entails v : (B_1 \sectty B_2)}

                \begin{llproof}
                   \ePf{\Sigma}{(B_1 \sectty B_2) \subtype s}   {Given}
                   \ePf{\Sigma}{B_k \subtype s}   {By inversion (\SubSectL{k})}
                   \ePf{\Sigma; \cdot}{v : B_k}   {Subderivation}
                \end{llproof}

                The result follows by i.h.

             \DerivationProofCase{\TypeDataI}
                  {\Sigma \entails c : A \arr s_0
                   \\
                   \Sigma; \cdot \entails v' : A}
                  {\Sigma; \cdot \entails \underbrace{c(v')}_{v} : s_0}

                  \begin{llproof}
                  \Hand  \eqPf{v}{c(v')}   {Above}
                  \Hand  \ePf{\Sigma; \cdot} {v' : A}   {Subderivation}
                      \ePf{\Sigma} {c : A \arr s_0}   {Subderivation}
                      \inPf{(c : A \arr t)}{\Sigma} {By inversion (\ConArr)}
                      \ePf{\Sigma} {t \subsort s_0}  {\ditto}
                      \ePf{\Sigma} {s_0 \subtype s}  {Given}
                      \ePf{\Sigma} {s_0 \subsort s}  {By inversion (\SubSort)}
                  \Hand   \ePf{\Sigma} {t \subsort s}  {By \SubsortTrans}
                  \end{llproof}
      \end{itemize}

  \item \textbf{Part (2)}:

        The only possible cases are \TypeSub,
      \TypeSectI,
      and \TypeArrI.

          \begin{itemize}
               \ProofCasesRules{\TypeSub, \TypeSectI%
               }
                    Similar to the respective cases for part (1).

              \DerivationProofCase{\TypeArrI}
                   {\Sigma; \Gamma, x : B_1 \entails e : B_2
                   }
                   {\Sigma; \Gamma \entails \underbrace{\Lam{x} e}_{v} : (B_1 \arr B_2)}

                   \begin{llproof}
                     \ePf{\Sigma} {(B_1 \arr B_2) \subtype (A_1 \arr A_2)}   {Given}
                   \Hand  \ePf{\Sigma} {A_1 \subtype B_1}   {By inversion}
                     \ePf{\Sigma} {B_2 \subtype A_2}   {\ditto}
                     \ePf{\Sigma; \Gamma, x : B_1}{e : B_2}  {Subderivation}
                   \Hand  \ePf{\Sigma; \Gamma, x : B_1}{e : A_2}  {By \TypeSub}
                   \end{llproof}

          \end{itemize}
          
  \item \textbf{Part (3)}:  Similar to part (2), but with \TypeUnitI instead of \TypeArrI.

  \item \textbf{Part (4)}:  Similar to part (2), but with \TypeProdI instead of \TypeArrI.
  \qedhere
  \end{itemize}
\end{proof}

\subsection{Operational semantics lemmas}

\begin{restatable}{lemma}{lemvaluesdontstep}
\XLabel{lem:values-dont-step}
  If $e$ is a value then there exists no $e'$ such that $e \step e'$.
\end{restatable}
\begin{proof}
  By induction on $e$.
\end{proof}

\subsection{Type preservation and progress}

\begin{restatable}[Pattern intersection]{lemma}{lempatternintersection}
\XLabel{lem:pattern-intersection}
~\\
  If
  $\vmatch{p_1}{v}{\theta_1}$
  and
  $\vmatch{p_2}{v}{\theta_2}$
  then
  $\vmatch{p_1 \patsect p_2}{v}{\theta}$.
\end{restatable}
\begin{proof}  By mutual induction on $p_1$ and $p_2$.

  \begin{itemize}
  \ProofCaseThing{p_1 = \emptypattern  \OR  p_2 = \emptypattern}
       Impossible: $\vmatch{\emptypattern}{v}{\dots}$ is not derivable.

  \ProofCaseThing{p_1 = \wildcard  \OR  p_2 = \wildcard}

      Consider the $p_1 = \wildcard$ case; the $p_2 = \wildcard$ case is similar.
      It is given that $\vmatch{p_2}{v}{\theta_2}$.
      We have 
      $p_1 \patsect p_2 = \wildcard \patsect p_2 = p_2$,
      so $\vmatch{p_1 \patsect p_2}{v}{\theta}$ (letting $\theta = \theta_2$).

  \ProofCaseThing{p_1 = p_{11} \patunion p_{12}}
     
      \begin{llproof}
        \vmatchPf{p_{11} \patunion p_{12}}{v}{\theta_1}   {Given}
      \end{llproof}

      By inversion (\MatchOr), either $p_{11}$ or $p_{12}$ is matched.
      Suppose the former; the latter is similar.

      \begin{llproof}
        \vmatchPf{p_{11}}{v}{\theta_1}   {By inversion (\MatchOr)}
        \vmatchPf{p_{11} \patsect p_2}{v}{\theta}   {By i.h.}
        \vmatchPf{(p_{11} \patsect p_2) \patunion (p_{12} \patsect p_2)}{v}{\theta}   {By \MatchOr}
        \eqPf{(p_{11} \patunion p_{12}) \patsect p_2}{(p_{11} \patsect p_2) \patunion (p_{12} \patsect p_2)}   {By def.\ of $\patsect$}
      \Hand  \vmatchPf{p_1 \patsect p_2}{v}{\theta}   {By above equality}
      \end{llproof}

  \ProofCaseThing{p_2 = p_{21} \patunion p_{22}}
      Similar to the previous case.

  \ProofCaseThing{p_1 = c(p_1')}

      By inversion on $\vmatch{c(p_1')}{v}{\theta_1}$, we have $v = c(v')$
      and $\vmatch{p_1'}{v'}{\theta_1}$.

      We already dealt with the cases for $p_2$ being $\emptypattern$, $\wildcard$
      or a $\patunion$.
      So
      by inversion on $\vmatch{p_2}{c(v')}{\theta_2}$,
      we have $p_2 = c(p_2')$
      and $\vmatch{p_2'}{v'}{\theta_2}$.

      By the definition of $\patsect$, we have $c(p_1') \patsect c(p_2') = c(p_1' \patsect p_2')$.
      By i.h., $\vmatch{p_1' \patsect p_2'}{v'}{\theta'}$.
      
      By \MatchCon, $\vmatch{c(p_1' \patsect p_2')}{c(v')}{\theta'}$.

  \ProofCaseThing{p_1 = \Pair{p_{11}}{p_{12}}}

      By inversion on $\vmatch{ \Pair{p_{11}}{p_{12}}}{v}{\theta_1}$, we have $v = \Pair{v_1}{v_2}$
      and $\vmatch{p_{11}}{v_1}{\theta_{11}}$
      and $\vmatch{p_{12}}{v_2}{\theta_{12}}$
      where $\theta_1 = \theta_{11} \composesubst \theta_{12}$.
      
      We already dealt with the cases for $p_2$ being $\emptypattern$, $\wildcard$
      or a $\patunion$.  So by inversion on $\vmatch{p_2}{\Pair{v_1}{v_2}}{\theta_2}$,
      we have 
      $\vmatch{p_{21}}{v_1}{\theta_{21}}$
      and $\vmatch{p_{22}}{v_2}{\theta_{22}}$
      where $\theta_2 = \theta_{21} \composesubst \theta_{22}$.
      
      By the definition of $\patsect$, we have
      $\Pair{p_{11}}{p_{12}}
      \patsect
      \Pair{p_{21}}{p_{22}})
      =
      \Pair{p_{11} \patsect p_{21}}{p_{21} \patsect p_{22}}$.
      
      By i.h., $\vmatch{p_{11} \patsect p_{21}}{v_1}{\theta_1'}$.
      By i.h., $\vmatch{p_{12} \patsect p_{22}}{v_2}{\theta_2'}$.
      
      By \MatchCon, $\vmatch{\Pair{p_{11} \patsect p_{21}}{p_{21} \patsect p_{22}}}{c(v')}{\theta_1' \composesubst \theta_2'}$.
  \qedhere
  \end{itemize}
\end{proof}

\begin{restatable}[Excluded middle for matching]{lemma}{lemexcludedmiddlematch}
\XLabel{lem:excluded-middle-match}
~\\
  If $\novmatch{p}{v}$ then there exists no $\theta$
  s.t.\ 
  $\vmatch{p}{v}{\theta}$.
\end{restatable}
\begin{proof}
  By induction on the given derivation of $\novmatch{p}{v}$.

  In all cases, at most one match rule could plausibly be applied,
  \eg for \NomatchAs, only \MatchAs has a conclusion of the right
  form.  In the \NomatchEmpty case, where $p = \emptypattern$,
  no match rule has a conclusion of the right form.

  \begin{itemize}
    \DerivationProofCase{\NomatchEmpty}
         {}
         {\novmatch{\emptypattern}{v}}
         
         The match rules are directed by the syntax of the pattern,
         and there is no match rule with $\emptypattern$ in its conclusion.
    
    \DerivationProofCase{\NomatchAs}
         {\novmatch{p_0}{v}}
         {\novmatch{x \As p_0}{v}}

         The only possibly applicable rule is \MatchAs, but
         by i.h., there exists no $\theta$ such that $\vmatch{p_0}{v}{\theta}$.
    
    \DerivationProofCase{\NomatchOr}
         {
           \novmatch{p_1}{v}
           \\
           \novmatch{p_2}{v}
         }
         {\novmatch{p_1 \patunion p_2}{v}}

         By i.h., neither $\novmatch{p_1}{v}{\theta_1}$
         nor $\novmatch{p_2}{v}{\theta_2}$, so whether we choose $k = 1$
         or $k = 2$, we can't apply \MatchOr.

    \ProofCasesRules{\NomatchConHead, \NomatchConInner, \NomatchUnit,
      \NomatchPairHead, \NomatchPairInner}
        
      Straightforward, using the i.h.\ as needed.
  \qedhere
  \end{itemize}
\end{proof}

\begin{restatable}[Choice]{lemma}{lempatternchoice}
\XLabel{lem:pattern-choice}
~\\
  If $\sigjudg{\Sigma}$
  and
  $\Ursig \entails p : \tau$
  and
  $\Sigma; \cdot \entails v : A$
  and
  $\Sigma \entails A \refines \tau$
  \\
  then
  either
  (1)
  $\vmatch{p}{v}{\theta}$
  or
  (2)
  $\novmatch{p}{v}$ and $\vmatch{\patcomplement p}{v}{\emptysubst}$.
\end{restatable}
\begin{proof}   By structural induction on $p$.
  
  \begin{itemize}
      \ProofCaseRule{\PattypeWild}
         By \MatchWild, $\vmatch{p}{v}{\emptysubst}$.

      \ProofCaseRule{\PattypeAs}
         We have $p = (x \As p_0)$.
         By i.h., either $\vmatch{p_0}{v}{\theta}$,
         or $\novmatch{p_0}{v}$ and $\vmatch{\patcomplement p_0}{v}{\cdot}$.

         If $\vmatch{p_0}{v}{\theta}$,
         then by \MatchAs, $\vmatch{x \As p_0}{v}{\theta}$.

         Otherwise, by \NomatchAs, $\novmatch{x \As p_0}{v}$.
         Above, we obtained $\vmatch{\patcomplement p_0}{v}{\cdot}$;
         by the definition of $\patcomplementsym$, we have
         $\patcomplement (x \As p_0) = \patcomplement p_0$,
         so $\vmatch{\patcomplement (x \As p_0)}{v}{\cdot}$.

      \ProofCaseRule{\PattypeEmpty}
         By \NomatchEmpty, $\novmatch{p}{v}$.
         By the definition of $\patcomplementsym$, we have
         $\patcomplement \emptypattern = \wildcard$.

         By \MatchWild, $\vmatch{\wildcard}{v}{\emptysubst}$.

      \DerivationProofCase{\PattypeOr}
         {
           \Sigma \entails p_1 : \tau
           \\
           \Sigma \entails p_2 : \tau
         }
         {\Sigma \entails (p_1 \patunion p_2) : \tau}

         If the i.h.\ gives $\vmatch{p_1}{v}{\theta}$,
         then by \MatchOr, $\vmatch{p_1 \patunion p_2}{v}{\theta}$.

         Otherwise, we have $\novmatch{p_1}{v}$
         and
         $\vmatch{\patcomplement p_1}{v}{\theta_1'}$:

         \begin{itemize}
         \item 
             If the i.h.\ on $\Sigma \entails p_2 : \tau$
             gives $\vmatch{p_2}{v}{\theta}$,
             then by \MatchOr, $\vmatch{p_1 \patunion p_2}{v}{\theta}$.

         \item
             Otherwise, we have $\novmatch{p_2}{v}$
             and
             $\vmatch{\patcomplement p_2}{v}{\theta_2'}$.

             By \NomatchOr, $\novmatch{p_1 \patunion p_2}{v}$.

             By the definition of $\patcomplementsym$, we have
             $\patcomplement (p_1 \patunion p_2)
             = (\patcomplement p_1) \patsect (\patcomplement p_2)$.

             By \Lemmaref{lem:pattern-intersection},
             $\vmatch{(\patcomplement p_1) \patsect (\patcomplement p_2)}{v}{\cdot}$,
             which was to be shown.
         \end{itemize}

      \ProofCaseRule{\PattypeUnit}

          By \Lemmaref{lem:inversion} (3), $v = \unit$.

          By \MatchUnit, $\vmatch{\unit}{\unit}{\emptysubst}$.

      \ProofCaseRule{\PattypePair}

          We have $p = \Pair{p_1}{p_2}$ and $\tau = \tau_1 * \tau_2$.

          By \Lemmaref{lem:inversion} (4), $v = \Pair{v_1}{v_2}$
          and $\Sigma; \cdot \entails v_1 : B_1$
          and $\Sigma; \cdot \entails v_2 : B_2$
          where $\Sigma \entails B_1 \subtype A_1$
          and $\Sigma \entails B_2 \subtype A_2$.

          If the i.h.\ gives (1), then we have $\vmatch{p_1}{v_1}{\theta_1}$:

          \begin{itemize}
          \item By i.h.\ ($p_2$), we have either (1) $\vmatch{p_2}{v_2}{\theta_2}$
            or (2) $\novmatch{p_2}{v_2}$ and $\vmatch{\patcomplement p_2}{v_2}{\theta_2'}$.

            If (1), apply \MatchPair.

            If (2), then:
            By \NomatchPairInner, $\novmatch{\Pair{p_1}{p_2}}{v}$.

            By \MatchWild and \MatchPair, $\vmatch{\Pair{\wildcard}{\patcomplement p_2}}{v}{\theta_2'}$.
            By \MatchOr, $\vmatch{\Pair{\patcomplement p_1}{\wildcard} \patunion \Pair{\wildcard}{\patcomplement p_2}}{v}{\theta_2'}$.  By the definition of $\patcomplement$,
            this is $\vmatch{\patcomplement \Pair{p_1}{p_2}}{\Pair{v_1}{v_2}}{\theta_2'}$.
          \end{itemize}
          
          Otherwise, we have $\novmatch{p_1}{v_1}$ and $\vmatch{\patcomplement p_1}{v_1}{\theta_1'}$.

          \begin{itemize}
          \item By rule \NomatchPairInner, $\novmatch{\Pair{p_1}{p_2}}{v}$.

            By \MatchWild and \MatchPair, $\vmatch{\Pair{\patcomplement p_1}{\wildcard}}{v}{\theta_1'}$.

            By \MatchOr, $\vmatch{\Pair{\patcomplement p_1}{\wildcard} \patunion \Pair{\wildcard}{\patcomplement p_2}}{\Pair{v_1}{v_2}}{\theta_1'}$.

            By the definition of $\patcomplement$,
            this is
            $\vmatch{\patcomplement \Pair{p_1}{p_2}}{\Pair{v_1}{v_2}}{\theta_1'}$.
          \end{itemize}

      \ProofCaseRule{\PattypeCon}

           We have $p = c(p_0)$.
           
           By \Lemmaref{lem:inversion} (1), $v = c_0(v_0)$
           and $\Sigma; \cdot \entails v_0 : A_0$
           and $(c_0 : A_0 \arr s_0) \in \Sigma$ and $\Sigma \entails s_0 \subsort s$.
           
           If $c \neq c_0$, then:

               \begin{llproof}
               \Hand  \novmatchPf{c(p_0)}{c_0(v_0)}   {By \NomatchConHead}
                   \vmatchPf{\wildcard}{v_0}{\cdot}   {By \MatchWild}
                   \vmatchPf{c_0(\wildcard)}{c_0(v_0)}{\cdot}   {By \MatchCon}
                   \vmatchPf{c(\patcomplement p_0) \patunion \dots \patunion c_0(\wildcard) \patunion \dots}{c_0(v_0)}{\cdot}   {By \MatchOr}
               \Hand  \vmatchPf{\patcomplement\big(c(p_0)\big)}{c_0(v_0)}{\cdot}    {By def.\ of $\patcomplement$}
               \end{llproof}

           If $c = c_0$, then:

               \begin{llproof}
               1~~  \ePf{\Ursig} {p_0 : \tau_0}    {By inversion on \PattypeCon}
                 \ePf{\Ursig} {c : (\tau_0 \arr d)}    {\ditto}
                 \proofsep
               2~~  \ePf{\Sigma; \cdot} {v_0 : A_0}    {Above}
                 \inPf{(c_0 : A_0 \arr s_0)} {\Sigma}    {Above}
                 \contypejudgPf{\Sigma}{c : A_0 \arr s_0}    {By inversion on $\sigjudg{\Sigma}$}
               3~~  \ePf{\Sigma} {A_0 \refines \tau_0}   {By inversion on \ContypeArr}
               \end{llproof}

               By i.h.\ on (1, 2, 3), either:

                     \begin{itemize}
                       \item (1) $\vmatch{p_0}{v_0}{\theta}$:

                         \begin{llproof}
                           \Hand  \vmatchPf{c(p_0)}{c(v_0)}{\theta}   {By \MatchCon}
                         \end{llproof}

                       \item (2) $\novmatch{p_0}{v_0}$ and $\vmatch{\patcomplement p_0}{v_0}{\cdot}$:

                           \begin{llproof}
                           \Hand  \novmatchPf{c(p_0)}{c(v_0)}   {By \NomatchConInner}
                             \vmatchPf{c(\patcomplement p_0)}{c(v_0)}{\cdot}   {By \MatchCon}
                             \decolumnizePf
                             \vmatchPf{c(\patcomplement p_0) \patunion c_1(\wildcard) \patunion \dots \patunion c_n(\wildcard)}{c(v_0)}{\cdot}   {By \MatchOr}
                           \Hand  \vmatchPf{\patcomplement \big(c(p_0)\big)}{c(v_0)}{\cdot}   {By def.\ of $\patcomplement$}
                           \end{llproof}
                     \qedhere
                     \end{itemize}
  \end{itemize}
\end{proof}

\thmintersect*
\begin{proof}
  By structural induction on $p$.

  Case-analyze the clause of the definition of the $\intersectx$ function:

  \begin{itemize}
  \item Case: $p = \wildcard$:

    \begin{llproof}
      \vmatchPf{\wildcard}{v}{\theta}   {Given}
      \eqPf{\theta}{\cdot}   {By inversion (\MatchWild)}
      \eqPf{\vec{\Bdec}}{\{(\cdot; \cdot \entails A)\}}   {By def.\ of $\intersectx$}
      \ePf{\Sigma; \cdot} {v : A}    {Given}
    \Hand \ePf{\Sigma; \cdot} {v : B}    {$B = A$}
    \Hand \ePf{\cdot} {\cdot : \Gamma'}    {By \SubstEmpty ($\Gamma' = \cdot$)}
    \Hand \ePf{\Sigma} {B \subtype A}    {By \Lemmaref{lem:subtyping-reflexivity}}
    \end{llproof}

    \smallskip

  \item Case: $p = \emptypattern$:

    We have $\vmatch{\emptypattern}{v}{\theta}$, which is not derivable:
    this case is impossible.

    \smallskip

  \item Case: $p = x \As p_0$:

    \begin{llproof}
      \vmatchPf{x \As p_0}{v}{\theta}   {Given}
      \vmatchPf{p_0}{v}{\theta}   {By inversion (\MatchAs)}
      \ePf{\Sigma; \cdot}{v : A}   {Given}
      \proofsep
      \eqPf{\intersect{\Sigma}{A}{x \As p_0}}
           {\vec{\Bdec}}
           {Given}
      \decolumnizePf
      \Pf{\Gamma' = (\Gamma_0', x : B)
        }{\AND}
        { \intrsctresult{\Sigma'}{\Gamma_0'}{B} \in \intersect{\Sigma}{A}{p_0}
        }
           {for all $\intrsctresult{\Sigma'}{\Gamma'}{B} \in \vec{\Bdec}$}
           \trailingjust{By definition of $\intersectx$}
      \decolumnizePf
    \Hand  \ePf{\Sigma\argle; \cdot}{v : B}   {By i.h.}
      \inPf{\intrsctresult{\Sigma'}{\Gamma_0'}{B}}{\vec{\Bdec}}   {\ditto}
      \ePf{\cdot}{\theta_0 : \Gamma_0'}   {\ditto}
    \Hand \ePf{\Sigma} {B \subtype A}    {\ditto}
      \proofsep
    \Hand  \ePf{\cdot; \cdot}{(\theta_0, v/x) : (\Gamma_0', x : B)}   {By \SubstVar}
    \end{llproof}

  \item Case: $p = p_1 \patunion p_2$:

    \begin{llproof}
        \eqPf{\intersect{\Sigma}{A}{p_1 \patunion p_2}}{\vec{\Bdec}}   {Given}
        \decolumnizePf
        \vmatchPf{p_1 \patunion p_2}{v}{\theta}   {Given}
        \vmatchPf{p_1}{v}{\theta}   {By inversion (\MatchOr) wlog}
        \inPf{\intrsctresult{\Sigma'}{\Gamma'}{B}}{\intersect{\Sigma}{A}{p_1}}  {By i.h.}
    \Hand   \ePf{\Sigma\argle; \cdot}{v : B}   {\ditto}
    \Hand   \ePf{\Sigma\argle; \cdot}{\theta : \Gamma'}  {\ditto}
    \Hand   \ePf{\Sigma} {B \subtype A}    {\ditto}
       \decolumnizePf
        \subseteqPf{\intersect{\Sigma}{A}{p_1}}{\vec{\Bdec}}   {By def.\ of $\intersectx$}
    \Hand    \inPf{\intrsctresult{\Sigma'}{\Gamma'}{B}}{\vec{\Bdec}}  {By a property of $\in$}
    \end{llproof}

  \item Case: $p = \Pair{p_1}{p_2}$:

      Throughout this case, interpret $k$ as universally quantified.  For example,
      \Lemmaref{lem:inversion} (4) shows both $\Sigma; \cdot \entails v_1 : A_1$
      and $\Sigma; \cdot \entails v_2 : A_2$.

      \begin{llproof}
              \vmatchPf{\Pair{p_1}{p_2}}{v}{\theta}   {Given}
              \eqPf{v}{\Pair{v_1}{v_2}}  {By inversion (\MatchPair)}
              \eqPfParenR{\theta}{\theta_1, \theta_2}  {\ditto}
              \vmatchPf{p_k}{v_k}{\theta_k}  {\ditto}
              \ePf{\Sigma; \cdot} {\Pair{v_1}{v_2} : A_1 * A_2}   {Given}
              \ePf{\Sigma; \cdot} {v_k : A_k}   {By \Lemmaref{lem:inversion} (4)}
              \proofsep
              \inPf
                   {\intrsctresult{\Sigma_k}{\Gamma_k}{B_k}}
                   {\intersect{\Sigma}{A_k}{p_k}}
                   {By i.h.}
              \ePf{\Sigma\bargle; \cdot}{v_k : B_k}   {\ditto}
              \ePf{\Sigma\bargle; \cdot}{\theta_k : \Gamma_k}   {\ditto}
              \ePf{\Sigma} {B_k \subtype A_k}    {\ditto}
              \proofsep
      \Hand   \ePf{\Sigma\arglebargle; \cdot}{\Pair{v_1}{v_2} : B_1 * B_2}   {By \TypeProdI}
      \Hand   \ePf{\Sigma\arglebargle; \cdot}{(\theta_1, \theta_2) : (\Gamma_1, \Gamma_2)}   {By properties of substitution}
       \Hand       \ePf{\Sigma} {(B_1 * B_2) \subtype (A_1 * A_2)}    {By \SubProd}
      \end{llproof}

      \smallskip
    
  \item Case: $p = c(p_0)$:

      \begin{llproof}
              \vmatchPf{c(p_0)}{v}{\theta}   {Given}
              \eqPf{v}{c(v_0)}  {By inversion (\MatchCon)}
              \vmatchPf{p_0}{v_0}{\theta}  {\ditto}
              \ePf{\Sigma; \cdot} {c(v_0) : s}   {Given}
              \ePf{\Sigma; \cdot} {v_0 : A_0}   {By \Lemmaref{lem:inversion} (1)}
              \inPf{(c : A_0 \arr s_c)} {\Sigma} {\ditto}
        \Hand    \ePf{\Sigma} {s_c \subsort s}   {\ditto}
              \decolumnizePf
              \eqPf{\intersect{\Sigma}{s}{c(p_0)}}{\vec{\Bdec}}  {Given}
              \inPf{\intrsctresult{\Sigma'}{\Gamma'}{B_0}}{\intersect{\Sigma}{A_0}{p_0}}
                         {By def.\ of $\intersectx$}
              \ePf{\Sigma\argle; \cdot} {v_0 : B_0} {By i.h.}
              \ePf{\Sigma\argle; \cdot} {\theta : \Gamma'} {\ditto}
              \ePf{\Sigma\argle}{B_0 \type}  {\ditto}
              \ePf{\Sigma\argle} {B_0 \subtype A_0}  {By \ditto}
              \ePf{\Sigma\argle; \Gamma'} {v_0 : A_0} {By \TypeSub}
              \decolumnizePf
              \inPf{\intresult{\Gamma'}{s_c}}{\intersect{\Sigma}{s}{c(p_0)}}
                         {By def.\ of $\intersectx$}
              \ePf{\Sigma\argle; \cdot} {c(v_0) : s_c}   {By \TypeDataI}
      \Hand        \ePf{\Sigma\argle; \cdot} {c(v_0) : s}   {By \TypeSub}
      \Hand   \ePf{\Sigma\argle; \cdot} {\theta : \Gamma'}   {Above}
      \end{llproof}
  \qedhere
  \end{itemize}
\end{proof}

  \begin{lemma}[Match preservation]
\Label{lem:match-preservation}
 ~\\
  If
  $\sigjudg{\Sigma}$
  and
  $\Sigma \entails A \type$
  and
  $\vmatch{p}{v}{\theta}$
  and
  $\matchassn \Sigma \cdot A p
   \entails
   ms : D$
 \\
  and
  $ms \step e'$
  \\
  then $\Sigma; \cdot \entails e' : D$.
\end{lemma}
\begin{proof}
  By induction on the derivation of $\dots \entails ms : D$.

  \begin{itemize}
  \ProofCaseRule{\TypeMsEmpty}

      \begin{llproof}
        \ePf{\matchassn \Sigma \cdot A p}  {\emptyms : D}  {Given}
        \eqPf{\intersect{\Sigma}{A}{p}} {\emptyset}  {Subderivation}
        \vmatchPf p v \theta   {Given}
      \end{llproof}

      By \Theoremref{thm:intersect},
      there exists a track in $\intersect{\Sigma}{A}{p}$ such that certain conditions hold.
      But $\intersect{\Sigma}{A}{p} = \emptyset$, a contradiction.
      Thus, this case is impossible.

  {\small\def\zzpreDerivationProofCase{\hspace{-6ex}}\DerivationProofCase{\TypeMs}
           {
             \arrayenvbl{
               \Sigma \entails A \refines \tau
               \\
               \Ursig \entails p_1 : \tau
             }
             \hspace{-1.5ex}
             \\
             \arrayenvbl{
                 \text{for all $\intresult{\Gamma'}{B}$}
                 \\
                 \text{~~$\in \intersect{\Sigma}{A}{p \patsect p_1}$:}
                 \\
                 ~~~~\Sigma; \cdot, \Gamma'
                              \entails
                              e_1 : D
             }
             \hspace{-2.0ex}
             \\
             \matchassn \Sigma \cdot A {(p \patsect \patcomplement p_1)}
             \entails
             ms' : D
           }
           {
             \matchassn \Sigma \cdot A p
             \entails
             \underbrace{\big( (\Match{p_1}{e_1}) \matchor ms'\big)}_{ms} : D
           }
}

           \begin{llproof}
             \sigjudgPf{\Sigma}   {Given}
             \ePf{\Ursig}{p : \tau}   {Subderivation}
             \ePf{\Sigma; \cdot}{v : A}   {Given}
             \ePf{\Sigma}{A \refines \tau}   {Subderivation}
           \end{llproof}

     If the derivation of $ms \step e'$ was concluded by \StepMatch, then
     $e' = [\theta_1]e_1$.

      \begin{llproof}
        \sigjudgPf{\Sigma}   {Given}
        \ePf{\Sigma}{v : A}   {Given}
        \ePf{\Sigma}{A \type}   {Given}
        \vmatchPf{p}{v}{\theta}    {Given}
        \vmatchPf{p_1}{v}{\theta_1}    {Subderivation of \StepMatch}
        \vmatchPf{(p \patsect p_1)}{v}{\theta'}    {By \Lemmaref{lem:pattern-intersection}}
        \inPf{\intresult{\Gamma_1}{B}}{\intersect{\Sigma}{A}{p \patsect p_1}}    {By \Theoremref{thm:intersect}}
        \ePf{\Sigma; \cdot}{v : B}   {\ditto}
        \ePf{\Sigma; \cdot}{\theta' : \Gamma_1}   {\ditto}
        \decolumnizePf
        \ePf{\Sigma; \cdot, \Gamma_1} {e_1 : D}  {Subderivation}
      \Hand        \ePf{\Sigma; \cdot} {[\theta']e_1 : D}  {By \Lemmaref{lem:multiple-value-substitution}}
      \end{llproof}

     Otherwise, the derivation was concluded by \StepElse, where
     $\novmatch{p_1}{v}$.

      \begin{llproof}
          \stepPf{ms'}{e'}   {Subderivation}
          \ePf{   \matchassn \Sigma \cdot A {(p \patsect \patcomplement p_1)}}
              {ms' : D}
              {Subderivation}
           \decolumnizePf
           \vmatchPf{p}{v}{\theta}   {Given}
           \vmatchPf{\patcomplement p_1}{v}{\cdot}   {Above}
           \vmatchPf{p \patsect (\patcomplement p_1)}{v}{\theta'}    {By \Lemmaref{lem:pattern-intersection}}
      \Hand    \ePf{\Sigma; \cdot}{e' : D}    {By i.h.}
      \end{llproof}
  \qedhere  
  \end{itemize}
\end{proof}

The preservation result allows for a longer signature, to model entering the scope of
a $\xdeclare$ expression or the arms of a \textkw{match}.
We implicitly assume that throughout the given typing derivation, all types
are well-formed under the local signature: whenever
we have a subderivation of $\Sigma; \Gamma \entails e' : B$, it is the case that
$\Sigma \entails B \type$.

\thmpreservation*
\begin{proof}
  By induction on the derivation of $\Sigma; \cdot \entails e : A$.

  The rules \TypeVar, \TypeUnitI, \TypeArrI, and \TypeSectI
  can only type values,
  which cannot step (Lemma \ref{lem:values-dont-step}), contradicting
  $e \step e'$: these cases are impossible.

  In most cases, the conditions about $\Sigma'$
  follow directly from the i.h.
  
  \begin{itemize}

    \ProofCaseRule{\TypeArrE}
        We have $e_1\,e_2 \step e'$.

       \begin{itemize}
       \item If $e_1 \step e_1'$, use the i.h.\ on the first subderivation,
         \Theoremref{thm:weakening} (v) on the second subderivation,
         and apply \TypeArrE.
       \item If $e_2 \step e_2'$, use the i.h.\ on the second subderivation,
          \Theoremref{thm:weakening} (v) on the first subderivation,
          then apply \TypeArrE.
       \item If $e_1 = \Lam{x} e_0$ and $e_2$ is a value,
         use \Lemmaref{lem:inversion} (2) on $\Sigma; \cdot \entails \Lam{x} e_0 : (B \arr A)$
         to get $\Sigma; x : B \entails e_0 : A$.

         Then use \Lemmaref{lem:value-substitution} to get
         $\Sigma; \cdot \entails [e_2/x]e_0 : A$, which was to be shown
         (letting $\Sigma' = \cdot$).
       \end{itemize}

    \ProofCaseRule{\TypeProdI}
       Either $e_1 \step e_1'$ or $e_2 \step e_2'$.  Apply the i.h.\ to the appropriate
       subderivation, use the weakening lemma on the other subderivation, and
       apply \TypeProdI.
    
    \DerivationProofCase{\TypeSub}
         {\Sigma; \cdot \entails e : B
           \\
           \Sigma \entails B \subtype A
         }
         {\Sigma; \cdot \entails e : A}
         
         \begin{llproof}
           \ePf{\Sigma; \cdot}{e : B}  {Subderivation}
           \ePf{\Sigma, \Sigma'; \cdot}{e' : B}  {By i.h.}
           \ePf{\Sigma} {B \subtype A}  {Subderivation}
           \ePf{\Sigma, \Sigma'} {B \subtype A}  {By \Lemmaref{lem:weakening-low} (iii)}
         \Hand  \ePf{\Sigma, \Sigma'; \cdot}{e' : A}  {By \TypeSub}
         \end{llproof}

    \DerivationProofCase{\TypeDataI}
          {\Sigma \entails c : B \arr s
           \\
           \Sigma; \cdot \entails e_0 : B}
          {\Sigma; \cdot \entails \underbrace{c(e_0)}_{e} : \underbrace{s}_{A}}

        By inversion on $c(e_0) \step e'$ we have $e' = c(e_0')$ and $e_0 \step e_0'$.
          
        \begin{llproof}
          \ePf{\Sigma; \cdot} {e_0 : B}  {Subderivation}
          \ePf{\Sigma, \Sigma'; \cdot} {e_0' : B}  {By i.h.}
          \proofsep
          \ePf{\Sigma} {c : B \arr s}   {Subderivation}
          \ePf{\Sigma, \Sigma'} {c : B \arr s}   {By \Lemmaref{lem:weakening-low} (iv)}
          \proofsep
        \Hand  \ePf{\Sigma, \Sigma'} {c(e_0') : s}   {Subderivation}
        \end{llproof}
    
     \DerivationProofCase{\TypeDataE}
         {\Sigma; \cdot \entails e : B
          \\
          \matchassn{\Sigma}{\cdot}{B}{\wildcard}
          \entails
          ms : A
         }
         {\Sigma; \cdot \entails (\Case{e}{ms}) : A}

         \begin{llproof}
           \sigjudgPf{\Sigma}   {Given}
           \ePf{\Sigma}{B \type}   {By implicit assumption}
           \vmatchPf{\wildcard}{v}{\cdot}   {By \MatchWild}
           \ePf{  \matchassn{\Sigma}{\cdot}{B}{\wildcard}} {ms : A}  {Subderivation}
         \Hand  \ePf{\Sigma; \cdot} {e' : A}     {By \Lemmaref{lem:match-preservation}}
         \Hand  \stepPf{\Case{e}{ms}} {e'}   {By \StepCase}
         \end{llproof}

     \DerivationProofCase{\TypeDeclare}
          {
                \sigjudg{(\Sigma, \Sigma')}
            \\
                \Sigma \entails A \type
                \\
                \Sigma, \Sigma'; \cdot \entails e_0 : A
          }
          {\Sigma; \cdot \entails (\declare{\Sigma'}{e_0}) : A}
       
       We have $(\declare{\Sigma'}{e_0}) \step e'$. 
       By inversion (\StepDeclare), $e' = e_0$.

       \begin{llproof}
       \Hand  \sigjudgPf{(\Sigma, \Sigma')}   {Premise}
         \ePf{\Sigma, \Sigma'; \cdot} {e_0 : A}   {Subderivation}
       \Hand  \ePf{\Sigma, \Sigma'; \cdot} {e_0 : A}   {$e' = e_0$}
       \end{llproof}
  \qedhere
  \end{itemize}
\end{proof}

\begin{lemma}[Match progress]
\Label{lem:match-progress}
 ~\\
  If
  $\sigjudg{\Sigma}$
  and
  $\vmatch{p}{v}{\theta}$
  and
  $\matchassn \Sigma \cdot A p
   \entails
   ms : D$
  then
  $ms \step e'$.
\end{lemma}
\begin{proof}
  By induction on the derivation of $\dots \entails ms : D$.

  \begin{itemize}
  \ProofCaseRule{\TypeMsEmpty}

      By the same reasoning as in the \TypeMsEmpty case of \Lemmaref{lem:match-preservation}, this case is impossible.

  {\def\zzpreDerivationProofCase{\hspace*{-5ex}}\small \DerivationProofCase{\TypeMs}
           {
             \arrayenvbl{
               \Sigma \entails A \refines \tau
               \\
               \Ursig \entails p_1 : \tau
             }
             \hspace{-2.0ex}
             \\
             \arrayenvbl{
                 \text{for all $\intresult{\Gamma'}{B}$}
                 \\ \text{$~~\in \intersect{\Sigma}{A}{p \patsect p_1}$:}
                 \\
                 ~~~~\Sigma; \cdot, \Gamma'
                              \entails
                              e_1 : D
             }
             \hspace{-2.0ex}
             \\
             \matchassn \Sigma \cdot A {(p \patsect \patcomplement p_1)}
             \entails
             ms' : D
           }
           {
             \matchassn \Sigma \cdot A p
             \entails
             \underbrace{\big( (\Match{p_1}{e_1}) \matchor ms'\big)}_{ms} : D
           }
}

           \begin{llproof}
             \sigjudgPf{\Sigma}   {Given}
             \ePf{\Ursig}{p : \tau}   {Subderivation}
             \ePf{\Sigma; \cdot}{v : A}   {Given}
             \ePf{\Sigma}{A \refines \tau}   {Subderivation}
           \end{llproof}

     By \Lemmaref{lem:pattern-choice}, either (1) $\vmatch{p_1}{v}{\theta_1}$,
     or (2) $\novmatch{p_1}{v}$ and $\vmatch{\patcomplement p_1}{v}{\cdot}$.
     
     For case (1), apply \StepMatch.

     For case (2), show $\vmatch{p \patsect \patcomplement p_1}{v}{\theta'}$
     as in the proof of \Lemmaref{lem:match-preservation},
     apply the i.h.\ to $ms'$,
     then apply \StepElse.
  \qedhere  
  \end{itemize}
\end{proof}

\thmprogress*
\begin{proof}
  By induction on the given derivation.

  \begin{itemize}
      \ProofCaseRule{\TypeVar}
         Impossible, because $\Gamma$ is empty.

      \ProofCaseRule{\TypeSub}
         Follows by i.h.\ on the typing subderivation.

      \ProofCasesRules{\TypeArrI, \TypeSectI, %
        \TypeUnitI}
         The rule requires that $e$ is a value.

      \ProofCaseRule{\TypeArrE}  We have $e = e_1\,e_2$.
         By i.h.\ on the subderivation typing $e_1$, either $e_1$ steps
         or $e_1$ is a value:

         \begin{itemize}
         \item If $e_1$ steps, the result follows by \StepContext.
         \item If $e_1$ is a value, then
           by i.h.\ on the subderivation typing $e_2$, either
           $e_2$ steps or $e_2$ is a value.

           In the former case, the result follows by \StepContext.

           In the latter case:

           \begin{llproof}
             \ePf{\Sigma; \cdot}{e_1 : A' \arr A}   {Subderivation}
             \Pf{}{}{e_1\text{~is a value}}   {Above}
             \eqPf{e_1}{(\Lam{x} e_0)}  {By \Lemmaref{lem:inversion} (2)}
           \Hand  \stepPf{(\Lam{x} e_0)\,e_2}{[e_2/x]e_0}  {By \StepBeta}
           \end{llproof}
         \end{itemize}

      \ProofCaseRule{\TypeProdI}  We have $e = \Pair{e_1}{e_2}$.

        By the i.h.\ on the subderivation typing $e_1$, either $e_1$ steps
        or $e_1$ is a value.

        In the former case, the result follows by \StepContext.
        
        In the latter case, use the i.h.\ on the subderivation typing $e_2$.
        If $e_2$ steps, apply \StepPairR.  Otherwise, $\Pair{e_1}{e_2}$ is a value.

      \ProofCaseRule{\TypeDataI}   We have $e = c(e_0)$.

        By the i.h.\ on the subderivation typing $e_0$,
        either $e_0$ steps or $e_0$ is a value.

        In the former case, the result follows by \StepContext.

        In the latter case, $c(e_0)$ is a value.

      \ProofCaseRule{\TypeDataE}  
      
           We have $e = \Case{e_0}{ms}$.

           By the i.h.\ on the subderivation typing $e_0$,
           either $e_0$ steps or $e_0$ is a value.

           In the former case, the result follows by \StepContext.

           In the latter case:

          \begin{llproof}
              \sigjudgPf{\Sigma}  {Given}
              \vmatchPf{\wildcard}{e_0}{\cdot}   {By \MatchWild}
              \ePf{\matchassn \Sigma \cdot B \wildcard} {ms : A}   {Subderivation}
         \Hand  \stepPf{\Case{e_0}{ms}}{e'}  {By \Lemmaref{lem:match-progress}}
          \end{llproof}

      \ProofCaseRule{\TypeDeclare}
          The result follows by \StepDeclare.
  \qedhere
  \end{itemize}
\end{proof}

\subsection{Bidirectional typing: soundness and completeness}

\thmbidirsoundness*

\begin{proof}
  By induction on the given derivation.

  In the \SynAnno case, $e = (e_0 : As)$.  Apply the i.h.\ to get $\Gamma \entails |e_0| : A$.
  Since $|(e_0 : As)| = |e_0|$, we have  $\Gamma \entails |(e_0 : As)| : A$,
  which was to be shown.

  In all other cases, apply the i.h.\ to each subderivation and apply the type
  assignment rule corresponding to the bidirectional rule (\TypeVar for \SynVar,
  \TypeSub for \ChkSub, and so on).
\end{proof}

\thmannotatability*
\begin{proof}
  By induction on the given derivation.

  For most cases:  Use the i.h.\ on each subderivation and apply the corresponding bidirectional rule,
  using part (1) of the i.h.\ for checking premises and part (2) for synthesizing premises.
  If the conclusion of the corresponding bidirectional rule is a checking judgment,
  then part (1) has been shown; part (2) follows by adding an annotation and using \SynAnno.
  Otherwise, the conclusion of the corresponding rule is a synthesis judgment;
  part (2) has been shown, and applying \ChkSub gives part (1).

  For the \TypeSectI case of part (1), we have $A = (A_1 \sectty A_2)$.
  The i.h.\ on the first subderivation
  yields $\Gamma \entails e^1_{\chk} \chk A_1$,
  and on the second subderivation it yields $\Gamma \entails e^2_{\chk} \chk A_2$.
  Let $e^{12}_{\chk}$ be $e$ with all annotations from $e^1_{\chk}$ \emph{and} $e^2_{\chk}$.
  If $e^1_{\chk}$ and $e^2_{\chk}$ have different annotations on the same subterm,
  then $e^{12}_{\chk}$ has their union;
  for example,
  if $e^1_{\chk} = (e_0 : B_1)$ and $e^2_{\chk} = (e_0 : (B_2, B_3))$
  then
  let $e^{12}_{\chk} = (e_0 : (B_1, B_2, B_3))$.
  Applying \ChkSectI gives $\Gamma \entails e^{12}_{\chk} \chk (A_1 \sectty A_2)$,
  which was to be shown.
  
  For the \TypeSectI case of part (2), follow part (1), then add an annotation, yielding
  \[
  \Gamma \entails \big(e^{12}_{\chk} : (A_1 \sectty A_2)\big) \syn (A_1 \sectty A_2)
  \qedhere
  \]
\end{proof}

\subsection{Bidirectional typing: decidability}

\begin{lemma}[Decidability]
\label{lem:decidability}
  Given instantiations of the meta-variables, the following judgments are
  decidable:

  \begin{enumerate}[(1)]
  \item 
  the refinement judgment $\Sigma \entails A \refines \tau$
  \item
  the well-formedness judgments $\typejudg \Sigma A$
  and $\contypejudg{\Sigma}{c : C}$
  \item
  the subsorting judgment $\Sigma \entails s_1 \subsort s_2$
  \item %
  the constructor typing judgment $\Sigma \entails c : C$
  \item
  the subtyping judgment $\Sigma \entails A \subtype B$
  \item %
  the safe extension judgments 
  $\safeextat{\Sigma}{S\sortblock{\blk}}{C}{c}{t}$
  and
  $\blkjudg{\Sigma}{S}{\blk}{\blkelem}$
  \item
  the signature well-formedness judgment $\sigjudg{\Sigma}$
  \item
  the pattern type judgment $\Ursig \entails p : \tau$
  \end{enumerate}

  Moreover, the $\intersectx$ function is computable.
\end{lemma}
\begin{proof}
  For (1) and (2), the type gets smaller in every nontrivial premise.

  For subsorting (3), construct the transitive closure.

  The single rule for (4) depends only on (3).

  For subtyping (5), at least one type gets smaller in each premise.

  The rules for (6) are not genuinely inductive, using only previous judgments.

  For (7), the signature gets smaller in the first premise.

  For (8), the pattern gets smaller in each premise.

  In the definition of $\intersectx$, the pattern gets smaller in each recursion.
\end{proof}

\begin{restatable}[Decidability]{theorem}{thmdecidability}
\label{thm:decidability}
~\\
  Given a signature $\Sigma$, context $\Gamma$ and expression $e$,
  the set of $A$ such that
  $\Sigma; \Gamma \entails e \syn A$ is decidable;
  and, given also a type $B$, 
  the judgment $\Sigma; \Gamma \entails e \chk B$ is decidable.
  
  Moreover, given a signature $\Sigma$, context $\Gamma$, pattern $p$,
  type $A$, matches $ms$ and type $B$, 
  the judgment $
             \matchassn{\Sigma}{\Gamma}{A}{\wildcard}
             \entails
             ms \chk B
             $ is decidable.
\end{restatable}
\begin{proof}
  The auxiliary judgments, such as subtyping (\ChkSub) and well-formedness (\ChkDeclare),
  are decidable by \Lemmaref{lem:decidability}.

  In each premise of each bidirectional typing rule, either

  \begin{enumerate}
  \item the expression gets smaller
    (\SynAnno, \ChkArrI, \SynArrE,
    \ChkProdI, \ChkDataI, \ChkDataE, \ChkDeclare), or
    
  \item the expression is the same; then, either
    
    \begin{enumerate}
    \item the conclusion is checking, the premise is checking, and the type gets smaller
      (\ChkSectI), or

    \item the conclusion is checking, and the premise is synthesizing (\ChkSub), or

    \item we are using \SynSectE{k}.
    \end{enumerate}
  \end{enumerate}

  That is, we can order the problems lexicographically, considering the expression first;
  then, we consider the synthesis problem smaller than the checking problem.
  For the rules \SynSectE{1} and \SynSectE{2},
  observe that, assuming the premise has been derived,
  each rule enumerates one part of the intersection; since type expressions
  are finite, only finitely many types can be so enumerated.
\end{proof}

\Label{apx:PF-LAST}  %

%% file: lr.bbl
\begin{thebibliography}{28}
\providecommand{\natexlab}[1]{#1}
\providecommand{\url}[1]{\texttt{#1}}
\expandafter\ifx\csname urlstyle\endcsname\relax
  \providecommand{\doi}[1]{doi: #1}\else
  \providecommand{\doi}{doi: \begingroup \urlstyle{rm}\Url}\fi

\bibitem[Comon et~al.(2008)Comon, Dauchet, Gilleron, Jacquemard, Lugiez, Tison,
  and Tommasi]{TATA-book}
H.~Comon, M.~Dauchet, R.~Gilleron, F.~Jacquemard, D.~Lugiez, S.~Tison, and
  M.~Tommasi.
\newblock Tree automata techniques and applications.
\newblock \url{https://gforge.inria.fr/frs/download.php/file/10994/tata.pdf},
  2008.
\newblock Release of 18 November 2008.

\bibitem[Cousot and Cousot(1977)]{Cousot77}
Patrick Cousot and Radhia Cousot.
\newblock Abstract interpretation: a unified lattice model for static analysis
  of programs by construction or approximation of fixpoints.
\newblock In \emph{Principles of Programming Languages}, pages 238--252, 1977.

\bibitem[Davies(2005)]{DaviesThesis}
Rowan Davies.
\newblock \emph{Practical Refinement-Type Checking}.
\newblock PhD thesis, Carnegie Mellon University, 2005.
\newblock CMU-CS-05-110.

\bibitem[Davies(2013)]{DaviesSMLCIDREnew}
Rowan Davies.
\newblock {SML} checker for intersection and datasort refinements (pronounced
  ``cider'').
\newblock \url{https://github.com/rowandavies/sml-cidre}, 2013.

\bibitem[Davies and Pfenning(2000)]{Davies00icfpIntersectionEffects}
Rowan Davies and Frank Pfenning.
\newblock Intersection types and computational effects.
\newblock In \emph{ICFP}, pages 198--208, 2000.

\bibitem[Dunfield(2007{\natexlab{a}})]{Dunfield07:Stardust}
Jana Dunfield.
\newblock Refined typechecking with {Stardust}.
\newblock In \emph{Programming Languages meets Program Verification (PLPV
  '07)}, 2007{\natexlab{a}}.

\bibitem[Dunfield(2007{\natexlab{b}})]{DunfieldThesis}
Jana Dunfield.
\newblock \emph{A Unified System of Type Refinements}.
\newblock PhD thesis, Carnegie Mellon University, 2007{\natexlab{b}}.
\newblock CMU-CS-07-129.

\bibitem[Dunfield and Krishnaswami(2013)]{Dunfield13}
Jana Dunfield and Neelakantan~R. Krishnaswami.
\newblock Complete and easy bidirectional typechecking for higher-rank
  polymorphism.
\newblock In \emph{ICFP}, 2013.
\newblock \href{http://arxiv.org/abs/1306.6032}{arXiv:{\tt 1306.6032 [cs.PL]}}.

\bibitem[Dunfield and Pfenning(2003)]{Dunfield03:IntersectionsUnionsCBV}
Jana Dunfield and Frank Pfenning.
\newblock Type assignment for intersections and unions in call-by-value
  languages.
\newblock In \emph{Found. Software Science and Computation Structures (FoSSaCS
  '03)}, pages 250--266, 2003.

\bibitem[Dunfield and Pfenning(2004)]{Dunfield04:Tridirectional}
Jana Dunfield and Frank Pfenning.
\newblock Tridirectional typechecking.
\newblock In \emph{Principles of Programming Languages}, pages 281--292, 2004.

\bibitem[Freeman(1994)]{FreemanThesis}
Tim Freeman.
\newblock \emph{Refinement Types for {ML}}.
\newblock PhD thesis, Carnegie Mellon University, 1994.
\newblock CMU-CS-94-110.

\bibitem[Freeman and Pfenning(1991)]{Freeman91}
Tim Freeman and Frank Pfenning.
\newblock Refinement types for {ML}.
\newblock In \emph{Programming Language Design and Implementation}, pages
  268--277, 1991.

\bibitem[Gentzen(1934)]{Gentzen35}
Gerhard Gentzen.
\newblock Untersuchungen {\"u}ber das logische {Schlie{\ss}en}.
\newblock \emph{Mathematische Zeitschrift}, 39:\penalty0 176--210, 405--431,
  1934.
\newblock English translation, \emph{Investigations into logical deduction}, in
  M. Szabo, editor, \emph{Collected papers of {Gerhard Gentzen}}
  (North-Holland, 1969), pages 68--131.

\bibitem[Harper et~al.(1993)Harper, Honsell, and Plotkin]{Harper93}
Robert Harper, Furio Honsell, and Gordon Plotkin.
\newblock A framework for defining logics.
\newblock \emph{Journal of the ACM}, 40\penalty0 (1):\penalty0 143--184, 1993.

\bibitem[Kennedy(1996)]{KennedyThesis}
Andrew Kennedy.
\newblock \emph{Programming languages and dimensions}.
\newblock PhD thesis, University of Cambridge, 1996.

\bibitem[Koot and Hage(2015)]{Koot15}
Ruud Koot and Jurriaan Hage.
\newblock Type-based exception analysis for non-strict higher-order functional
  languages with imprecise exception semantics.
\newblock In \emph{Proceedings of the 2015 Workshop on Partial Evaluation and
  Program Manipulation}, pages 127--138, 2015.

\bibitem[Lovas(2010)]{LovasThesis}
William Lovas.
\newblock \emph{Refinement Types for Logical Frameworks}.
\newblock PhD thesis, Carnegie Mellon University, 2010.
\newblock CMU-CS-10-138.

\bibitem[Pfenning and Sch{\"u}rmann(1999)]{Pfenning99:Twelf}
Frank Pfenning and Carsten Sch{\"u}rmann.
\newblock System description: {Twelf}---a meta-logical framework for deductive
  systems.
\newblock In \emph{Int'l Conf. Automated Deduction (CADE-16)}, pages 202--206,
  1999.

\bibitem[Pientka and Dunfield(2010)]{Pientka10:Beluga}
Brigitte Pientka and Jana Dunfield.
\newblock {Beluga}: A framework for programming and reasoning with deductive
  systems (system description).
\newblock In \emph{Int'l Joint Conference on Automated Reasoning (IJCAR)},
  pages 15--21, 2010.

\bibitem[Poswolsky and Sch{\"u}rmann(2009)]{Poswolsky08:Delphin}
Adam Poswolsky and Carsten Sch{\"u}rmann.
\newblock System description: {Delphin}---a functional programming language for
  deductive systems.
\newblock In \emph{Int'l Workshop on Logical Frameworks and Meta-Languages:
  Theory and Practice (LFMTP'08)}, volume 228 of \emph{Electronic Notes in
  Theoretical Computer Science}, pages 135--141, 2009.

\bibitem[Reynolds(1983)]{Reynolds83}
John~C. Reynolds.
\newblock Types, abstraction, and parametric polymorphism.
\newblock In \emph{Information Processing 83}, pages 513--523. Elsevier, 1983.
\newblock \url{http://www.cs.cmu.edu/afs/cs/user/jcr/ftp/typesabpara.pdf}.

\bibitem[Reynolds(1996)]{Reynolds96:Forsythe}
John~C. Reynolds.
\newblock Design of the programming language {Forsythe}.
\newblock Technical Report CMU-CS-96-146, Carnegie Mellon University, 1996.

\bibitem[Rondon et~al.(2008)Rondon, Kawaguchi, and Jhala]{Rondon08}
Patrick Rondon, Ming Kawaguchi, and Ranjit Jhala.
\newblock Liquid types.
\newblock In \emph{Programming Language Design and Implementation}, pages
  159--169, 2008.

\bibitem[Vazou et~al.(2013)Vazou, Rondon, and Jhala]{Vazou13}
Niki Vazou, Patrick~M. Rondon, and Ranjit Jhala.
\newblock Abstract refinement types.
\newblock In \emph{European Symp. on Programming}, pages 209--228, 2013.

\bibitem[Wright(1995)]{Wright95:value-restriction}
Andrew~K. Wright.
\newblock Simple imperative polymorphism.
\newblock \emph{Lisp and Symbolic Computation}, 8\penalty0 (4):\penalty0
  343--355, 1995.

\bibitem[Xi(1998)]{XiThesis}
Hongwei Xi.
\newblock \emph{Dependent Types in Practical Programming}.
\newblock PhD thesis, Carnegie Mellon University, 1998.

\bibitem[Xi and Pfenning(1999)]{Xi99popl}
Hongwei Xi and Frank Pfenning.
\newblock Dependent types in practical programming.
\newblock In \emph{Principles of Programming Languages}, pages 214--227, 1999.

\bibitem[Zeilberger(2009)]{ZeilbergerThesis}
Noam Zeilberger.
\newblock \emph{The Logical Basis of Evaluation Order and Pattern-Matching}.
\newblock PhD thesis, Carnegie Mellon University, 2009.
\newblock CMU-CS-09-122.

\end{thebibliography}
